\documentclass[11pt]{article} 

\title{Estimating an Extreme Bayesian Network via Scalings}

\author{Claudia Kl\"uppelberg\thanks{Center for Mathematical Sciences, Technische Universit\"at M\"unchen,  85748 Garching, Boltzmannstrasse 3, Germany, e-mail: cklu@tum.de, mario.krali@tum.de}
\and 
Mario Krali\footnotemark[1]
}
\usepackage{amssymb}
\usepackage{amsfonts}
\usepackage{amsmath}
\usepackage{amsthm}
\usepackage[round, authoryear]{natbib}
\usepackage{graphicx}
\usepackage{caption}
\usepackage{booktabs}
\usepackage{subcaption}
\usepackage[usenames, dvipsnames]{xcolor}
\usepackage{verbatim}
\usepackage{dsfont}
\usepackage{array}
\usepackage{color}
\usepackage{relsize}
\usepackage{lmodern}
\usepackage{url}
\usepackage{siunitx}
\usepackage{MnSymbol}
\usepackage{tikz}
\usetikzlibrary{positioning}
\usepackage{mathtools}
\usepackage{hyperref}
\usepackage{pgf, tikz}
\usepackage{algorithm}
\usepackage[noend]{algpseudocode}
\usepackage{tikz-3dplot}
\usepackage[titletoc]{appendix}
\usepackage[justification=centering]{caption}
\usetikzlibrary{calc,tikzmark}

\usetikzlibrary{arrows, automata}
\usepackage[english]{babel}
\numberwithin{equation}{section}
\usepackage{caption} 
\newtheorem{theorem}{Theorem}
\newtheorem{lemma}{Lemma}
\newtheorem{remark}{Remark}
\newtheorem{example}{Example}
\newtheorem{proposition}{Proposition}
\newtheorem{definition}{Definition}

\newtheorem{corollary}{Corollary}
\newtheorem{fig}{Figure}

\newcommand{\bthe}{\begin{theorem}}
\newcommand{\ethe}{\end{theorem}}

\newcommand{\ben}{\begin{enumerate}}
\newcommand{\een}{\end{enumerate}}

\newcommand{\bit}{\begin{itemize}}
\newcommand{\eit}{\end{itemize}}

\newcommand{\beq}{\begin{equation}}
\newcommand{\eeq}{\end{equation}}

\newcommand{\ble}{\begin{lemma}}
\newcommand{\ele}{\end{lemma}}

\newcommand{\bde}{\begin{definition}\rm}
\newcommand{\ede}{\halmos\end{definition}}

\newcommand{\bco}{\begin{corollary}}
\newcommand{\eco}{\end{corollary}}

\newcommand{\bpr}{\begin{proposition}}
\newcommand{\epr}{\end{proposition}}

\newcommand{\brem}{\begin{remark}\rm}
\newcommand{\erem}{\halmos\end{remark}}

\newcommand{\bproof}{\begin{proof}}
\newcommand{\eproof}{\end{proof}}

\newcommand{\bexam}{\begin{example}\rm}
\newcommand{\eexam}{\halmos\end{example}}

\newcommand{\bfi}{\begin{fig}}
\newcommand{\efi}{\end{fig}}

\newcommand{\btab}{\begin{tab}}
\newcommand{\etab}{\end{tab}}

\newcommand{\beao}{\begin{eqnarray*}}
\newcommand{\eeao}{\end{eqnarray*}\noindent}

\newcommand{\beam}{\begin{eqnarray}}
\newcommand{\eeam}{\end{eqnarray}\noindent}

\newcommand{\barr}{\begin{array}}
\newcommand{\earr}{\end{array}}

\newcommand{\bdis}{\begin{displaymath}}
\newcommand{\edis}{\end{displaymath}\noindent}

\def\P{{\mathbb P}}
\def\E{{\mathbb E}}
\def\R{{\mathbb R}}

\def\P{\mathbb{P}}

\def\cals_+{{\cals_+}}
\def\cald{{\mathcal{D}}}

\def\cals{{\mathcal{S}}}

\newcommand{\bs}{\boldsymbol}

\newcommand{\bsz}{\boldsymbol{Z}}
\newcommand{\bsx}{\boldsymbol{X}}

\newcommand{\stp}{\stackrel{P}{\rightarrow}}
\newcommand{\std}{\stackrel{d}{\rightarrow}}

\newcommand{\al}{{\alpha}}

\newcommand{\eps}{\varepsilon}

\newcommand{\DAG}{{\rm DAG}}

\newcommand{\EDM}{{\rm EDM}}

\newcommand{\an}{{\rm an}}
\newcommand{\pa}{{\rm pa}}

\newcommand{\An}{{\rm An}}
\newcommand{\Pa}{{\rm Pa}}

\let\norm\undefined 
\DeclarePairedDelimiter\norm{\lVert}{\rVert}
\newcommand{\halmos}{\quad\hfill\mbox{$\Box$}}  

\def\P{{\bf {\mathbb{P}}}}

\defcitealias{GKO}{Gissibl et al. (2018)}
\defcitealias{gkl}{Gissibl et al. (2018)}

\renewcommand{\harvardurl}[1]{(Available from: \textrm{#1})}

\allowdisplaybreaks[4]

\begin{document}


\maketitle

\begin{abstract}
Recursive max-linear vectors model causal dependence between its components by expressing each node variable as a max-linear function of its parental nodes in a directed acyclic graph and some exogenous innovation. 
Motivated by extreme value theory, innovations are assumed to have regularly varying distribution tails. 
We propose a scaling technique in order to determine a causal order of the node variables. 
All dependence parameters are then estimated from the estimated scalings. Furthermore, we prove asymptotic normality of the estimated scalings and dependence parameters based on asymptotic normality of the empirical spectral measure. 
Finally, we apply our structure learning and estimation algorithm to financial data and food dietary interview data. 

\end{abstract}

\noindent
{\em AMS 2010 Subject Classifications:}  primary:
\,\,\,60G70; 
\,\,\,62-09; 
\,\,\,62G32;  
secondary: \,\,\,65S05 

\noindent
{\em Keywords:}
causal order, directed acyclic graph, extreme value statistics, graphical model, recursive max-linear model, regular variation, structural equation model, structure learning.

\section{Introduction}

Human society is continuously faced with challenges arising from factors of both uncontrollable and/or synthetic nature. The former is manifested through events such as natural disasters, in particular climate extremes like heavy rainfall or storms, unusually high/low temperatures, or river flooding. Similarly, synthetic factors correspond to those catastrophes influenced by human intervention, for instance industry fire, terrorist attacks, or a financial market crash. 
Such events occur rarely in isolation, but are rather interconnected, and occur simultaneously across certain instances; for example, floods disseminate through a river network, or extreme losses occur across several financial sectors.
Such events make it necessary to not only understand 
dependencies between rare events, but also their causal structure. 

When modeling rare events, one faces by definition a limited amount of data. While extremes in a univariate setting are well studied, multivariate extremes still remain a focus of present research. This is partly due to the augmented dimensionality problem, which affects crucially non-parametric methods (see \citet{DHF}, Chapter~7), but also by the lack of a parametric family to characterize interdependencies (see \citet{beirlant}, Chapters~8, 9).

Recently, there has been interest in graphical models for modeling dependencies between extreme risks, which brings not only a potential complexity reduction, but also allows for modelling cause and effect in the context of extreme risk analysis. 
The model we consider in the present paper originates from \citet{gk}, where max-linear structural equation models have been proposed and investigated. The underlying graphical structure of the model is a directed acyclic graph (DAG), also called a Bayesian network. 
Identifiability and estimation of recursive max-linar models are investigated in \citetalias{gkl}.  We refer to \citet{lau} and \citet{diest} for details on graphical modeling and graph theory, respectively.

Some other methods for combining graphical modeling with extremes have been proposed recently.
In \citet{segers}, Markov trees with regularly varying node variables are investigated using the so-called tail chains.
In \citet{engelke:hitz:18}, a new approach using conditional independence relations between node variables is introduced, when considering undirected graphical models for extremes. This work is based on the assumption of a decomposable graph as well as the existence of density, which then leads to a Hammersley-Clifford type factorization of the latter into a lower dimensional setting. The method is applied to the estimation of flood in the Danube river network.
A recursive max-linear model has been fitted to data from the EURO STOXX 50 Index in \citet{einmahl2016}, where the structure of the DAG is assumed to be known.

High dimensions are a serious challenge of dependence modeling of extreme events, and as a consequence most of the applications so far have focused on a lower dimensional setting. 
{An exception is} \citet{cooley}, who present a new approach to extract the dependence structure from a regularly varying random vector. The authors propose the use of a dependence summary similar to the extreme dependence measure from \citet{lars}, which can be considered an analogue to the covariance.
Aiming at reducing the complexity, the authors propose a decomposition technique alike that of the Principal Component Decomposition for normal distributions. 
Other attempts aiming at dimension reduction of extremes involve \citet{chautru}, and \citet{JanWan}, who present clustering approaches, or \citet{HKK}, who propose a factor analysis for extremes.

In the present paper we develop a new structure learning and estimation algorithm for the recursive max-linear model in \cite{gk}.
Our approach is motivated by \cite{cooley} and applies to regularly varying node variables, which is a common assumption for extreme risk modelling. Multivariate node distributions have heavy-tailed marginals and are eponymous to those which lie in the domain of attraction of multivariate Fr\'echet distributions; see \citet{sres}, Section~5.4.2 (Proposition~5.15). 
We refer to \citet{sres,ResnickHeavy} for further details on regular variation.

Our multivariate regular variation setting is similar to that in \citetalias{GKO}, which investigates the use of the tail dependence coefficients matrix towards the recovery of a causal order and {identifiability} of the max-linear coefficient matrix. {Their method has the severe drawback that the initial nodes have to be known. For instance, in a DAG with two nodes and one edge there can only be one initial node, which can not be determined by the tail dependence coefficient as it is symmetric. This problem is to be encountered also in a DAG of larger size with several initial nodes, thus being a serious disadvantage to find a complete causal order. 
 In contrast, other than regular variation itself, our methodology is free of assumptions.
 More recently, in a heavy-tailed setting,
 \citet{gnecco:19} propose a method for identifying a causal order from the estimated conditional means of the integral transforms of pairs of nodes.}
 
 For arbitrary recursive max-linear models, a different identification and estimation method based on a generalized MLE can be found in \citetalias{gkl} and \citet{KL2017}. An extension of this method to models with observational noise has been investigated in \citet{BK}. {In these papers all innovations have to be independent and identically distributed, which is stronger than the tail assumptions imposed by regular variation.} 

We develop a new {non-parametric methodology} aimed at applying recursive max-linear models to extreme phenomena in a multivariate regular variation setting.
First, targeting the problem of recovering a causal structure as a graphical model on a DAG, we propose a new technique via the scaling parameters of multivariate marginal distributions. 
This scaling technique allows for the manipulation of the dependence structure between extremes by simple scalar multiplication. 
These manipulations then uncover specific parts of the spectral measure, which fully characterize the dependence structure of interest
to pave the way for estimating the causal dependence structure of the model.
{Second, we estimate the spectral measure empirically, where we focus on the relevant parts for the estimation of the required scaling.}
Asymptotic properties of the empirical spectral measure proven as an extension of a result of \cite{lars} lead to consistent and asymptotically normal estimates of all dependence parameters. 

The application of the proposed methodology to financial data and to food dietary data shows that the recursive max-linear model can model multivariate extremes from real-life data with the goal of inferring causality for high risks.

Our paper is structured as follows. Preliminaries on graph theoretical terminology and regular variation, including the scaling parameter, are introduced in Section~2. 
Section~3 provides relevant properties of recursive max-linear models with regularly varying node variables.
In Section~4 we show how the dependence structure of a recursive max-linear model can be identified from the scaling parameters of the model. 
Section~5 prepares for the causal inference by applying the scaling technique to find initial nodes as well as to reorder all other components into generations.
Section~\ref{regcase} deals with statistical inference of the model. We propose non-parametric estimators of the relevant scalings, which also yield estimators of the dependence parameters.
This allows us to estimate a partial order of the nodes and, in particular, a well-ordered graphical model on a DAG.
We also show the asymptotic normality of the model dependence parameters. 
Finally, Section~7 is dedicated to two applications, namely to a real world financial data set of industry portfolio returns, as well as food dietary interview data.

\section{Preliminaries}\label{s0}

\subsection{Some graphical notation}\label{s1}

Let $\mathcal{D}=(V, E)$ be a directed acyclic graph (DAG) with nodes $V=\{1,\dots,{d}\}$  and edges $E=\{(j,i): i\in V \mbox{ and } j\in\pa(i)\}$, where $\pa(i)$ are the parents of node $i$.  
Each node of $\mathcal{D}$ is associated with a random variable, and dependence between two random variables can be represented via an edge connecting the corresponding nodes; for background see \citet{lau}.

Throughout we use the following notation.
A path $p_{ji}\coloneqq[l_0=j\to l_1 \to\cdots\to l_{m}=i]$ from node $j$ to $i$ has length $|p_{ji}|=m$,
and we summarize all paths from $j$ to $i$ in the set $P_{ji}$.

For a node $i$ with parents $\pa(i)$ we set 
 $\Pa{(i)}=\pa(i)\cup \{i\} $, likewise, we denote by $\an(i)$ the ancestors of $i$ and set $\An(i)=\an(i)\cup \{i\}$. The ancestral set of some subset $C\subset V$ of nodes is denoted by $\an(C)$ or $\An(C)=\an(C)\cup C$.
We also work with the following two notions throughout.

\bde \let\qed\relax \label{generation}
(i) \, We call $i\in V$ an {\em initial node}, if $\pa(i)=\emptyset$, and denote by $V_0$ the set of all initial nodes.\\
(ii) \,	In a DAG $\mathcal{D}$, a {\em generation of nodes} is the set of all nodes that have a longest path of same length from any initial node. 
	Let $G_0=V_0$, then the $i$-th generation of nodes is defined by: 
	$$G_i=\{k\in V\setminus\underset{l<i}{\cup} G_l: \underset{p_{jk}\in P_{jk}: \hspace{1mm}j\in V_0}{\max}  {|p_{jk}|=i}\}.$$
	\ede

The following two auxiliary results provide some properties of this concept.

\ble\label{lem2.2}
In a \DAG\ $\mathcal{D}$ there is no path between two nodes of the same generation.
\ele

\bproof
	Suppose that there exists a path $p_{ij}$ in some generation $G_k\subset V$, $k\geq 1$ for nodes $i,j\in G_k$ on $\mathcal{D}$. A longest path from $V_0$ to $i$ would be of length $k$, say $p_{ti}$ for some $t\in V_0$. Extend now the same path along $p_{ij}$ to get $p_{tj}=[t\to\cdots\to i\to \cdots \to j]$.  Clearly $p_{tj}$ is longer than $p_{ti}$,  
	giving a contradiction to $j\in G_k$.
\eproof

The next result proves useful; its proof is not difficult and can be found in \cite{Krali}, Lemma 3.3.

\ble\label{lem2.3}
	Consider a \DAG\ $\mathcal{D}=(V,E)$ with $|V|= {d}$, and the set $V_0$ of initial nodes. Suppose that $\cald$ has $l$ generations. Then for $i\in\{1,\dots,l\}$, $1\le l\le d$ and $k\notin G_e$ for $e<i$, we have $k\in G_i$ if and only if for all $j\in \underset{m\ge i} {\cup} G_m$ it holds that $j\notin \emph{an}(k)$.
\ele

\bde
	A directed graph $\mathcal{D}=(V,E)$ is {\em well-ordered}, if for all $i\in V$ we have $i<j$ for all $j\in \pa(i)$. {We refer to such an order as a {\em causal order}.}
	\label{orderdef}
\ede

Note that we employ a reverse ordering than in \cite{gk}.

\subsection{Multivariate Regular Variation} \label{subsect2}

Considering max-linear models from an extreme risk perspective, we focus on node variables, which are multivariate regularly varying. 
Throughout all random objects are defined on a probability space $(\Omega,\mathcal{A},\P)$.

Multivariate regular variation can be defined in various ways, and we shall work with the following two equivalent definitions (cf. \citet{ResnickHeavy}, Theorem 6.1).

\bde[Multivariate regular variation]\label{mrv}\\
(a) \, A random vector $\boldsymbol{X}\in\mathbb{R}^{ {d}}_+$ is {\em multivariate regularly varying} if there exists a sequence $b_n\to\infty$ as $n\to\infty$ such that
\beao
n\mathbb{P}(\boldsymbol{X}/{b_n}\in \cdot)\overset{v}{\to}\nu_{\bsx}(\cdot), \hspace{5mm}n\to\infty,
\eeao
where $\overset{v}{\to}$ denotes vague convergence in $M_+(\mathbb{R}_+^{ {d}}\setminus\{\boldsymbol{0}\})$, the set of non-negative Radon measures on $\mathbb{R}_+^{d}\setminus\{\boldsymbol{0}\}$. The measure $\nu_{\bsx}$ is called {\em exponent measure} of $\bsx$.\\
(b) \,
A random vector  $\boldsymbol{X}\in\mathbb{R}^{ {d}}_+$ is {\em multivariate regularly varying} if for any choice of the norm $\|\cdot\|$ there exists a finite measure $H_{\boldsymbol{X}}$ on the positive unit sphere $\Theta_+^{{ {d}}-1}=\{\boldsymbol{\omega}\in \mathbb{R}^{ {d}}_+: \norm{\boldsymbol{\omega}}=1\}$ and a sequence $b_n\to \infty$ as $n\to\infty$ such that for the {\em polar representation} $(R,\boldsymbol{\omega})\coloneqq(\norm{\boldsymbol{X}}, \boldsymbol{X}/\norm{\boldsymbol{X}})$ of $\bsx$, 
\begin{align*}
n\mathbb{P}[({R}/{b_n},\boldsymbol{\omega})\in \cdot)]\overset{v}{\to} \nu_\alpha\times H_{\boldsymbol{X}}(\cdot), \hspace{5mm}n\to\infty,
\end{align*}
in $M_+((0,\infty]\times\Theta_+^{{ {d}}-1})$, $d\nu_\alpha(x)=\alpha x^{-\alpha-1}dx$ for some $\alpha>0$, and for Borel subsets $C\subseteq \Theta_+^{{ {d}}-1}$,
\begin{align*}
H_{\boldsymbol{X}}(C)\coloneqq\nu_{\boldsymbol{X}}\big(\{\boldsymbol{y}\in\mathbb{R}^{ {d}}_+\setminus\{\boldsymbol{0}\}: \norm{\boldsymbol{y}}\geq 1, \boldsymbol{y}/\norm{\boldsymbol{y}} \in {C}\}\big).
\end{align*}
The measure $H_{\boldsymbol{X}}$ is called the {\em spectral measure}.\\
(c) \,
If $\bsx$ satisfies the above definition, we write  $\boldsymbol{X}\in RV^{ {d}}_+(\alpha)$, and $\alpha$ is called the {\em index of regular variation}.
\ede

As explained in Theorem~6.5 of \cite{ResnickHeavy}, starting with an arbitrary vector $\bsx$ with positive components, we can always standardize all marginals to $\boldsymbol{X}\in RV^{ {d}}_+(2)$ with normalizing sequence as in (b) chosen as $b_n=\sqrt{n}$. 
This implies that all scaling information is pushed into $H_{\bsx}$.

Fix now $\alpha=2$, and the Euclidean norm $\|\cdot\|_2$, such that the positive unit sphere is $\Theta_+^{{ {d}}-1}=\{\boldsymbol{\omega}\in \mathbb{R}^{ {d}}_+: \norm{\boldsymbol{\omega}}_2=1\}$. 
The following scaling parameters have been used in \cite{cooley} for a dependence summary statistics. 

\bde\label{scaledef}
Let $\bsx\in RV^{ {d}}_+(2)$ and consider its polar representation $(R,\boldsymbol{\omega})$ as in Definition~\ref{mrv}(b) such that $\omega_i=\frac{X_i}{R}$ for $i=1,\dots, {d}$. 
For every $1\le i,j\le { {d}}$ define 
 \begin{align*}
\sigma_{ij}^2 = \sigma_{\boldsymbol{X}_{ij}}^2&\coloneqq\int_{\Theta_+^{{ {d}}-1}}\omega_i \omega_j dH_{\boldsymbol{X}}(\boldsymbol{\omega}), \quad \boldsymbol{\omega}=(\omega_1,\dots,\omega_{ {d}})\in \Theta_+^{{ {d}}-1}.
 \end{align*}
We abbreviate $\sigma_i = \sigma_{{X}_{i}}=\sigma_{\boldsymbol{X}_{ii}}$ and call it the {\em scaling} or {\em scaling parameter} of $X_i$. 
 \ede
 
The following auxiliary results are well-known and simple consequences of the definitions of regular variation. For the sake of completeness, we provide short proofs.
 
 \ble\label{usescale}
 Let $\bsx\in RV^{ {d}}_+(2)$ and choose $b_n=\sqrt{n}$. \\
 (a) \, Then $\lim_{n\to \infty} n\mathbb{P}({{X_i}}/\sqrt{n}>z)= z^{-2}\sigma_{i}^2.$\\
(b) \, Let $H_{\boldsymbol{X}}$ be the spectral measure of $\bsx$, then
$H_{\boldsymbol{X}}(\Theta_+^{{ {d}}-1})=\sum_{i=1}^{{ {d}}}\sigma_{i}^2.$
\ele

\bproof

(a) From the homogeneity of the exponent measure and its polar representation in Definition~\ref{mrv}(b) we obtain
\begin{align*}
\lim\limits_{n\to \infty} n\mathbb{P}({{X_i}}/\sqrt{n}>z)&=\nu_{\boldsymbol{X}}({\{\boldsymbol{x}\in\R_+^{ {d}} : \frac{\boldsymbol{x}}{\norm {\boldsymbol{x}}}_2\in \Theta_+^{{d}-1}, x_i > z\}})\\
&=\int_{\{\boldsymbol{\omega}\in \Theta_+^{ {d}-1}\}} \int_{\{r> z/{\omega_i}\}} 2r^{-3}dr  dH_{{\boldsymbol{X}}}(\boldsymbol{\omega})
=z^{-2}\sigma_{i}^2. 
\end{align*}
(b) We simply compute the total mass of the $ {d}$-dimensional unit simplex  $\Theta_+^{{d}-1}$:
\begin{align*}
H_{\boldsymbol{X}}(\Theta_+^{{d}-1})=\int_{\Theta_+^{{d}-1}}dH_{\boldsymbol{X}}(\boldsymbol{\omega})=\int_{\Theta_+^{{d}-1}}\sum_{i=1}^{d}\omega_i^2 dH_{\boldsymbol{X}}(\boldsymbol{\omega})=\sum_{i=1}^{{d}}\int_{\Theta_+^{{d}-1}}\omega_i^2 dH_{\boldsymbol{X}}(\boldsymbol{\omega})=\sum_{i=1}^{{d}}\sigma_{i}^2.
\end{align*}
\eproof

Immediately from Definition~\ref{scaledef} and Lemma~\ref{usescale} above we find for $i\in\{1,\dots, {d}\}$ that, if $X_i$ has scaling $\sigma_{i}$, then $cX_i$ has scaling $c \sigma_{i}$ for every $c>0$.

\begin{remark}\rm 
    (i) \, \label{specmass2} As $H_{\bsx}$ is a finite measure, it can be normalised to a probability measure by defining 
    $$\tilde{H}_{\bsx}(\cdot)\coloneqq \frac{H_{\boldsymbol{X}}(\cdot)}{H_{\boldsymbol{X}}(\Theta_+^{{d}-1})}.$$ 
    (ii) \, \label{enx} Define $\bs{\omega} := (\omega_1,\dots,\omega_{{d}}) =(X_1/R,\dots,X_{{d}}/R)$.
    Let $f\colon\Theta_+^{{d}-1}\to\mathbb{R}_+$ be a continuous function. Since $f$ is compactly supported (on $\Theta_+^{{d}-1}$), by vague convergence we have 
\begin{align}\label{empdist}
\mathbb{E}_{\tilde{H}_{\bsx}}[f(\boldsymbol{\omega})]& \coloneqq\lim\limits_{x\to \infty} \mathbb{E}[f(\boldsymbol{\omega})\mid R>x] = \int_{\Theta_+^{{d}-1}} f(\bs{\omega}) d\tilde{H}_{\bsx}(\bs{\omega}).
\end{align}
    \halmos
\end{remark}

For a simple assessment of the dependence structure of the components of a random vector, various summary measures have been introduced; see e.g. Sections~8.2.7 and 9.5.1 of \citet{beirlant}.
We note that Definition \ref{scaledef} is a non-normalized version of the \emph{extreme dependence measure} (EDM), which is a bivariate dependence measure on the positive unit sphere $\Theta_+^{{d}-1}$ and measures the limit of conditional cross-moments in the radial components of two random variables. 
The EDM has been introduced in Section~3 of \citet{edms}. 
A more refined version can be found in Propositions~3 and~4 in \citet{lars}, where also more details on the EDM are given. 

\bde[Extreme dependence measure (EDM)]\label{edmdef}
Let $\boldsymbol{X}\in RV^d_+(\alpha)$. Then for any two components $X_i,X_j$ of $\boldsymbol{X}$, setting $(\omega_i,\omega_j) := \big(\frac{X_i}R , \frac{X_j}R\big)$, the \EDM\ is given by
\begin{align*}
\EDM(X_i,X_j)=\lim\limits_{x\to\infty}\mathbb{E}\Big[\dfrac{X_i}R \dfrac{X_j}R\,\Big|\, R>x\Big]= \mathbb{E}_{\tilde{H}_{\bsx}}[\omega_i\omega_j] = \int_{\Theta_+^{d-1}} \omega_i\omega_j d\tilde{H}_{\bsx}(\bs{\omega}).
\end{align*}
\ede

\section{Recursive Max-linear Models}

Recursive max-linear models were introduced in \citet{gk} and estimated with different methods in \citetalias{gkl,GKO}; \citet{KL2017}. 
We summarize notation and results needed in the present paper.

A {\em max-linear structural equation model} $\boldsymbol{X}$ on a DAG $\mathcal{D}$ is defined as
\begin{align}\label{semequat}
X_i\coloneqq {\underset{k\in \textrm{pa}(i)}{\bigvee}}c_{ik}X_k\vee c_{ii}Z_i,\hspace{5mm} i=1,\dots,{d}, 
\end{align}
for independent random variables $Z_1,\dots,Z_{d}$, which have  support $\mathbb{R}_+$ and are atomfree, and edge weights $c_{ik}$ which are positive for all $i \in V$ and $k\in \textrm{pa}(i)\cup \{i\}$. We call $\boldsymbol{Z}=(Z_1,\dots,Z_{d})$ with these properties an {\em innovations vector}. 

Define now the operator $\times_{\max}$ between two matrices $C\in\mathbb{R}_+^{{d}\times q}$ and $D\in\mathbb{R}_+^{q\times l}$ by $$(C\times_{\max}D)_{ij}= \overset{q}{{\underset{k=1}{\bigvee}}}c_{ik}d_{kj} ,\hspace{5mm} i=1,\dots,{d}, \hspace{2mm} j=1,\dots,l. $$

From Theorem~2.2 of \cite{gk} we know that a max-linear structural equation model $\boldsymbol{X}$ from \eqref{semequat} has a solution in terms of its innovations $\boldsymbol{Z}$, which can be found by a path analysis.
For each path $p_{ji}=[j\to k_1\to\cdots \to k_l=i]$ of length $l\ge 1$ from $j$ to $i$ define the {\em path weights} $d(p_{ji})\coloneqq c_{jj}c_{k_1j}\cdots c_{ik_{l-1}}$.
Furthermore, define for $i=1,\dots,{d},$
	\begin{align*}
	a_{ij}=\underset{p_{ji}\in P_{ji}}{\bigvee}d(p_{ji}) \mbox{ for } j\in {\An}(i),\quad a_{ij}=0 \mbox{ for }  j\in V\setminus {\An}(i),\quad a_{ii}=c_{ii}.
	\end{align*}
	Then $\boldsymbol{X}$ can be written as the {\em recursive max-linear (ML) vector}:
	\begin{align}
	X_i={{\underset{j\in \An(i)}{\bigvee}}}a_{ij}Z_j,\hspace{5mm} i=1,\dots,d.
	\label{rmlmequat}
	\end{align}
The matrix $A=(a_{ij})_{i,j=1,\dots,d}$ is called the {\em ML coefficient matrix}.
Furthermore, a path $p_{ji}$ from $j$ to $i$ such that $a_{ij}=d(p_{ji})$ is called \emph{max-weighted}.

\subsection{Regular Variation of a Recursive Max-Linear Vector}

Let $\boldsymbol{Z}\in\R^{d}_+$ be an {innovations vector} and $A\in \mathbb{R}^{{d}\times {d}}_+$ a ML coefficient matrix. 
Throughout this paper we let the innovations vector $\boldsymbol{Z}\in RV_+^{d}(\alpha)$ have independent and standardized components; i.e., $n\mathds{P}(n^{-1/\alpha} Z_i>x)\to x^{-\alpha}$ as $n\to\infty$ for all $i=1,\dots,{d}$. 
According to \cite{ResnickHeavy}, p.~193f, this is equivalent to the spectral measure of $\bsz$ being {a discrete measure} on the basis vectors $\bs{e}_i$ for $i=1,\dots,{d}$ with {unit mass} on each of the $\bs{e}_i$.
Reformulating \eqref{rmlmequat}, the recursive ML random vector has representation
\begin{align}\label{ML}
\boldsymbol{X}=A\times_{\max}\boldsymbol{Z}.
\end{align}
If the innovations vector $\boldsymbol{Z} \in {RV}_+^{d}(\alpha)$, then by simple calculations given e.g. in Proposition~A.2 of \citetalias{GKO}, see also Proposition~4.1 of \cite{Krali},
$\bsx\in {RV}_+^{d}(\alpha)$ 
with discrete spectral measure
	\begin{equation}
	H_{\boldsymbol{X}}(\cdot) = \sum_{k=1}^{d}\norm{a_k}^\alpha \delta_{\big\{ \frac{{a_k}}{\norm{a_k}} \big\}}(\cdot),
	\label{discretspecteq}
	\end{equation}	
where 	$a_k=(a_{1k},\dots,a_{{d}k})^\top$ is the $k$-th column of $A$. Obviously, the entries of $A$ are the dependence parameters of $\bsx$.

Using the representation of $H_{\boldsymbol{X}}$ in \eqref{discretspecteq} together with Remark~\ref{specmass2} (ii) we obtain the following lemma.

\ble\label{le:fcont}
	Let $\boldsymbol{X}$ be a recursive ML vector  as in \eqref{ML}. 
	Let $f:\Theta_+^{{d}-1}\to\R_+$ be continuous, and define the radial components of $\bsx$ as
	$\bs{\omega}=(\omega_1,\dots,\omega_{d})=(X_1/R,\dots,X_{d}/R)$. Then
\begin{equation*}
\E_{\tilde{H}_{\bsx}} [f(\bs{\omega})]
= \frac{1}{\sum_{i=1}^{d}\norm{a_i}^\alpha}\sum_{k=1}^{d} {\norm{a_k}^\alpha} f\Big(\frac{{a_{1k}}}{\norm{a_k}},\dots,\frac{{a_{dk}}}{\norm{a_k}}\Big).
\end{equation*}
\ele

For $\al=2$ and the Euclidean norm, the scalings of the recursive ML random vector \eqref{ML} can be expressed by the matrix $A$ as follows.

\bpr\label{scalee}
Let ${\boldsymbol{X}}=A\times_{\max}\boldsymbol{Z}$, where $\boldsymbol{Z}\in RV^{d}_+(2)$ is an innovations vector, and $A\in \mathbb{R}^{{d}\times{d}}_+$. Then $\sigma_{ij}^2=(AA^T)_{ij}$. Moreover, every component $X_k$ of $\bsx$ has scaling $\sigma^2_k=\sum_{i=1}^{d} a_{ki}^2$ 
	for $k=1,\dots,{d}$.
\epr

\bproof
From Definition~\ref{scaledef} and \eqref{discretspecteq}  we find for $i\neq j$
	\begin{align*}
	\sigma_{ij}^2 =\int_{\Theta_+^{{d}-1}}\omega_i \omega_j dH_{\boldsymbol{X}}(\boldsymbol{\omega})
	= \sum_{k=1}^{d}\norm{a_k}_2^2 \frac{a_{ik}}{\norm{a_k}_2}\frac{a_{jk}}{\norm{a_k}_2}=\sum_{k=1}^{d}a_{ik}a_{jk} 
	=(AA^T)_{ij}.
	\end{align*}
	The calculation of squared scalings $\sigma^2_i$ is analogous.

\eproof

Finally, we consider the standardized recursive ML random vector $\bsx$ from \eqref{ML} by standardizing the ML coefficient matrix.

\bde[Standardized ML coefficient matrix]\\
Define
\begin{align}\bar{A}=(\bar{a}_{ij})_{{d}\times{d}}\coloneqq \bigg(\frac{a_{ij}^\alpha}{\sum_{k\in \textrm{An(i)}}a_{ik}^\alpha}\bigg)_{{d}\times{d}}^{1/\alpha}=\bigg(\frac{a_{ij}^\alpha}{\sum_{k=1}^{d}a_{ik}^\alpha}\bigg)_{{d}\times{d}}^{1/\alpha}.\label{abar}
\end{align}
Then $\bar{A}$ is referred to as \emph{standardized ML coefficient matrix}.
\ede 

\brem\label{rem34}
Since the innovations vector $\boldsymbol{Z}$ is standardized, all scaling information is in $\bar{A}$:
Proposition~\ref{scalee} entails that the recursive ML vector $\boldsymbol{X}=\bar{A}\times_{\max}\boldsymbol{Z}$ has components with squared scalings $\sigma^2_{i}=(\bar{A}\bar{A}^T)_{ii}=1$ for $i=1,\dots,{d}$.
\erem

The following result has been proven in Lemma 2.1 of \citetalias{GKO}. 

\ble\label{ineq}
Assume that the DAG corresponding to the recursive ML vector $\boldsymbol{X}$ is well-ordered.
Then  
$$\bar{a}_{jj}>\bar{a}_{ij}\quad\mbox{for all}\quad i\in V\quad\mbox{and}\quad j\in \an(i).$$
\ele

We summarize all model assumptions used throughout the rest of the paper.\\

\noindent
{\bf {Assumptions:}}
\begin{enumerate}
\item[(A1)]
The innovations vector $\boldsymbol{Z}\in RV^{d}_+(2)$ has independent and standardized components.
\item[(A2)]
We work with the Euclidean norm $\|\cdot\|_2$.
\item[(A3)]
The ML coefficient matrix $A$ is standardized as in eq. \eqref{abar}, such that the components of $\bsx$ are standardized. 
\end{enumerate}

 \section{Identification of the ML Coefficient Matrix From Scalings}\label{s4}
 
 In this section we consider a recursive ML vector $\boldsymbol{X}=A\times_{\max}\boldsymbol{Z}$ such that (A1)-(A3) are satisfied. 
 We show how to {identify} $A$ from $\bsx$, when $\boldsymbol{X}$ is a recursive ML vector on a well-ordered DAG; i.e., 
\begin{align}
 \boldsymbol{X}=\begin{bmatrix}
 X_1 \\X_2\\ \vdots \\ X_{d} \end{bmatrix}=A\times_{\max}\boldsymbol{Z}=\begin{bmatrix}
 a_{11} & a_{12}& \hdots  & a_{1{d}} \\
 0 &   a_{22} &\hdots   & a_{2{d}}\\
 \vdots& \vdots & \ddots & \vdots\\
 0&   0 &\hdots    &a_{{d}{d}}
 \end{bmatrix}\times_{\max}\boldsymbol{Z}.
 \label{dageq1}
\end{align}
We identify $A$ from the squared scalings of maxima over combinations of components of $\boldsymbol{X}$.
For a set $\bs{h}\subseteq\{1,\dots,{d}\}$ we define 
\begin{align}\label{maxp} 
M_{\bs{h}}\coloneqq\max_{i\in \bs{h}} (X_i).
\end{align}

We first compute the relevant squared scalings.

\ble\label{scalcoll}
The random variable $M_{\bs{h}}$ is again max-linear, in particular $M_{\bs{h}}\in RV^1_+(2)$ with squared scalings as follows:\\
(a) Let $\bs{h}\subseteq\{1,\dots,{d}\}$, then 
	\beam\label{scaleq}
	\sigma_{M_{\bs{h}}}^2 &=&\sum_{k=1}^{d} (\bigvee_{i\in \bs{h}} a_{ik}^2).
	\eeam
(b) If $\bs{h} =\{1,\dots,{d}\}$, then $\sigma_{M_{\bs{h}}}^2 = \sum_{k=1}^{d}  a_{kk}^2.$
\ele

\begin{proof} 
	\emph{(a)} Starting with $X_i=\underset{j=1,\dots,{d}}{\bigvee}a_{ij}Z_j$ for $i=1,\dots,{d}$, we calculate:
	\begin{equation*}
	M_{\bs{h}}=\max_{i\in \bs{h}}(X_i)
	=\underset{i\in \bs{h}}{\bigvee}\hspace{1mm}  \underset{j=1,\dots,{d}}{\bigvee}a_{ij}Z_j
	=\underset{j=1,\dots,{d}}{\bigvee} (\underset{i\in \bs{h}}{\bigvee}a_{ij})Z_j.
	\end{equation*}
	By eq. \eqref{discretspecteq} $M_{\bs{h}}$ is regularly varying. In order to compute the squared scaling $\sigma_{M_{\bs{h}}}^2$, we use the same arguments as in Proposition~\ref{scalee} which yields \eqref{scaleq}.\\
	\emph{(b)} This follows directly from \emph{(a)} in combination with Lemma \ref{ineq}.
\end{proof}

We now illustrate the identification of the ML coefficient matrix by the following example. 

\bexam\label{examp2}
	Let $\boldsymbol{X}$ be a recursive ML vector on a well-ordered DAG satisfying (A1)-(A3)  
	such that
	\begin{align*}
	\boldsymbol{X}=A\times_{\max}\boldsymbol{Z}=
	\begin{bmatrix}
	a_{11} & a_{12}  & a_{13} \\
	0 &   a_{22}     &a_{23}\\
	0&   0     & a_{33}
	\end{bmatrix}\times_{\max}\boldsymbol{Z}.
	\end{align*}
Note that by standardization every row must have norm 1.

	We first compute the diagonal entries. By standardization, $a^2_{33}=\sigma_{3}^2=1$.
	By Lemma \ref{ineq}  we know that $a_{ii}>a_{ki}$ for $k<i$.
	Let $M_{ij}$ for $1\le i,j\le 3$ and $M_{123}$ be defined as in \eqref{maxp}.
	From Lemma \ref{scalcoll} we obtain
	\beao
	\barr{rcrcr}
	\sigma_{M_{123}}^2&=& a_{11}^2+a_{22}^2+a_{33}^2 &=& a_{11}^2+a_{22}^2+1,\\
	\sigma_{M_{23}}^2&=& a^2_{21}\vee a^2_{31} +a_{22}^2+a_{33}^2& =& 0 \, + a_{22}^2+1.
	\earr
	\eeao
	From this we first find $a_{11}^2=\sigma_{M_{123}}^2-\sigma_{M_{23}}^2$. 
	Similarly $a_{22}^2=\sigma_{M_{23}}^2-\sigma_{3}^2$.
	
	The next step is to find the remaining entries in the first row of $A$, namely $a_{12}$ and $a_{13}$. \\
	Proceeding with $a_{12}$ we find from Lemma \ref{scalcoll} for $M_{13}$ first
	$\sigma_{M_{13}}^2=a_{11}^2+a_{12}^2+a_{33}^2=a_{11}^2+a_{12}^2+\sigma_{3}^2,$
	which yields 
	$$a_{12}^2=\sigma_{M_{13}}^2-\sigma_{3}^2-a_{11}^2=\sigma_{M_{13}}^2+\sigma_{M_{23}}^2-\sigma_{M_{123}}^2-\sigma_{3}^2.$$ 
	Finally, we find $a_{13},a_{23}$, since the rows of $A$ have norm 1.
\eexam

We now proceed by proving the correctness of the above recursion, which gives rise to Algorithm~\ref{recalg2} below.

\begin{proposition}\label{estalg2}
Let $\boldsymbol{X}$ be a recursive ML vector on a well-ordered DAG satisfying  (A1)-(A3).
Then the following recursion yields the standardized ML coefficient matrix $A$:
	\begin{align}
	a_{{d}{d}}^2  =  \sigma_{d}^2=1 \, \mbox{ and } \, a_{ii}^2 \, =\, \sigma_{M_{i,\dots,{d}}}^2-\sigma_{M_{i+1,\dots,{d}}}^2 \quad	 & i=1,\dots,{d}-1,  \label{recformula1}\\
	a_{ij}^2 = \sigma_{M_{i,j+1,j+2,\dots,{d}}}^2-\sigma_{M_{j+1,j+2,\dots,{d}}}^2-\sum_{k=i}^{j-1}a_{ik}^2 \quad \quad
	&
	i=1,\dots,{d}-2,  j =i+1,\dots,{d}-1.
	\label{recformula2}\\
	a_{i{d}}^2 = \sigma_{i}^2-\sum_{k=i}^{{d}-1}a_{ik}^2 \, = \, 1-\sum_{k=i}^{{d}-1}a_{ik}^2 \quad\quad\quad\quad\quad & i=1,\dots,{d}-1. \label{recformula3}
	\end{align}
\end{proposition}

\begin{proof}
	We first show \eqref{recformula1}.
	From Lemma \ref{scalcoll} we find for $i=1,\dots,{d}-1$:
	\begin{align*}
	\sigma_{M_{i,i+1,\dots,{d}}}^2&=\sum_{k=i}^{d}a_{kk}^2\quad\mbox{ and }\quad
	\sigma_{M_{i+1,i+2,\dots,{d}}}^2 \, =\sum_{k=i+1}^{d}a_{kk}^2,
	\end{align*}
	which implies that $a_{ii}^2=\sigma_{M_{i,\dots,{d}}}^2-\sigma_{M_{i+1,\dots,{d}}}^2$.
	For $i=p$, by standardization of $A$ we have $a_{{d}{d}}^2={\sigma_{{d}}^2}=1$.\\
	In order to prove \eqref{recformula2} we compute first:
	\begin{align*}
	\sigma_{M_{i,j+1\dots,{d}}}^2&=\sum_{k=i}^{j}a_{ik}^2+\sum_{k=j+1}^{d}a_{kk}^2\quad\mbox{ and }\quad
	\sigma_{M_{j+1,\dots,{d}}}^2\, =\sum_{k=j+1}^{d}a_{kk}^2,
	\end{align*}
	which implies
	\begin{align}
	\sigma_{M_{i,j+1\dots,{d}}}^2-	\sigma_{M_{j+1,\dots,{d}}}^2&=\sum_{k=i}^{j}a_{ik}^2.
	\label{ijentry}
	\end{align}
	Fix now $i\in\{1,\dots,{d}-1\}$. We proceed by induction over $j$. We start with the initial index $j=i+1.$ By \eqref{recformula1} we know all $a_{ii}$ for $i=1,\dots,{d}$, 
       and by (\ref{ijentry}),
	$$ \sigma_{M_{i,i+2\dots,{d}}}^2-	\sigma_{M_{i+2,\dots,{d}}}^2-a_{ii}^2=\sum_{k=i}^{i+1}a_{ik}^2-a_{ii}^2=a_{i,i+1}^2.$$
	By the induction hypothesis, suppose that we have found $a_{ij}$ for all $j\in \{i+1,...,l-1\}$, where $i+2<l<{d}$.  Let $j=l$.
	Then, it is straightforward to see that
	$$\sigma_{M_{i,j+1\dots,{d}}}^2-	\sigma_{M_{j+1,\dots,{d}}}^2-\sum_{k=i}^{j-1}a_{ik}^2=\sum_{k=i}^{j}a_{ik}^2-\sum_{k=i}^{j-1}a_{ik}^2=a_{ij}^2.$$
	Equation \eqref{recformula3} follows from the fact that $A$ is standardized, hence, all rows have norm 1 (by Remark~\ref{rem34}, $\sigma_{i}^2=\sum_{k=i}^{d}a_{ik}^2=1.$)
\end{proof}

The Algorithm corresponding to Proposition~\ref{estalg2} reads as follows.

\begin{algorithm}[H]
	\caption{Computation of the ML Coefficient Matrix A}\label{recalg2}
	\begin{algorithmic}[1]
		\Procedure{}{}
		\State \textbf{Set} $A = (0)_{{d}\times{d}}$
		\State \textbf{for} $i = 1,\dots,{d}-1$ \textbf{do}
		\State \hspace{5mm}\textbf{Compute} $\sigma_{M_{i,i+1,\dots,{d}}}^2; \sigma_{M_{i+1,\dots,{d}}}^2$
		\State \hspace{5mm}\textbf{Set} $a_{ii}^2=\sigma_{M_{i,i+1,\dots,{d}}}^2-\sigma_{M_{i+1,\dots,{d}}}^2$
		\State \hspace{10.2mm} $a_{{d}{d}}^2=\sigma_{{d}}^2$
		\State \hspace{5mm}\textbf{if} $i\in\{1,\dots,{d}-2\}$ \textbf{do}
		\State \hspace{10mm}\textbf{for} $j=i+1,\dots,{d}-1$ \textbf{do}
		\State \hspace{15mm}\textbf{Compute} $\sigma_{M_{i,j+1,j+2,\dots,{d}}}^2; \sigma_{M_{j+1,j+2,\dots,{d}}}^2$
		\State \hspace{15mm}\textbf{Set} $a_{ij}^2=\sigma_{M_{i,j+1,j+2,\dots,{d}}}^2-\sigma_{M_{j+1,j+2,\dots,{d}}}^2-\sum_{k=i}^{j-1}a_{ik}^2$
		\State \hspace{10mm}\textbf{end for}
			%
		
		\State\hspace{10mm}\textbf{Set} $a_{i{d}}^2=\sigma_{i}^2-\sum_{k=i}^{{d}-1}a_{ik}^2$
		\State \hspace{5mm}\textbf{end if}
		\State\hspace{5mm}\textbf{Set} $a_{{d}-1,{d}}^2=\sigma_{{{d}-1}}^2-a_{{d}-1,{d}-1}^2$
		\State \textbf{end for}. 
		\EndProcedure
	\end{algorithmic}
\end{algorithm}

In Proposition~\ref{estalg2} we have shown that we can compute the diagonal entries of $A$ from the squared scalings $\sigma_{M_{1,2,\dots,{d}}}^2,\sigma_{M_{2,3,\dots,{d}}}^2,\dots,$ $ \sigma_{M_{{d}-1,{d}}}^2,\sigma_{d}^2$ by a recursion algorithm.
Furthermore, we have identified the non-diagonal entries of the $i$-th row of $A$ from
$$(\sigma_{M_{i,i+1,\dots,{d}}}^2,\sigma_{M_{i,i+2,\dots,{d}}}^2,\dots, \sigma_{M_{i,{d}}}^2 ,\sigma_{i}^2),\quad i=1,\dots,{d}.$$
We summarize all these quantities into one column vector ${S}_{M}\in \mathbb{R}^{{d}({d}+1)/2}_+$; i.e.,
\begin{small}
    \begin{align}\label{sm}
	{S}_{M}\coloneqq( {\sigma}_{M_{1,2,\dots,{d}}}^2, {\sigma}_{M_{1,3,\dots,{d}}}^2,\dots,  {\sigma}_{M_{1,{d}}}^2 , {\sigma}_{1}^2,  {\sigma}_{M_{2,3,\dots,{d}}}^2, {\sigma}_{M_{2,4,\dots,{d}}}^2,\dots,  {\sigma}_{M_{2,{d}}}^2, {\sigma}_{2}^2,\dots , {\sigma}_{M_{{d}-1,{d}}}^2, {\sigma}_{{d}-1}^2,  {\sigma}_{d}^2)^\top.
\end{align}
\end{small}
Consider the row-wise vectorized version of the squared entries of the upper triangular matrix $A$, where we use $A^2$ for the {matrix with squared entries of $A$} and its vectorized version 
\begin{align}\label{vecA}
{A^2}\coloneqq(a_{11}^2,\dots,a_{1{d}}^2,a_{22}^2,\dots,a_{2{d}}^2,\dots..,a_{{d}-1,{d}-1}^2, a_{{d}-1,{d}}^2,a_{{d}{d}}^2)^\top.
\end{align}
Note that both vectors ${A^2}$ and ${S}_{M}$ show a similar structure, built from row vectors with ${d},{d}-1,\ldots,1$ components, respectively; so both have ${d}({d}+1)/2$ components.
By means of Proposition~\ref{estalg2} we show that ${A^2}$ can be written as a linear transformation of $S_M$.

\bthe\label{consT}
	Let $S_M$ and ${A^2}$ be as in (\ref{sm}) and (\ref{vecA}), respectively. 	Then
\beam \label{Alinear} 
{A^2} &=& T \, S_M,
\eeam
where $T\coloneqq (t_{uv})_{k\times k}\in \mathbb{R}^{k\times k}$ for $k={d}({d}+1)/2$ has non-zero entries in the rows corresponding to the non-zero components $a^2_{ij}$ in the vector (\ref{vecA}) given by
\begin{enumerate}
\item[]
$a^2_{ii}:\quad t_{l_{ii}, l_{ii}}=1$, $t_{l_{ii}, l_{i+1,i+1}}=-1 $ for  $i=1,\dots,{d}-1$;
\item[]
$a^2_{{d}{d}}:\quad t_{l_{ii}, l_{ii}}=1$ for $i={d}$; 
\item[]
$a^2_{ij}:\quad t_{l_{ij},l_{ij}}=1,  t_{l_{ij},l_{j+1,j+1}}=-1, t_{l_{ij},l_{i,j-1}}=-1, t_{l_{ij},l_{jj}}=1$ for $i<j\leq {d}-1$; 
\item[]
$a^2_{i{d}}:\quad t_{l_{i{d}},l_{i{d}}}=1, t_{l_{i{d}},l_{i,{d}-1}}=-1,  t_{l_{i{d}},l_{{d}{d}}}=1$ for $i=1,\dots,{d}-1$,
\end{enumerate}
where {$l_{ij}=(j-{d})+\sum_{k=0}^{i-1}({d}-k)$ for $i=1,...,d$ and $j\geq i$}. All other entries of $T$ are equal to zero.

\ethe

\bproof
(i) \, From \eqref{recformula1} we know that 
\begin{align*}
	 a_{ii}^2 \, =\, \sigma_{M_{i,\dots,{d}}}^2-\sigma_{M_{i+1,\dots,{d}}}^2, \,\,	  i=1,\dots,{d}-1, \, \, \mbox{ and } \, a_{{d}{d}}^2  =  \sigma_{{d}}^2=1
\end{align*}
(ii) \, Starting from \eqref{recformula2} we show by induction that for $i=1,\dots,{d}-2$ and $j=i+1,\dots,{d}-1$,
	\begin{align}\label{AS1}
	a_{ij}^2=(\sigma_{M_{i,j+1,j+2,\dots,{d}}}^2-\sigma_{M_{j+1,j+2,\dots,{d}}}^2)-(\sigma_{M_{i,j,\dots,{d}}}^2-\sigma_{M_{j,\dots,{d}}}^2).
	\end{align}
	For $j=i+1$ we clearly have that $a_{i,i+1}^2=(\sigma_{M_{i,i+2,\dots,{d}}}^2-\sigma_{M_{i+2,\dots,{d}}}^2)-(\sigma_{M_{i,i+1,\dots,{d}}}^2-\sigma_{M_{i+1,\dots,{d}}}^2).$
	Suppose now that this holds for all $i<j<{d}-2$. We show now that it holds for $j+1$. More specifically,
	\begin{align*}
	a_{i,j+1}^2&=(\sigma_{M_{i,j+2,j+3,\dots,{d}}}^2-\sigma_{M_{j+2,j+3,\dots,{d}}}^2)-\sum_{k=i}^{j}a_{ik}^2\\
	&=(\sigma_{M_{i,j+2,j+3,\dots,{d}}}^2-\sigma_{M_{j+2,j+3,\dots,{d}}}^2)-\sum_{k=i}^{j}[(\sigma_{M_{i,k+1,k+2,\dots,{d}}}^2-\sigma_{M_{k+1,k+2,\dots,{d}}}^2)-(\sigma_{M_{i,k,\dots,{d}}}^2-\sigma_{M_{k,\dots,{d}}}^2)]\\
	&=(\sigma_{M_{i,j+2,j+3,\dots,{d}}}^2-\sigma_{M_{j+2,j+3,\dots,{d}}}^2)-(\sigma_{M_{i,j+1,\dots,{d}}}^2-\sigma_{M_{j+1,\dots,{d}}}^2),
	\end{align*}
	where the last equality is due to the telescoping sum after noting that $\sigma_{M_{i,i,\dots,{d}}}^2=\sigma_{M_{i,\dots,{d}}}^2$.
	
	(iii) Similar to (ii), for $a_{i{d}}$ with $i<{d}$,
	\begin{align}\label{ASp}
	a_{i{d}}^2=\sigma_{i}^2-(\sigma_{M_{i,{d}}}^2-\sigma_{{d}}^2),
	\end{align}
	while for $i={d}$ we obtain again $a_{{d}{d}}^2=\sigma_{{d}}^2$.
	
	(iv) \, The results in (i)-(iii) show already the linearity between the vectors ${A^2}$ and $S_M$. 
	It remains to construct the matrix $T=(t_{uv})_{k\times k}$ for $k={d}({d}+1)/2$ such that ${A^2}=T S_M$. 
	We start by renumbering the vector components in ${A^2}$ and replacing the double indices $ij$ for $i=1,\dots,{d}$ and $j\ge i$ by 
	\begin{align}\label{renumber}
	l_{ij}=(j-{d})+\sum_{k=0}^{i-1}({d}-k),\hspace{4mm}{j\geq i}.
	\end{align}

Then the vector in \eqref{vecA} becomes $(a^2_1,a^2_2,\dots,a^2_{{d}({d}+1)/2})$.
Moreover, \eqref{renumber} maps $ii$ into $l_{ii}$ and $l_{i,i+k}=l_{ii}+k$ for $i=1,\dots,{d}$ and $1\le i+k\leq {d}$.

Also notice that for all $i=1,\dots,{d}$, by the structure of $S_M$, its $l_{ij}$-th component is $S_{l_{ij}}=\sigma_{M_{i,j+1,\dots,{d}}}^2$ for $i\le j<{d}$, and $S_{l_{i{d}}}=\sigma_{i}^2$.

	(v) \, We construct now $T$, where by (i)-(iii) $T$ contains many zeros, and we focus on the non-zero entries.
	
	Since $a^2_{ii}$ becomes $a^2_{l_{ii}}$ for $i=1,\dots,{d}$
	and, by the structure of $S_M$, the $l_{11},\dots,l_{{d}{d}}$-th components of $S_M$ are $\sigma_{M_{1,2,\dots,{d}}}^2, \sigma_{M_{2,3,\dots,{d}}}^2,$ $\dots,\sigma_{M_{{d}-1,{d}}}^2,\sigma_{M_{d}}^2=\sigma_{d}^2$, respectively,
	in each $l_{ii}$-th row of $T$ there must be a 1 on the diagonal; i.e. $t_{l_{ii},l_{ii}}=1$. Furthermore, $t_{l_{ii}, l_{i+1,i+1}}=-1$, and all other entries in this row are 0.
	
	Similarly, we find the other non-zero entries by representation \eqref{AS1} and \eqref{ASp}.
\eproof

\bexam
We illustrate the linear transformation (\ref{Alinear}) for ${d}=4$, which clarifies the structure also for higher dimensions. For a recursive ML vector with 4 nodes, by (\ref{AS1}) and (\ref{ASp}) the identity ${A}^2=T S_M$ becomes
$$
\begin{bmatrix}
a_{11}^2\\
a_{12}^2\\
a_{13}^2\\
a_{14}^2\\
\hline
a_{22}^2\\
a_{23}^2\\
a_{24}^2\\
\hline
a_{33}^2\\
a_{34}^2\\
\hline
a_{44}^2\\
\end{bmatrix}
\, = \,
\left[
\begin{array}{rrrr|rrr|rr|r}
1  & 0 & 0 & 0 \, & -1 & 0 & 0 \, & 0 & 0 \, & 0\\
-1 & 1 & 0 & 0 \, & 1 & 0 & 0 \, & -1 & 0 \, & 0\\
0  & -1 & 1 & 0 \, & 0 & 0 & 0 \, & 1 & 0 \, & -1\\
0  & 0 & -1 & 1 \, & 0 & 0 & 0 \, & 0 & 0 \, & 1\\
\hline
0 & 0 & 0 & 0 \, & 1 & 0 & 0 \, & -1 & 0 \, & 0\\
0 & 0 & 0 & 0 \, & -1 & 1 & 0 \, & 1 & 0 \, & -1\\
0  & 0 & 0 & 0 \, & 0 & -1 & 1 \, & 0 & 0 \, & 1\\
\hline
0 & 0 & 0 & 0 \, & 0 & 0 & 0 \, & 1 & 0 \, & -1\\
0  & 0 & 0 & 0 \, & 0 & 0 & 0 \, & -1 & 1 \, & 1\\
\hline
0  & 0 & 0 & 0 \, & 0 & 0 & 0 \, & 0 & 0 \, & 1\\
\end{array}
\right]
\, \times \, 
\begin{bmatrix}
\sigma^2_{M_{1,2,3,4}}\\
\sigma^2_{M_{1,3,4}}\\
\sigma^2_{M_{1,4}}\\
\sigma^2_{1}\\
\hline
\sigma^2_{M_{2,3,4}}\\
\sigma^2_{M_{2,4}}\\
\sigma^2_{2}\\
\hline
\sigma^2_{M_{3,4}}\\
\sigma^2_{3}\\
\hline
\sigma^2_{4}\\
\end{bmatrix}.
$$
\eexam

\section{Reordering the Vector Components}\label{s5}

In Section~4 we have assumed that the DAG underlying the recursive ML vector $\boldsymbol{X}$ is well-ordered. In a real life situation this will rarely be the case, and the components of $\bsx$ have to be reordered. 
In this section we use again the scalings for finding a causal order of the components of $\boldsymbol{X}$. This is achieved by first identifying the initial nodes, which can be ordered arbitrarily within all initial nodes. The same applies for every following generation: within each generation the order is arbitrary. All such obtained partial orders correspond to equivalent well-ordered DAGs and we construct one representative DAG by the method {as follows.}

We start with an auxiliary result which ensures that the recursive ML vector $\bsx=A \times_{\max}\boldsymbol{Z}$ is invariant with respect to column permutations of the ML coefficient matrix $A$.

\ble	\label{colpermutation}
	Let $\boldsymbol{X}\in\mathbb{R}^{d}_+$ be a recursive ML vector with ML coefficient matrix $A\in\R_+^{{d}\times{d}}$ and innovations vector $\bsz\in\R_+^{d}$. 
	Let $\pi$ be a permutation of the columns of $A$. 
	Then $\boldsymbol{X}^{\pi}=\boldsymbol{X}$. 
\ele

\begin{proof}
Denote by $\pi: \{1,\dots,{d}\}\to \{1,\dots,{d}\}$ an arbitrary permutation of the columns of $A$,	and notice that an arbitrary component $i\in \{1,\dots,{d}\}$ of $\boldsymbol{X}$ is given by 
	$$X_i=\underset{k=1,\dots,{d}}{\bigvee}a_{ik}Z_k =\underset{\pi(k)=1,\dots,{d}}{\bigvee}a_{i\pi(k)}Z_{\pi(k)}=\underset{k'=1,\dots,{d}}{\bigvee}a_{ik'}Z_{k'}=:X_i^{\pi}$$ and, 
	therefore, $\boldsymbol{X}^{\pi}=\boldsymbol{X}$.
\end{proof}

{Since by Lemma~\ref{colpermutation} the distribution of $\bs X$ is invariant with respect to column permutations,} we can assume that an arbitrarily ordered recursive ML vector $\boldsymbol{X^*}=(X_{1^*},\dots,X_{{d}^*})=A^*\times_{\max}\boldsymbol{Z}$, needs only row permutations in $A$, denoted by $\nu: (1^*,\dots,{d}^*) \to (1,\dots,{d})$, to become well-ordered:
\[
\nu: A^*=\begin{bmatrix}
a_{1^*1}&\hdots &a_{1^*{d}}\\
\vdots& \ddots & \vdots\\
a_{{d}^*1}&\hdots &a_{{d}^*{d}}
\end{bmatrix}
\to
A_\nu = A 
=\begin{bmatrix}
a_{11}&\hdots &a_{1{d}}\\
\vdots& \ddots & \vdots\\
0&\hdots &a_{{d}{d}}
\end{bmatrix}.
\]
We refer to entries of the matrix $A^*$ as $a_{i^*k}$ and to entries from the row-permuted matrix $A_\nu$ as $a_{ik}\coloneqq a_{\nu(i^*)k}$, corresponding to a reordered vector $\boldsymbol{X}_{\nu}=\bsx$ {in distribution} on a well-ordered DAG.

\subsection{Reordering the Vector Components: Finding the Initial Nodes}  \label{inod}

In order to find an initial node, we fix one node which we want to investigate and extend the notation from \eqref{maxp} to maxima over a ${d}$-tuple of partly scaled random variables:
for $a>0$ we define for $m \in\{1,\dots,{d}\}$,
\begin{equation}\label{Mmam}
M_{{-m},am} \coloneqq\max(X_1,\dots ,X_{m-1}, aX_m, X_{m+1},\dots,X_{d}).
\end{equation}
By Lemma~\ref{scalcoll}, also $M_{{-m},am}\in RV_+(2)$. 
The following theorem provides necessary and sufficient conditions for the identification of initial nodes. 

\begin{theorem}
	\label{theo2}
	Let $\boldsymbol{X^*}=(X_{1^*},\dots ,X_{d^*})$ be an arbitrarily ordered recursive ML vector with ML coefficient matrix $A^*$ satisfying  (A1)-(A3). 
	Then the following holds.\\
		(a) If ${m^*\in\{1^*,\dots ,{d}^*\}}$ is an initial node of the recursive ML vector $\boldsymbol{X}_{\nu}$ in a well ordered DAG, then for all scalars $a>1$ it holds that 
		\begin{align}\label{lem2cr}
		\sigma_{M_{-{m^*}, am^*}}^2-\sigma_{M_{1,\dots,{d}}}^2=a^2-1.
		\end{align}
		(b) If there exists a scalar $a>1$, such that for ${m^*\in\{1^*,\dots ,{d}^*\}}$ eq. (\ref{lem2cr}) holds, then $m^*$  is an initial node of the recursive ML vector $\boldsymbol{X}_{\nu}$ in a well ordered DAG.
\end{theorem}

\begin{proof}
	$(a)$ 
	Let $X_{m^*}$ be the component of $\boldsymbol{X^*}$ such that $m^*$ is an initial node.  
	W.l.o.g. we may set 
	 $X_{\nu(m^*)}=X_{{d}}$. By the representation (\ref{dageq1}), and given the standardized scalings, we know that $a_{\nu(m^*),1}=\dots=a_{\nu(m^*),{d}-1}=0$, and $a_{\nu(m^*),{d}}=1$. 
    By Lemma \ref{scalcoll}(b), for some $a>1$, we compute 
	\[\sigma_{M_{1,\dots,{d}}}^2=a_{11}^2+\dots +a_{{d}-1, {d}-1}^2+1
	\quad\mbox{and}\quad
	\sigma_{M_{-{m^*}, a m^*}}^2=a_{11}^2+\dots +a_{{d}-1, {d}-1}^2+a^2. \]
	Taking the difference yields \eqref{lem2cr}.\\
	$(b)$ \, We prove this by contradiction. 
	Let $\nu$ be a row permutation that transforms $\bsx^*$ into a recursive ML vector $\bsx_{\nu}$ on a well-ordered DAG, and suppose that for some non-initial node, say $k^*\in\{1^*,\dots,{d}^*\}$, there exists some $a>1$ such that 
	\begin{align}\label{contra1}
	    \sigma_{M_{-k^*, a k^*}}^2=\sigma_{M_{1,\dots,{d}}}^2 +a^2-1.
	\end{align} 
	Given that a recursive ML vector can have more than one initial node, w.l.o.g. assume that there are $0<l<{d}$ initial nodes. Since $k^*$ is not an initial node, we know by Definition \ref{orderdef} that $\nu({k^*})\leq {d}-l$. 
	Furthermore, as ${k^*}$ must have an ancestor, there exists some node $j>\nu({k^*})$ for $j\in\{1,...,{d}\}\setminus{\{\nu(k^*)\}}$, such that $j\in \textrm{an}(\nu({k^*}))$ and, hence, $a_{\nu({k^*})j}>0$. 
	Since $\nu$ is a row permutation, which permutes $\boldsymbol{X^*}$ into a recursive ML vector on a well ordered DAG, we know that there exists some $j^*\in\{1^*,...,{d}^*\}\setminus{\{k^*\}}$, such that $\nu(j^*)=j$. By Lemma \ref{ineq}, and since $j>\nu({k^*})$, and $j\in \textrm{an}(\nu({k^*}))$, it follows that $a_{jj}>a_{\nu(k^*)j}> 0$.
	This implies that
	\begin{align*}
	\sigma_{M_{-\nu(k^*), a\nu(k^*)}}^2 &= \sum^{j=\nu(k^*)-1}_{j=1} a_{jj}^2 + a^2a_{\nu(k^*), \nu(k^*)}^2 +\sum_{j=\nu(k^*)+1}^{d} (a^2a_{\nu(k^*)j}^2)\vee a_{jj}^2\\
	\sigma_{M_{1,\dots,{d}}}^2 &= \sum^{j=\nu(k^*)-1}_{j=1} a_{jj}^2 +a_{\nu(k^*), \nu(k^*)}^2+\dots+a_{{d}{d}}^2. 
	\end{align*}
	The difference gives 
	\begin{align}\label{nod1dif}
	\sigma_{M_{-\nu(k^*), a\nu(k^*)}}^2-\sigma_{M_{1,\dots,{d}}}^2 & = (a^2-1)a_{\nu(k^*), \nu(k^*)}^2+\sum_{j=\nu(k^*)+1}^{d} ((a^2a_{\nu(k^*)j}^2)\vee a_{jj}^2-a_{jj}^2).
	\end{align} 
	Next, for the summands in the sum on the right-hand side, following Lemma \ref{ineq}, we obtain 	
	\begin{align}\label{bound}
	(a^2a_{\nu(k^*)j}^2)\vee a_{jj}^2- a_{jj}^2= \begin{cases}
	a^2a_{\nu(k^*)j}^2-a_{jj}^2<(a^2-1)a_{\nu(k^*)j}^2, \hspace{2mm}$  	if$ \hspace{1mm} {a^2a_{\nu(k^*)j}^2>a_{jj}^2}\\
	0, \hspace{6cm}  $else$.
	\end{cases}
	\end{align}
	This implies
	$$\sum_{j=\nu(k^*)+1}^{d} (a^2a_{\nu(k^*)j}^2\vee a_{jj}^2-a_{jj}^2)<\sum_{j=\nu(k^*)+1}^{d}(a^2-1)a_{\nu(k^*)j}^2,$$
	which, when combined with eq. \eqref{nod1dif}, yields the inequality
	\begin{align*}
	\sigma_{M_{-\nu(k^*), a\nu(k^*)}}^2-\sigma_{M_{1,\dots,{d}}}^2
	< (a^2-1) \, \Big(a_{\nu(k^*), \nu(k^*)}^2+\sum_{j=\nu(k^*)+1}^{d}a_{\nu(k^*)j}^2\Big) \, = \, a^2-1.
	\end{align*}
However, this is a contradiction to eq. \eqref{contra1}.
\end{proof}

\subsection{Reordering the Vector Components: Finding the Descendants} \label{desnod}

Once we have identified the initial nodes of the recursive ML vector $\boldsymbol{X}$, we provide a necessary and sufficient criterion for identifying the causal order of the descendants. 

We proceed iteratively by identifying every new generation in the DAG.
Suppose we have found all nodes which belong to a certain number of generations, and that there are $h\le {d}-1$ such nodes which we have ordered as ${d},{d}-1,\dots,{d}-h+1$.  
Let $X_{\nu^{-1}({d})},\dots,X_{\nu^{-1}({d}-h+1)}$ be the corresponding components in the arbitrarily ordered recursive ML vector $\boldsymbol{X^*}.$

The next logical step is to investigate, whether ${m^*}\in\{1^*,\dots,{d}^*\}\setminus \{\nu^{-1}({d}),\dots,$ $\nu^{-1}({d}-h+1)\}$, belongs to the next generation of nodes in the causal order. Define 
$\boldsymbol{h}\coloneqq \{\nu^{-1}({d}),\dots,\nu^{-1}({d}-h+1)\}$ and let $\boldsymbol{h}^c$ contain all other components.
Then we take the maximum over a ${d}$-tuple of partly scaled random variables: for $a>0$ define
\begin{equation}\label{partscaled}
M_{\boldsymbol{h}_a, m^*_a, \{\boldsymbol{h}\cup \{m^*\}\}^c}\coloneqq\max(aX_{\nu^{-1}({d})},\dots,aX_{\nu^{-1}({d}-h+1)}, aX_{m^*}, \max_{k^*\notin \{\bs{h}\cup\{m^*\}\}}X_{k^*}).
\end{equation} 

\begin{theorem} \label{mainorder}
	Let $\boldsymbol{X^*}=(X_{1^*},\dots,X_{{d}^*})$ be an arbitrarily ordered recursive ML vector with ML coefficient matrix $A^*$ satisfying  (A1)-(A3). Let $\boldsymbol{h}=\{\nu^{-1}({d}),\dots,\nu^{-1}({d}-h+1)\}$ be the first $h$ nodes of the recursive ML vector $\boldsymbol{X}_\nu$ of a well-ordered DAG,
	which have already been ordered. Then the following holds:\\
		(a) If ${m^*}\notin\boldsymbol{h}$ has no ancestors in $\boldsymbol{h}^c$, then for all scalars $a>1$ it holds that
	\begin{align}
	\sigma_{\boldsymbol{M}_{\boldsymbol{h}_a,m^*_a,\{\boldsymbol{h}\cup \{m^*\}\}^c}}^2-\sigma_{M_{1,...,d}}^2=(a^2-1)\sigma_{M_{\boldsymbol{h}, m^*}}^2.\label{maindif}
	\end{align}
		(b) If there exists a scalar $a>1$ such that (\ref{maindif}) holds, then we identify ${m^*}\notin\boldsymbol{h}$ as the $(h+1)$-th node.
\end{theorem}

\begin{proof}
	$(a)$ W.l.o.g. let ${m^*}$ be a node such that $\nu(m^*)={d}-h$. 
	Consider the squared scaling of ${M}_{\boldsymbol{h}_a,m^*_a,\{\boldsymbol{h}\cup \{m^*\}\}^c}$ as in \eqref{partscaled}, and $M_{1,\dots,{d}}$. 
	By representation \eqref{dageq1} and Lemma \ref{scalcoll}, and following similar steps as in the proof of Theorem \ref{theo2}(i), we find
	\begin{align*}
	\sigma_{M_{\boldsymbol{h}_a,m^*_a,\{\boldsymbol{h}\cup \{m^*\}\}^c}}^2-\sigma_{M_{1,\dots,{d}}}^2&=a^2(\sum_{i={d}-h+1}^{d}a_{ii}^2+a_{\nu(m^*),\nu(m^*)}^2)+\sum_{j=1}^{{d}-h-1}a_{j j}^2-(\sum_{i=1}^{d}a_{ii}^2)\\
	&=(a^2-1)(\sum_{i={d}-h}^{d}a_{ii}^2)=(a^2-1)\sigma_{M_{\boldsymbol{h},m^*}}^2.
	\end{align*}
	Therefore, (\ref{maindif}) is satisfied. \\
	$(b)$\, Suppose now that for 
	$X_{m^*}$ (${m^*}\notin\boldsymbol{h}$) there exists some $a>1$ such that 
	\begin{align}\label{useeq}
	\sigma_{\boldsymbol{M}_{\boldsymbol{h}_a,m^*_a,\{\boldsymbol{h}\cup \{m^*\}\}^c}}^2-\sigma_{M_{1,\dots,{d}}}^2=(a^2-1)\sigma_{M_{\boldsymbol{h}, m^*}}^2.
	\end{align}
	We have the following system of equalities:
	\begin{align*}	
	\sigma_{M_{1,\dots,{d}}}^2 &=\sum_{i=1}^{d}a_{ii}^2;\\ 
	\sigma_{M_{\boldsymbol{h},m^*}}^2 &=\sum_{i={d}-h+1}^{d}a_{ii}^2+\sum_{j=\nu(m^*)}^{{d}-h} a_{\nu(m^*)j}^2;\\
	\sigma_{\boldsymbol{M}_{\boldsymbol{h}_a,m^*_a,\{\boldsymbol{h}\cup \{m^*\}\}^c}}^2 &=a^2(\sum_{i={d}-h+1}^{d}a_{ii}^2+a_{\nu(m^*), \nu(m^*)}^2)+\sum_{j=\nu(m^*)+1}^{{d}-h} ((a^2a_{\nu(m^*)j}^2)\vee a_{jj}^2) +\sum_{j=1}^{\nu(m^*)-1}a_{jj}^2; \\ 
	\sigma_{\boldsymbol{M}_{\boldsymbol{h}_a,m^*_a,\{\boldsymbol{h}\cup \{m^*\}\}^c}}^2-\sigma_{M_{1,\dots,{d}}}^2 &=(a^2-1)(\sum_{i={d}-h+1}^{d}a_{ii}^2+a_{\nu(m^*),\nu(m^*)}^2) +\sum_{j=\nu(m^*)+1}^{{d}-h} ((a^2a_{\nu(m^*)j}^2)\vee a_{jj}^2-a_{jj}^2).
	\end{align*}	
	The summands in the last summation are non-negative. \\
	For $\nu(m^*)={d}-h$ we can take $m^*$ as the $(h+1)$-th node. 
	Then the difference is equal to $(a^2-1)\sigma_{M_{\boldsymbol{h},m^*}}^2$. 
	Similarly, if $a_{\nu(m^*)j}=0$ for $\nu(m^*)+1\leq j\leq {d}-h$, then this implies that 
	${\nu^{-1}(j)}\notin \textrm{an}(m^*)$ and, thus, that $m^*$ can be chosen as the $(h+1)$-th node. \\
	Suppose now that $\nu(m^*)< {d}-h$, and $a_{\nu(m^*)j}> 0$ for some $\nu(m^*)+1\leq j\leq {d}-h$. Then, by Lemma~\ref{ineq} and since $a>1$, a bound similar to that in (\ref{bound}) gives
	$$\sum_{j=\nu(m^*)+1}^{{d}-h} (a^2a_{\nu(m^*)j}^2\vee a_{jj}^2-a_{jj}^2)<\sum_{j=\nu(m^*)+1}^{{d}-h}(a^2-1)a_{\nu(m^*)j}^2$$
	which contradicts (\ref{useeq}), since 
	\begin{align*}
	\sigma_{\boldsymbol{M}_{\boldsymbol{h}_a,m^*_a,\{\boldsymbol{h}\cup \{m^*\}\}^c}}^2-\sigma_{M_{1,...,d}}^2<(a^2-1)\sigma_{M_{\boldsymbol{h}, m^*}}^2.
	\end{align*}
	Therefore we have that either $\nu(m^*)={d}-h$, or $a_{\nu(m^*)j}=0$ for $\nu(m^*)+1\leq j\leq {d}-h$. In both cases $m^*$ can be chosen as the $(h+1)$-th node.
\end{proof}

One of the consequences of the proof of Theorem~\ref{mainorder} provides a criterion, when two or more components of $\boldsymbol{X}$ are neither descendants nor ancestors of one another in a DAG. 

\bco\label{exzero}
	Let $\boldsymbol{X}$ be as in Theorem \ref{mainorder} and suppose that we have found the first $h$ nodes. If
	$\sigma_{\boldsymbol{M}_{\boldsymbol{h}_a,m^*_a,\{\boldsymbol{h}\cup \{m^*\}\}^c}}^2-\sigma_{M_{1,\dots,{d}}}^2=(a^2-1)\sigma_{M_{\boldsymbol{h}, m^*}}^2$ for $m^*\in\{i^*,j^*\}\cap\boldsymbol{h}^c$ and some $a>1$, then $a_{i^*j^*}=a_{j^*i^*}=0.$ 
\eco

	As another direct consequence of Theorem~\ref{mainorder} we obtain the following corollary.
	
\bco\label{algord}
	Let $\boldsymbol{X}$ be as in Theorem \ref{mainorder} and suppose that we have found the first $h$ nodes.
	Let $X_{i^*},X_{j^*}$ be such that $X_{i^*}$ is the $(h+1)$-th node, while $X_{j^*}$ belongs to a different generation. 
	Define for $m^*\in\{1^*,\dots,{d}^*\}$ and $a>1$
	$$\Delta_{m^*}\coloneqq \sigma_{\boldsymbol{M}_{\boldsymbol{h}_a,m^*_a,\{\boldsymbol{h}\cup \{m^*\}\}^c}}^2-\sigma_{M_{1,\dots,{d}}}^2-(a^2-1)\sigma_{M_{\boldsymbol{h}, m^*}}^2.$$  Then $\Delta_{i^*}>\Delta_{j^*}$.
\eco

\bexam[The reordering algorithms for a 10 nodes model]\label{10ex}\\
We assess the performance of the {reordering procedure} 
based on Theorems~\ref{theo2} and~\ref{mainorder}. 
We assume that the recursive ML vectors have standard Fr\'echet(2) components, which
allows us to estimate the scalings by the standard MLE given for $M_{\bs{h}}$ by (see \citet{Krali}, Section 5.4 for details)
$$\hat\sigma^2_{M_{\bs{h}}}= \Big(\frac1n\sum_{\ell=1}^n \frac1{m_{\bs{h}_\ell}^{2}}\Big)^{-1},$$
where $m_{\bs{h}_\ell}$ is the empirical maximum $\vee_{i\in \boldsymbol{h}}X_{\ell i}$ for the $\ell$-th observation.
We consider the $10$-nodes DAG depicted in Figure \ref{fig:M1}.
All diagonal entries of the edge weight matrix $C_{10}\in\smash{\mathbb{R}^{10\times 10}_+}$ (see (\ref{semequat})) are set to $c_{ii}=1$ and
the squares of the non-diagonal non-zero entries of the upper-triangular matrix are drawn from a discrete uniform distribution over $\smash{\{2/1,2/2,\dots,2/8\}}$. 
We have chosen edge weights $c_{ij}$ larger than 1, equal to 1, and smaller than 1 to capture amplifying and downsizing risk in the network.
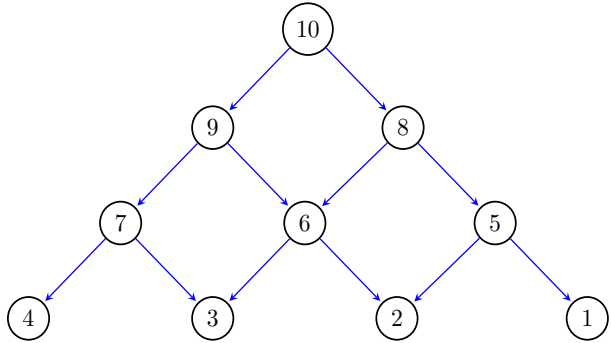
\begin{figure}[H]
	\centering
	\resizebox{8cm}{4.5cm}{\begin{tikzpicture}[
		> = stealth,
		shorten > = 1pt, 
		auto,
		node distance = 2cm, 
		semithick 
		]
		\tikzstyle{every state}=[
		draw = black,
		thick,
		fill = white,
		minimum size = 4mm,scale=1
		]
		\node[state] (10) {$10$};
		\node[state] (8) [below right=1cm and 1cm of 10] {$8$};
		\node[state] (9) [below left=1cm and 1cm of 10] {$9$};
		\node[state] (7) [below left=1cm and 1cm  of 9] {$7$};
		\node[state] (6) [below right=1cm and 1cm  of 9] {$6$};
		\node[state] (5) [below right=1cm and 1cm  of 8] {$5$};
		\node[state] (4) [below left=1cm and 1cm  of 7] {$4$};
		\node[state] (3) [below right=1cm and 1cm  of 7] {$3$};
		\node[state] (2) [below right=1cm and 1cm  of 6] {$2$};
		\node[state] (1) [below right=1cm and 1cm  of 5] {$1$};
		
		\path[->][blue] (10) edge node {} (9);
		\path[->][blue] (10) edge node {} (8);
		\path[->][blue] (9) edge node {} (7);
		\path[->][blue] (9) edge node {} (6);
		\path[->][blue] (6) edge node {} (2);
		\path[->][blue] (6) edge node {} (3);
		\path[->][blue] (7) edge node {} (4);
		\path[->][blue] (7) edge node {} (3);
		\path[->][blue] (8) edge node {} (5);
		\path[->][blue] (8) edge node {} (6);
		\path[->][blue] (5) edge node {} (1);
		\path[->][blue] (5) edge node {} (2);
		\end{tikzpicture}}
	\caption{DAG with $10$ nodes.} \label{fig:M1}
\end{figure}

To perform the reordering we turn both Theorem~\ref{theo2} and Theorem~\ref{mainorder} into Algorithm \ref{alg8} and Algorithm \ref{alg9}, respectively.
For every node, say ${i^*}$, Algorithm \ref{alg8} checks the criterion in Theorem~\ref{theo2} (see Line 4 below). 
The difference between $\sigma_{M_{-\nu(i^*), a\nu(i^*)}}^2-\sigma_{M_{1,\dots,{d}}}^2$ and $a^2-1$ is entered into the ${i^*}$-th- component of the vector $\Delta$. 
Finally, the vector $N$ is filled with non-zero components only for those nodes, which satisfy the criterion of Theorem~\ref{theo2}.
Algorithm \ref{alg9} follows a similar logic. 

When estimating the scalings, then the $\hat{\Delta}_{i^*}$-s in Algorithms~\ref{alg8} and \ref{alg9} are a.s. different from 0.
Hence, both algorithms are adapted to allow for small bounds to both quantities, which have to be chosen appropriately.  We are interested in checking if the final order identifies the nodes in accordance with their respective generations. To this end, we choose a small bound $\epsilon_3$, which enables Algorithm~\ref{alg9} to return more than one node per iteration step; namely the generations. Based on simulation experience, we chose $a=\sqrt{2}$ and $\eps_1=0.1, \eps_2=0.05, \eps_3=0.1$.
All simulations and data analysis are done using \textsf{R}, \citet{R.ref}.

\begin{algorithm}[H]
	\caption{Identifying the initial nodes in $\boldsymbol{X}_{10}$}\label{alg8}
	\begin{algorithmic}[1]
		\Procedure{}{}
		\State \textbf{Set} $\hat{\Delta} = (0)_{1\times {d}}; N = (0)_{1\times {d}}; a>1; \eps_1>0; \eps_2>0$
		\State \textbf{for} $i^* = {d},\dots,1$ \textbf{do}
		\State \hspace{5mm} 
		$\hat{\Delta}_{i^*}=\hat\sigma_{M_{-i^*, ai^*}}^2-\hat\sigma_{M_{1,...,d}}^2 -a^2+1$
		\State\hspace{5mm}  \textbf{if} $(\hat{\Delta}_{i^*}\geq 0$ \textbf{and}  $\hat{\Delta}_{i^*}\leq\eps_1)$, \textbf{or}  $(\hat{\Delta}_{i^*}\leq 0$ \textbf{and}  $\hat{\Delta}_{i^*}\geq-\eps_2)$,  \textbf{then}
		\State\hspace{10mm}$N_{i^*}=i^*$
		%
		\State\hspace{5mm} \textbf{else} $N_{i^*}=0$
		\State\textbf{end for}. 
		\EndProcedure
	\end{algorithmic}
\end{algorithm}

\begin{algorithm}[ht]
	\caption{Identifying the ($h+1$)-th {node} of $\boldsymbol{X}_{10}$}\label{alg9}
	\begin{algorithmic}[1]
		\Procedure{}{}
		\State \textbf{Set} $\hat{\Delta} = (0)_{1\times{d}}; N = (0)_{1\times {d}}; a>1; \epsilon_3>0; {\boldsymbol{h}=\{\nu^{-1}({d}),\dots,\nu^{-1}({d}-h+1)\}}$
		\State \textbf{for} $i^*\notin\boldsymbol{h}$ \textbf{do}
		\State \hspace{5mm} 
		$\hat{\Delta}_{i^*}=\hat\sigma_{\boldsymbol{M}_{\boldsymbol{h}_a,i^*_a,\{\boldsymbol{h}\cup \{i^*\}\}^c}}^2-\hat\sigma_{M_{1,...,d}}^2-(a^2-1)\hat\sigma_{M_{\boldsymbol{h}, i^*}}^2$
		%
		\State\hspace{5mm}  \textbf{if} $|\hat{\Delta}_{i^*}|\leq \epsilon_3$,  \textbf{then}
		\State\hspace{10mm}$N_{i^*}=i^*$
		\State\hspace{5mm} \textbf{else} $N_{i^*}=0$
		\State\textbf{end for}. 
		\EndProcedure
	\end{algorithmic}
\end{algorithm}
As the initial order is irrelevant, we set w.l.o.g. $\boldsymbol{X}^*_{10}=(X_1,\dots,X_{10})$.
We perform 100 simulation runs for each of the sample sizes $n\in\{2000,3000,5000,$ $10000\}$. 

The reorderings are obtained by first applying Algorithm \ref{alg8} and then Algorithm \ref{alg9}. 
For some simulations it has happened that the bounds in Line 5 in Algorithms~\ref{alg8} and~\ref{alg9} are not satisfied by any of the components. 
We indicate this in Table~\ref{tab1}, where the column ``Valid Runs'' corresponds to the number of simulation runs (out of $100$) for which the conditions of the bounds in the algorithms are satisfied each time the ``if" loop is entered. 
\begin{table}[htp]
	\centering
	{\begin{tabular}{rrrrr}
			\toprule
			Sample size & Simulation Runs &Valid Runs& Correctly Reordered & Success Ratio \\
			\midrule
			2000 &100& 81& 65 & 80.24\% \\
			3000 &100& 91 & 80 & 87.91\% \\
			5000 &100& 96 & 94 & 97.92\% \\
			10000 & 100& 99 & 99 & 100\% \\
			\bottomrule
		\end{tabular}}%
		\caption{\small{Results of the reordering procedure for the 10-node DAG in Figure \ref{fig:M1}. Correct orders are only those with initial node $V_0=\{10\}$, and generations $G_1=\{8,9\}, G_2=\{5,6,7\},$ $ G_3=\{1,2,3,4\}$ with arbitrary order within each generation.}}
		\vspace*{-3mm}\label{tab1}
	\end{table}%

The column ``Correctly Reordered'' gives the number of runs for which Algorithms~\ref{alg8} and~\ref{alg9} return the generations $V_0=\{10\}, G_1=\{8,9\}, G_2=\{5,6,7\}, G_3=\{1,2,3,4\}$.
The column ``Success Ratio'' presents the ratio of ``Correctly Reordered" over ``Valid Runs".

\eexam

\section{Statistical Theory for Regularly Varying Innovations}\label{regcase}

The discrete spectral measure of the ML model poses serious challenges towards the objective of estimation. \citet{einmahl2016m}, \citet{einmahl2012m}, and \citet{einmahl2016} develop estimation procedures for the stable tail dependence function. In both, \cite{einmahl2012m} and \cite{einmahl2016}, the methods are also applied to models with discrete spectral measure, whose dependence parameters can be obtained from the stable tail dependence function.
\citet{JanWan} provide a new way for estimating the atoms of the spectral measure on the unit sphere by using a clustering approach. 
For our purposes we resort to the empirical spectral measure.

Let $\boldsymbol{X}_1,\dots,\boldsymbol{X}_n$ be an i.i.d. sample  of $\boldsymbol{X}\in RV^{d}_+(2)$. 
For $\ell=1,\dots,n$ define
\begin{align}\label{pol.est}
R_\ell\coloneqq\norm{\boldsymbol{X}_\ell}_2\quad\mbox{and}\quad {\boldsymbol{\omega}_\ell=({\omega}_{\ell1},...,{\omega}_{\ell {d}})\coloneqq\frac{\boldsymbol{X}_\ell}{R_\ell}},
\end{align}
 to obtain their respective polar representation $\{(R_\ell, {\boldsymbol{\omega}}_\ell): \ell=1,\dots,n\}$. 
A consistent estimator for the standardized spectral measure $\tilde{H}_{\bsx}$ as in Remark~\ref{specmass2} is based on the limit relation \eqref{empdist} and given 
e.g. in eq. (9.32) in Chapter~9.2 of \cite{ResnickHeavy} as ($\lfloor s \rfloor$ denotes the integer part of $s\in\R$ )
$$\tilde{H}_{\bsx,\lfloor n/k\rfloor} (\cdot) = \frac{\sum_{\ell=1}^{n}\mathds{1}{\{(R_\ell/b_{\lfloor\frac{n}{k}\rfloor},{\boldsymbol{\omega}}_\ell) \in [1,\infty] \times \cdot\}}}{\sum_{\ell=1}^{n}\mathds{1}{\{R_\ell/b_{\lfloor\frac{n}{k}\rfloor}\geq1\}}}\std \tilde{H}_{\bsx} (\cdot),$$
as $n\to\infty$, $k\to\infty$, $k/n\to 0$.
Since $R^{(k)}/b_{\lfloor\frac{n}{k}\rfloor}\stp 1$ (cf. below eq. (9.32) of \cite{ResnickHeavy}), where $R^{(k)}$ is the $k$-th largest among $R_1,...,R_n$.
Hence, setting ${b}_{\lfloor\frac{n}{k}\rfloor}=R^{(k)}$, the denominator becomes $\sum_{\ell=1}^{n}\mathds{1} {\{R_\ell\geq R^{(k)}\}}=k$ and the above estimator reads 
$$\tilde{H}_{\bsx,\lfloor n/k\rfloor} (\cdot) = \frac{1}{k}\sum_{\ell=1}^{n} \mathds{1}{\{R_\ell\geq R^{(k)}, \boldsymbol{\omega}_\ell\in\cdot \}}.
$$
Then an estimator for $\E_{\tilde{H}_{\bsx}} [f(\bs{\omega})]$ is given by
\begin{align}\label{specemp}
\hat{\mathbb{E}}_{\tilde{H}_{\boldsymbol{X}}}[f({\boldsymbol{\omega}})]=\frac{1}{k}\sum_{\ell=1}^{n}f({\boldsymbol{\omega}}_\ell)\mathds{1}{\{R_\ell \geq R^{(k)} \}}.
\end{align}
Our goal is to estimate the squared scalings ${\sigma}_{M_{\boldsymbol{h}}}^2$ of $M_{\bs{h}}$ for $\bs{h}\subseteq\{1,\dots,{d}\}$ as in \eqref{maxp}. 
To this end we choose $f:\Theta_+^{{d}-1}\to\R_+$ defined  for $\ell=1,\dots,n$ via
$f(\boldsymbol{\omega}_\ell)={d}(\underset{k \in \boldsymbol{h}}{\bigvee}\omega_{\ell k}^2)$, which is a continuous function. 

\bde\label{estscaling2}[Non-parametric scaling estimators]
Let $\boldsymbol{X}\in RV^{d}_+(2)$ and  $\boldsymbol{X}_1,\dots,\boldsymbol{X}_n$ be an i.i.d. sample of $\boldsymbol{X}$ with respective polar representations $(R_\ell, {\boldsymbol{\omega}}_\ell)$ for  $\ell=1,\dots,n$ as in \eqref{pol.est}. For $1\le k\le n$ estimate {$\sigma^2_{i}$ for $i=1,\dots,d$ and ${\sigma}_{M_{\boldsymbol{h}}}^2$ for $\bs{h}\subseteq\{1,\dots,d\}$} by \eqref{specemp} as 
\begin{align}\label{estscal}
\hat\sigma^2_{i} = \frac{{d}}{k}\sum_{\ell=1}^{n}{\omega}_{\ell i}^2\mathds{1}{\{R_\ell \geq R^{(k)}\}}
\quad\mbox{and}\quad 
\hat{\sigma}_{M_{\boldsymbol{h}}}^2 
=\frac{{d}}{k}\sum_{\ell=1}^{n}\underset{j\in \boldsymbol{h}}{\bigvee}{\omega}_{\ell j}^2\mathds{1}{\{R_\ell \geq R^{(k)}\}}.
\end{align}
\ede

Finally, we estimate $A$ by the linear transformation as given in Theorem \ref{consT}.

\subsection{Asymptotic Normality of the Non-paramatric Estimators}

 \citet{lars} have proven in their Theorem~1 a CLT for the EDM as in Definition~\eqref{edmdef} based on the fact that 
the function $f:\Theta_+\to \mathbb{R}_+$ given by $f(\boldsymbol{\omega})=\omega_1\omega_2$ is continuous.
This result can be generalised to any continuous functions $f:\Theta_+^{{d}-1}\to \mathbb{R}$ and any dimension ${d}$.

The following CLT holds for every $\boldsymbol{X}\in RV^{d}_+(2)$.
Again we use the polar representation \eqref{pol.est}.
Let the radial component $R$ of $\bsx$ have distribution function $F$. 

\begin{theorem}[Central Limit Theorem]\label{clt}
	Let $\boldsymbol{X}\in RV^{d}_+(2)$
	and $\boldsymbol{X}_1,\dots,\boldsymbol{X}_n$ be i.i.d. copies of $\boldsymbol{X}$. Choose $k$ such that $k=o(n)$ and $k\to\infty$ as $n\to\infty$. Let $f\colon \Theta_+^{{d}-1}\to {\mathbb{R}_+}$ be a continuous function. 
Assume that 
	\begin{align}
	\lim\limits_{n\to \infty}\sqrt{k}\bigg(\frac{n}{k}\mathbb{E}[f({{\boldsymbol{\omega}_1}})\mathds{1}{\{R_1\geq b_{\lfloor\frac{n}{k}\rfloor}t^{-1/ \alpha}\}}] 
	- \mathbb{E}_{\tilde{H}_{\bsx}}[f({\boldsymbol{\omega}}_1)]\frac{n}{k}\bar{F}(b_{\lfloor\frac{n}{k}\rfloor} t^{-1/\alpha}) \bigg)=0
	\label{assump}
	\end{align}
	holds locally uniformly for $t\in[0, \infty)$, and 
	$\emph{Var}_{\tilde{H}_{\bsx}}(f({\boldsymbol{\omega}}))>0$.
	Then
	\begin{align*}
	\sqrt{k}(\hat{\mathbb{E}}_{\tilde{H}_{\boldsymbol{X}}}[f({\boldsymbol{\omega}})]-\mathds{E}_{\tilde{H}_{\boldsymbol{X}}}[f({\boldsymbol{\omega}})])\overset{\mathcal{D}}{\to} N(0, \sigma^2), \hspace{5mm} n\to\infty.
	\end{align*}
\end{theorem}

\begin{proof}
	The proof follows closely Theorem 1 of \cite{lars}, which only covers the case of $f(\omega_1,\omega_2)=\omega_1\omega_2$. Going through this proof line by line we find that it applies to every continuous function $f:\Theta_+^{{d}-1}\to \R_+$, where the asymptotic variance $\sigma^2$ has to be adapted to the chosen function $f$.
\end{proof}

Assumption (\ref{assump}) requires the dependence between ${\boldsymbol{\omega}}_\ell$ and $R_\ell$ for  $R_\ell>b_{\lfloor\frac{n}{k}\rfloor}$ to decay sufficiently fast as $n\to \infty$. 

\subsection{Asymptotic Normality of the Scalings of Maxima}

From Theorem~\ref{consT} we know that we can identify $A$ by a known linear transformation $T$ from the vector ${S}_{M}\in \mathbb{R}^{{d}({d}+1)/2}_+$ of squared scalings defined in \eqref{sm}.
Each of these squared scalings we estimate by $\hat\sigma^2_{i}$ and $\hat{\sigma}_{M_{\boldsymbol{h}}}^2$ as in eq. \eqref{estscal} and
denote the resulting estimation vector by 
\begin{align}\label{smest}
	\hat{S}_{M}\coloneqq( {\hat\sigma}_{M_{1,2,\dots,{d}}}^2, {\hat\sigma}_{M_{1,3,\dots,{d}}}^2,\dots,  {\hat\sigma}_{M_{1,{d}}}^2 , {\hat\sigma}_{1}^2,  {\hat\sigma}_{M_{2,3,\dots,{d}}}^2, {\hat\sigma}_{M_{2,4,\dots,{d}}}^2,\dots,  {\hat\sigma}_{M_{2,{d}}}^2, {\hat\sigma}_{{2}}^2,\dots , {\hat\sigma}_{M_{{d}-1,{d}}}^2, {\hat\sigma}_{{{d}-1}}^2,  {\hat\sigma}_{{d}}^2).
\end{align}
We show asymptotic multivariate normality of the vector $\hat S_M$. 
To this end we use the Cram\'er-Wold device and a properly chosen continuous function $f$ on $\Theta_+^{{d}-1}$ to which we then apply Theorem \ref{clt}.

\begin{theorem}
	\label{scalclt}
	Let $\boldsymbol{X}=A\times_{\max}\boldsymbol{Z}$ be a recursive ML vector satisfying  (A1)-(A3), and let $\boldsymbol{X}_1,\dots,\boldsymbol{X}_n$ be i.i.d. copies of $\boldsymbol{X}$. Choose $k$ such that $k=o(n)$ and $k\to\infty$ as $n\to\infty$.
	Furthermore, assume that (\ref{assump}) holds {for $f(\boldsymbol{\omega})={d}(\underset{i\in\boldsymbol{h}_i}{\bigvee}\omega_i^2)$} and that 
	 $\emph{Var}_{\tilde{H}_{\bsx}}({d}(\underset{i\in\boldsymbol{h}_i}{\bigvee}\omega_i^2))>0$  for all $\boldsymbol{h}_i\subset \{1,\dots,{d}\}$ such that $S_{M_{\boldsymbol{h}_i}}$ is a component of $S_M$.
	Then
	$$\sqrt{k}(\hat{S}_M-S_M)\overset{\mathcal{D}}{\to} N(0, W_M),\quad n\to\infty,$$ 
	where the entries of the covariance matrix $W_M$ are given by the right-hand sides of the following two limits.
	The diagonal entries for $\boldsymbol{h}\subset \{1,\dots,{d}\}$ satisfy $$\lim_{n\to\infty} \, k\emph{Var}(\hat{S}_{M_{\boldsymbol{h}}}) = {d}^2 \emph{Var}_{\tilde{H}_{\bsx}}(\underset{i\in \boldsymbol{h}}{\bigvee}{\omega}_i^2)),$$ 
	and the non-diagonal entries for two different sets $\boldsymbol{h}_i\neq \boldsymbol{h}_j\subset \{1,\dots,{d}\}$ are given by $$\lim_{n\to\infty} \, k\,\emph{Cov}(\hat{S}_{M_{\boldsymbol{h}_i}},\hat{S}_{M_{\boldsymbol{h}_j}}) = \frac{{d}^2}{2}( \emph{Var}_{\tilde{H}_{\bsx}}(\underset{k\in \boldsymbol{h}_i}{\bigvee}{  \omega}_k^2+\underset{l\in \boldsymbol{h}_j}{\bigvee}{  \omega}_l^2)-\emph{Var}_{\tilde{H}_{\bsx}}(\underset{k\in \boldsymbol{h}_i}{\bigvee}{  \omega}_k^2)-\emph{Var}_{\tilde{H}_{\bsx}}(\underset{l\in \boldsymbol{h}_j}{\bigvee}{  \omega}_l^2)).$$
	Moreover, the covariance matrix $W_M$ is singular. 
\end{theorem}

\begin{proof}
	By the Cram\'er-Wold device we have to show asymptotic normality of $\sqrt{k}\boldsymbol{t}^T(\hat{S}_M-{S}_M)$ for every $\boldsymbol{t}\in \mathbb{R}^{{d}({d}+1)/2}$. 
	We first index the entries of $S_M$ according to $\boldsymbol{h}_i$ for $i\in\{1,\dots,{d}({d}+1)/2\}$. 
	Then we re-write
	\begin{align*}
	\boldsymbol{t}^T{S}_M&=\sum_{i=1}^{{d}({d}+1)/2}t_iS_{M_{\boldsymbol{h}_i}}=\sum_{i=1}^{{d}({d}+1)/2}{d}t_i\mathbb{E}_{\tilde{H}_{\bsx}}[\underset{j\in \boldsymbol{h}_i}{\bigvee}{  \omega}_j^2]
	=\mathbb{E}_{\tilde{H}_{\bsx}}[{d}\sum_{i=1}^{{d}({d}+1)/2}t_i(\underset{j\in \boldsymbol{h}_i}{\bigvee}{  \omega}_j^2)].
	\end{align*}
	Now choose $f:\Theta_+^{{d}-1}\to\R_+$ as 
	\begin{equation*}
	f(\boldsymbol{\omega})={d}\sum_{i=1}^{{d}({d}+1)/2}t_i(\underset{j\in \boldsymbol{h}_i}{\bigvee}{\omega}_j^2),    
	\end{equation*} 
	which is---as a linear function of continuous functions---itself continuous on $\Theta_+^{{d}-1}$. 
	The empirical estimator for $f$ is by \eqref{estscal} given as
	\begin{align}\label{sing_ref}
	    \mathbb{\hat{E}}_{\tilde{H}_{\boldsymbol{X}}}[f({\boldsymbol{\omega}})]=\frac{{d}}{k}\sum_{\ell=1}^{n}\sum_{i=1}^{{d}({d}+1)/2}t_i(\underset{j\in \boldsymbol{h}_i}{\bigvee}{\omega}_{\ell j}^2)\mathds{1}_{\{R_\ell\geq R^{(k)}\}}.
	\end{align}
	Applying Theorem~\ref{clt} for the given choice of $f$ it follows that 
	$$\sqrt{k}\boldsymbol{t}^T({\hat{S}}_M-{{S}}_M)\overset{\mathcal{D}}{\to} N(0,w_M),\hspace{5mm}n\to\infty,$$
	where $w_M={d}^2\textrm{Var}_{\tilde{H}_{\bsx}}(\sum_{i=1}^{{d}({d}+1)/2}t_i(\underset{j\in \boldsymbol{h}_i}{\bigvee}{  \omega}_j^2))$.
	By the Cram\'er-Wold device this implies that $$\sqrt{k}({\hat{S}}_M-{{S}}_M)\overset{\mathcal{D}}{\to} N(0,W_M),\quad n\to\infty.$$
	To show that $W_M$ is singular, let $t_i=1$ for $\boldsymbol{h}_i$ such that $|\boldsymbol{h}_i|=1$, and set the remaining components of the vector $\boldsymbol{t}$ to zero. 
	Summarize all these $\bs{h}_i$ into the set $\mathcal{H}^1\coloneqq \{i\in\{1,\dots,{d}({d}+1)/2\}: |\boldsymbol{h}_i|=1\}$, and note that $|\mathcal{H}^1|={d}$ since there are exactly ${d}$ such entries in $S_M$, namely $\sigma_1^2,\sigma_2^2,...,\sigma_{d}^2$ corresponding to the dimension of $\boldsymbol{X}$.
	Then we obtain 
	\begin{align*}
	\boldsymbol{t}^T\hat{S}_M&=\frac{{d}}{k}\sum_{\ell=1}^{n}\sum_{i\in \mathcal{H}^1}t_i(\underset{j\in \boldsymbol{h}_i}{\bigvee}  {\omega}_{\ell j}^2)\mathds{1}_{\{R_\ell\geq R^{(k)}\}}\\
	&=\frac{{d}}{k}\sum_{\ell=1}^{n}\sum_{m=1}^{d}t_m   {\omega}_{\ell m}^2\mathds{1}_{\{R_\ell\geq R^{(k)}\}}\\
	&=\frac{{d}}{k}\sum_{\ell=1}^{n}\sum_{m=1}^{d}   {\omega}_{\ell m}^2\mathds{1}_{\{R_\ell\geq R^{(k)}\}}\\
	&=\mathbb{\hat{E}}_{\tilde{H}_{\boldsymbol{X}}}[{d}\sum_{m=1}^{d}  {\omega}_{m}^2]
	\, = \, {d},
	\end{align*}
	where the last line is due to \eqref{sing_ref} for the particularly chosen $f(\bs{\omega})={d}\sum_{m=1}^{d}  {\omega}_{m}^2={d}$, and the last equality follows from $\frac{{d}}{k}\sum_{\ell=1}^{n}\mathds{1}_{\{R_\ell\geq R^{(k)}\}}={d}$.
	Since this is non-random, the limiting multivariate normal distribution is degenerate.
	
	We proceed now by computing the entries of the covariance matrix.\\
	First, for $i=1,\dots,{d}({d}+1)/2$ we compute the $i$-th diagonal element of $W_M$, corresponding to the asymptotic variance $\textrm{Var}(\hat {S}_{M_{\boldsymbol{h}}})=\textrm{Var}(\hat{\sigma}_{M_{\boldsymbol{h}}}^2)$. 
	We find this from Theorem~\ref{clt} as $\textrm{Var}_{\tilde{H}_{\bsx}}(f(\boldsymbol{  \omega}))$ for $f(\boldsymbol{\omega})={d}(\underset{i\in\boldsymbol{h}}{\bigvee}\omega_i^2)$:
	\begin{align}\label{asy.var}
	\lim_{n\to\infty} \, k \, \textrm{Var}(\hat{S}_{M_{\boldsymbol{h}_i}}) = {d}^2\textrm{Var}_{\tilde{H}_{\bsx}}(\underset{i\in \boldsymbol{h}}{\bigvee}{  \omega}_i^2)={d}^2\Big(\mathbb{E}_{\tilde{H}_{\bsx}}[\underset{i\in \boldsymbol{h}}{\bigvee}{  \omega}_i^4]- (\mathbb{E}_{\tilde{H}_{\bsx}}[\underset{i\in \boldsymbol{h}}{\bigvee}{  \omega}_i^2])^2\Big).
	\end{align}
	Next, when computing the covariance we simply use the identity $2\textrm{Cov}(X,Y)=\textrm{Var}(X+Y)-\textrm{Var}(X)-\textrm{Var}(Y)$.
	Let $\boldsymbol{h}_i\neq \boldsymbol{h}_j\subseteq \{1,\dots,{d}\}$. Consider $f(\boldsymbol{\omega})={d}\big(\underset{k\in \boldsymbol{h}_i}{\bigvee}\omega_k^2+\underset{l\in \boldsymbol{h}_j}{\bigvee}\omega_l^2\big)$. 
	Then using again Theorem \ref{clt} we get
	\begin{align}\label{as.covar}
	\lim_{n\to\infty} \, 2k \, \textrm{Cov}(\hat{S}_{M_{\boldsymbol{h}_i}},\hat{S}_{M_{\boldsymbol{h}_j}}) = {d}^2( \textrm{Var}_{\tilde{H}_{\bsx}}(\underset{k\in \boldsymbol{h}_i}{\bigvee}{  \omega}_k^2+\underset{l\in \boldsymbol{h}_j}{\bigvee}{\omega}_l^2)-\textrm{Var}_{\tilde{H}_{\bsx}}(\underset{k\in \boldsymbol{h}_i}{\bigvee}{\omega}_k^2)-\textrm{Var}_{\tilde{H}_{\bsx}}(\underset{l\in \boldsymbol{h}_j}{\bigvee}{\omega}_l^2)).
	\end{align}
\end{proof}

The asymptotic covariance matrix $W_m$ can also be expressed in terms of the squared entries of the matrix $A$.

\bco \label{mlclt}
Let the assumptions of Theorem~\ref{scalclt} hold.
Then the entries of $W_M$ are given by the right-hand sides of the following two limits:
On the diagonal we obtain
	$$\lim_{n\to\infty} \, k\,\emph{Var}(\hat{S}_{M_{\boldsymbol{h}}}) = {d}(\sum_{j=1}^{{d}}\underset{i\in \boldsymbol{h}}{\bigvee}\frac{a_{ij}^4}{\norm {a_j}^2_2})-
	  (\sum_{j=1}^{d}(\underset{i\in \boldsymbol{h}}{\bigvee}a_{ij}^2))^2>0,$$
	where $(\sum_{j=1}^{d}(\underset{i\in \boldsymbol{h}}{\bigvee}a_{ij}^2))^2=\sigma_{M_{\boldsymbol{h}}}^4$.
	For the non-diagonal entries we obtain
	$$\lim_{n\to\infty} \,k\,\emph{Cov}(\hat{S}_{M_{\boldsymbol{h}_i}},\hat{S}_{M_{\boldsymbol{h}_j}}) = {d}\sum_{k=1}^{d}(\frac{1}{\norm {a_k}^2_2}(\underset{m\in \boldsymbol{h}_i}{\bigvee}a_{mk}^2)(\underset{l\in \boldsymbol{h}_j}{\bigvee}a_{lk}^2))
	-((\sum_{k=1}^{d}\underset{m\in \boldsymbol{h}_i}{\bigvee}a_{mk}^2)(\sum_{k=1}^{d}\underset{l\in \boldsymbol{h}_j}{\bigvee}a_{lk}^2))^2,$$
	where $\sum_{k=1}^{d}\underset{m\in \boldsymbol{h}_i}{\bigvee}a_{mk}^2 = \sigma_{M_{\boldsymbol{h}_i}}^2$ and  $\sum_{k=1}^{d}\underset{l\in \boldsymbol{h}_j}{\bigvee}a_{lk}^2= \sigma_{M_{\boldsymbol{h}_j}}^2.$
	\eco
	
\begin{proof}
		With the explicit form of the spectral measure (\ref{discretspecteq}), expression (\ref{asy.var}) becomes:  
		\begin{align*}
		{d}^2\textrm{Var}_{\tilde{H}_{\bsx}}(\underset{i\in \boldsymbol{h}}{\bigvee}{\omega}_i^2)
		&={d}(\sum_{j=1}^{d}\underset{i\in \boldsymbol{h}}{\bigvee}\frac{a_{ij}^4}{\norm {a_j}^2_2})-    (\sum_{j=1}^{d}(\underset{i\in \boldsymbol{h}}{\bigvee}a_{ij}^2))^2.
		\end{align*}
Similarly, for the covariance, from (\ref{as.covar}) we obtain:
	\begin{align*}
	&{d}^2\textrm{Var}_{\tilde{H}_{\bsx}}(\underset{m\in \boldsymbol{h}_i}{\bigvee}{\omega}_m^2+\underset{l\in \boldsymbol{h}_j}{\bigvee}{\omega}_l^2)\\
	&={d}\sum_{k=1}^{d}(\underset{m\in \boldsymbol{h}_i}{\bigvee}\frac{a_{mk}^4}{\norm {a_k}^2_2}+\underset{l\in \boldsymbol{h}_j}{\bigvee}\frac{a_{lk}^4}{\norm {a_k}^2_2}+\frac{2}{\norm {a_k}^2_2}(\underset{m\in \boldsymbol{h}_i}{\bigvee}a_{mk}^2)(\underset{l\in \boldsymbol{h}_j}{\bigvee}a_{lk}^2))
-((\sum_{k=1}^{d}\underset{m\in \boldsymbol{h}_i}{\bigvee}a_{mk}^2)+ (\sum_{k=1}^{d}\underset{l\in \boldsymbol{h}_j}{\bigvee}a_{lk}^2))^2 
	\end{align*}
Summing the variance terms together we obtain
	\begin{align*}
	2k\textrm{Cov}_{\tilde{H}_{\bsx}}({S}_{M_{\boldsymbol{h}_i}},{S}_{M_{\boldsymbol{h}_j}})
	&=2{d}\sum_{k=1}^{d}(\frac{1}{\norm {a_k}^2_2}(\underset{m\in \boldsymbol{h}_i}{\bigvee}a_{mk}^2)(\underset{l\in \boldsymbol{h}_j}{\bigvee}a_{lk}^2))-2(\sum_{k=1}^{d}\underset{m\in \boldsymbol{h}_i}{\bigvee}a_{mk}^2)(\sum_{k=1}^{d}\underset{l\in \boldsymbol{h}_j}{\bigvee}a_{lk}^2).
	\end{align*}  
	The identities giving $\sigma_{M_{\boldsymbol{h}}}^2$, $\sigma_{M_{\boldsymbol{h}_i}}^2$, and $\sigma_{M_{\boldsymbol{h}_j}}^2$ follow from \eqref{scaleq}.
\end{proof}

Finally we prove asymptotic normality of the estimated ML coefficient matrix $A$ computed via Theorem~\ref{consT} as $A^2= T S_M$. 
As $T$ is a deterministic matrix, we obtain ${\hat{A}}^2$ as $T \hat{S}_M$. 
Then Theorem~\ref{scalclt} and Corollary~\ref{mlclt} gives the asymptotic normality of the estimated ML coefficient matrix $A$.

\begin{theorem}\label{mlclt2}
Let $\boldsymbol{X}=A\times_{\max}\boldsymbol{Z}$ be a recursive ML vector satisfying  (A1)-(A3), and let $\boldsymbol{X}_1,\dots,\boldsymbol{X}_n$ be i.i.d. copies of $\boldsymbol{X}$.
	Let the assumptions of Theorem~\ref{scalclt} hold and
	assume that 
	the ML coefficient matrix $A =({a}_{ij})_{{d}\times{d}}$ satisfies
	$${d}(\sum_{j=1}^{d}\underset{i\in \boldsymbol{h}}{\bigvee}\frac{a_{ij}^4}{\norm {a_j}^2_2})-    (\sum_{j=1}^{d}(\underset{i\in \boldsymbol{h}}{\bigvee}a_{ij}^2))^2 >0.$$
	Estimate ${\hat{A}}^2 = T \hat S_M$ with $T$ as in Theorem \ref{consT} and $\hat S_M$ as in \eqref{smest}.
	 Then
	$$\sqrt{k}({\hat{A}}^2-{A}^2)\overset{\mathcal{D}}{\to} N(0, TW_MT^T),\hspace{5mm}n\to\infty,$$
	where $W_M$ is the covariance matrix in Corollary~\ref{mlclt}.
\end{theorem}

	\section{Data Applications}\label{s8}

	\subsection{Structure Learning and Estimation}\label{s81}

For estimating a causal order as well as for the estimation of the ML coefficient matrix $A$ we need estimates for the scalings in {Algorithms~\ref{recalg2}, \ref{alg8}, and \ref{alg9}}, respectively. Notice that in Proposition~\ref{estalg2} we estimate the ML coefficient matrix $A$ for a well-ordered DAG, so that we first estimate the order of the nodes and then $A$.
The structure learning is based on the scalings of $M_{\boldsymbol{h}_a, m^*_a, \{\boldsymbol{h}\cup \{m^*\}\}^c}$ as defined in \eqref{partscaled}, and the estimation of $A$ as in Proposition~\ref{estalg2} is based on scalings of $M_{\bs{h}}$ as in \eqref{maxp}. {In contrast to Example~\ref{10ex} we make no distributional assumptions on $\bsx$, but use the non-parametric estimation method developed in Section~\ref{regcase}.}

{For the non-parametric estimation of all scalings needed in the algorithms, as in Section~\ref{regcase} we denote by $k$ the number of upper order statistics corresponding to the radial threshold used for the estimation of {the spectral measure}. We recall that it has to be chosen as $k=o(n)$ and we choose $k\approx\sqrt{n}$.
We observe that we may have rather few components exceeding the radii computed as in \eqref{pol.est} based on all $d$ components.  
	If some components of an observation are very large, then other components may not exceed the corresponding threshold. 
	As the choice of the $k$ {upper order statistics} is based on these radii, there may be rather few exceedances in some component.
Hence, we resort to lower dimensional vectors $\boldsymbol{X_{{\boldsymbol{q}}}}=(X_i: i\in {\boldsymbol{q}})$ for appropriate sets ${\boldsymbol{q}}\subseteq\{1,...,d\}$ for the estimation of the various scalings.

	\subsubsection{Structure Learning}\label{s811}
	
	We want to apply Algorithms \ref{alg8} and \ref{alg9} for structure learning by replacing the theoretical scalings by their estimated counterparts. However, we have to modify both algorithms to account for estimation errors. 
	
As the limited number of exceedances {of radii} in some components is particularly critical for identification of the initial nodes, 
	we modify the estimation procedure, which has been presented in \eqref{Mmam}, and estimate for every $m\in\{1,\dots,d\}$ the initial nodes only based on $\max(X_{i}, aX_{m})$ for $a>1$ (corresponding to $\bs{q}=\{i,m\}$), but then for all $i\in\{1,\dots,d\}\setminus \{m\}$.
	This means that the identification of the initial nodes is carried out by applying a pairwise version of Theorem~\ref{theo2} for all pairs.
	The following pairwise version of Algorithm~\ref{alg8}
	identifies the intial nodes of the DAG. 
	It is based on Theorem~5.6 and Algorithm~3 of \cite{Krali}, adapted for possible estimation errors.
 The positive bounds $\eps_1, \eps_2$ have to be chosen appropriately to ensure that the estimates $\hat{\Delta}_{i^*m^*}$ are close to zero. The scalar $a>1$ has to be chosen in accordance with Theorem~\ref{theo2}.
		
		\begin{algorithm}[H]
			\caption{Identifying the initial nodes of $\boldsymbol{X}$ }
			\label{datdalg}
			\begin{algorithmic}[1]
				\Procedure{}{}
				\State \textbf{Set} $\hat{\Delta} = (0)_{d\times d}; {N} = (0)_{1\times d}; {a>1}; \eps_1>0;  \eps_2>0$ 
				\State \textbf{for} $m^* \in \{d^*,...,1^*\}$ \textbf{do}
				\State \hspace{5mm}\textbf{for} $i^*\in \{d^*,...,1^*\}$ \textbf{do}
				\State \hspace{10mm} 
				$\hat{\Delta}_{i^*,m^*}=\hat{\sigma}_{M_{a m^*,i^*}}^2-\hat{\sigma}_{M_{i^*,m^*}}^2-a^2+1$
				\State \hspace{5mm}\textbf{end for} 
				\State \hspace{5mm}\textbf{Set} $\hat{\Delta}_{m^*} = (\hat{\Delta}_{1^*,m^*},...,\hat{\Delta}_{d^*,m^*})$
				\State \textbf{end for}
				\State \textbf{Set} $I_1 = \{m^*\in\{d^*,...,1^*\}:\{\underset{i^*\in\{d^*,...,1^*\}}{\max}\hat{\Delta}_{i^*,m^*}\leq\eps_1\}\cap \{\underset{i^*\in\{d^*,...,1^*\}}{\min}\hat{\Delta}_{i^*,m^*}\geq-\eps_2\}\}$
					
				\State \textbf{for} $m^* \in \{d^*,...,1^*\}$ \textbf{do}
				\State  \hspace{5mm}\textbf{if} $m^*\in I_1$,  \textbf{then} \\
				\hspace{15mm}
				$N_{m^*}=m^*$
				\State \hspace{5mm}\textbf{else} $N_{m^*}=0$
				\State \textbf{end for}.
				\
				\EndProcedure
			\end{algorithmic}
		\end{algorithm}

		{	Once the initial nodes are identified, we proceed finding the descendants. When searching for the $(h+1)$-th node, we apply the following modification of Algorithm~\ref{alg9}, which has again been adapted for possible estimation errors by an application of Corollary~\ref{algord}.}
			
			\begin{algorithm}[H]
				\caption{Identifying the ($h+1$)-th {node} of $\boldsymbol{X}$}\label{last.alg1}
				\begin{algorithmic}[1]
					\Procedure{}{}
					\State \textbf{Set} $\hat{\Delta} = (0)_{1\times{d}}; N = (0)_{1\times {d}}; a>1; {\boldsymbol{h}=\{\nu^{-1}({d}),\dots,\nu^{-1}({d}-h+1)\}}$
					\State \textbf{for} $i^*\notin \boldsymbol{h}$ \textbf{do} 
					\State \hspace{5mm} 
					$\hat{\Delta}_{i^*}=\hat{\sigma}_{{M}_{\boldsymbol{h}_a,i^*_a,\{\boldsymbol{h}\cup \{i^*\}\}^c}}^2-\hat{\sigma}_{{M_{\boldsymbol{h},i^*,\{\boldsymbol{h}\cup \{i^*\}\}^c}}}^2-(a^2-1)\hat{\sigma}_{M_{\boldsymbol{h}, i^*}}^2$
					%
					\State\textbf{end for} 
					
					\State \textbf{for} $i^*\notin\boldsymbol{h}$ \textbf{do}
					\State \hspace{5mm} \textbf{if}  $\hat{\Delta}_{i^*}=\underset{i^*\notin\boldsymbol{h}}{\max}\hat{\Delta}_{i^*}$, \textbf{then}
					\State \hspace{10mm} $N_{i^*}=i^*$
					\State \hspace{5mm} \textbf{else} $N_{i^*}=0$
					\State \textbf{end for}. 
					\EndProcedure
				\end{algorithmic}
			\end{algorithm}
				{In contrast to Example~\ref{10ex}, where our goal was to identify the generations of the graph, here we are only interested in a causal order of the nodes, and thus modify Algorithm \ref{alg9} based on Corollary~\ref{algord} so that it returns a unique node at each step of the if-loop.}
	\subsubsection{Estimation of the Scalings}\label{s812}
	
	According to Line 4 of Algorithm~\ref{last.alg1} three squared scalings need to be estimated:
	
	-${\sigma}_{{M}_{\boldsymbol{h},i^*,\{\boldsymbol{h}\cup \{i^*\}\}^c}}^2$\hspace{1.6mm}: By \eqref{partscaled}, this estimate involves all $d$ components of $\boldsymbol{X}$. 
	Thus, we set $\boldsymbol{q}\coloneqq\{1,...,d\}$ and proceed as in step (i) below;
	
	-${\sigma}_{M_{\boldsymbol{h}, i^*}}^2\hspace{1.42cm}$: We estimate the spectral measure based on $\boldsymbol{q}\coloneqq\boldsymbol{h}\cup\{i^*\}$  and proceed as in step (i) below; {such scalings we also need to estimate in Line 5 of Algorithm~\ref{datdalg};}
	
	-${\sigma}_{{M}_{\boldsymbol{h}_a,i^*_a,\{\boldsymbol{h}\cup \{i^*\}\}^c}}^2$: By \eqref{partscaled}, this is the estimated scaling of the rescaled vector $\boldsymbol{X}$ and we follow step (ii) below.  
	
	\medskip
	
\noindent
	For Algorithm~\ref{recalg2} we have to estimate the following squared scalings for $i,j\in\{1,...,d\}$ and $i\leq j+1$:
	
	-${\sigma}_{M_{i,j,j+1,...,d}}^2$: We estimate the spectral measure based on $\boldsymbol{q}\coloneqq\{i,j,j+1,...,d\}$.
	
	-${\sigma}_{M_{j,j+1,...,d}}^2$\hspace{1.75mm}:  Here we set $\boldsymbol{q}\coloneqq\{j,j+1,...,d\}$.

The following two steps modify the setting of Section~\ref{regcase} and summarize the estimation of the scalings in both Algorithms~\ref{datdalg}, \ref{last.alg1}, and Algorithm~\ref{recalg2}.

	(i) For the estimation of the squared scalings $\sigma_{M_{\boldsymbol{q}}}^2$ of $M_{{\boldsymbol{q}}}$ 
	for some subset of components ${\boldsymbol{q}}\subseteq\{1,...,d\}$ {as specified above},
	we take $\boldsymbol{X_{{q}}}=(X_i: i\in {\boldsymbol{q}})$, and compute for each observation $\ell\in\{1,\dots,n\}$, the (reduced) polar representation
	\begin{align}\label{bh}
	R_\ell\coloneqq\norm{\boldsymbol{X}_{\ell\boldsymbol{q}}}_2 =  (\sum_{i\in\boldsymbol{q}} X_{\ell i}^2)^{1/2}
	\quad\mbox{and}\quad {\boldsymbol{\omega}_\ell=({\omega}_{\ell i}: i\in \boldsymbol{q})\coloneqq\frac{\boldsymbol{X}_{\ell\boldsymbol{q}}}{R_\ell}}, \hspace{5mm} \ell=1,...,n.
	\end{align}
	Then for $1\le k\le n$ we estimate ${\sigma}_{M_{\boldsymbol{q}}}^2$ as
\begin{align*}
\hat{\sigma}_{M_{\boldsymbol{q}}}^2 =\frac{|\boldsymbol{q}|}{k}\sum_{\ell=1}^{n}\underset{j\in \boldsymbol{q}}{\bigvee}{\omega}_{\ell k}^2\mathds{1}{\{R_\ell \geq R^{(k)}\}}.
\end{align*}

	When $\boldsymbol{q}=\{i\}$, corresponding to the scaling of a single component, then by \eqref{bh}, $\omega=1$, and plugging this in the estimator in \eqref{estalg2}, we obtain $\hat{\sigma}_i^2=1$, which is the true scaling parameter $\sigma_i=1$.
	
	(ii) For the estimation of the scaling of $M_{\boldsymbol{h}_a, m^*_a, \{\boldsymbol{h}\cup \{m^*\}\}^c}$, we replace $\boldsymbol{X_q}$ above by the rescaled version $(aX_{\bs{h}}, aX_{m^*}, X_{\{\bs{h}\cup\{m^*\}\}^c})$, and re-apply the same procedure {as in (i)} to obtain new polar representations, say $(R_{a_\ell },\boldsymbol{\omega}_{a_\ell })$, and then estimate:
	$$\hat{\sigma}_{M_{\boldsymbol{h}_a, m^*_a, \{\boldsymbol{h}\cup \{m^*\}\}^c}}^2 
	=\frac{(a^2-1)(|\boldsymbol{h}|+1)+d}{k}\sum_{\ell=1}^{n}\underset{j=1 }{\overset{d}{\bigvee}}{\omega}_{a_\ell j}^2\mathds{1}{\{R_{a_\ell  }\geq R_a^{(k)}\}}.$$
	The numerator, $(a^2-1)(|\boldsymbol{h}|+1)+d$ corresponds to the new mass of the spectral measure as a consequence of the scaling by $a$ of the components involved in $X_{\boldsymbol{h}}$ and $X_{m^*}$, see for instance Lemma~\ref{usescale}(b).
	
\subsubsection{Estimating the ML Coefficient Matrix}\label{s813}

After having estimated also the scalings needed for Algorithm~\ref{recalg2}, we have to take care of estimation errors. Indeed, it can happen that entries of $A^2$ are estimated as being negative.
For the two data examples to follow we simply set $\hat A=\sqrt{\max (\hat A^2,0)}$, {with the square root taken entrywise}, and keeping all positive estimates. For larger networks it may be advisable to choose a thresholding or lasso procedure to obtain a sparse graph. 

	\subsection{Industry Portfolio Data}\label{s82}
			
	The data consist of seven time series of value-averaged daily percentage returns, each assigned to one of seven industry portfolios as part of the 30-Industry-Portfolio in the Kenneth French Data Library available at \url{https://mba.tuck.dartmouth.edu/pages/faculty/ken.french/data library.html}.
	
	All 30 portfolios have been analysed in \cite{cooley,JanWan}, where in \cite{cooley} it is suggested that the tails of the data are regularly varying. We refer to the introduction for more details on the objectives of these papers.

The seven portfolios we consider are:
	Chemicals (Ch), Fabricated Products (FP), Electrical Equipment (EE),  Healthcare (H), Smoke (S),  Utilities (U), and Others (O), which includes products which are not specific to any of the other listed industries. 
	A precise description of the data, in particular of each industry sector can be found on the website above.
	
	The data has been collected over the years 1950-2015. 
	{Since the time series over this time period is non-stationary and, in particular, since also the dependence structure changes over this long period,}
	we have selected the time window from $01.06.1989$ to $15.06.1998$ containing $2285$ observations, which show marginal stationarity. 
	To each of the 7 time series we have fitted moving average processes of order 3, {with the exception of Others where we have fitted a moving average process of order 4}, and  performed a Ljung-Box test with 8 lags on the residuals. The test did not reject the independence hypothesis of the residuals, supporting the assumption that the time series are stationary. Figure~\ref{TS} depicts the time series plots in \%-returns for the seven industries. 

\begin{figure}[ht]
	\vspace{-.4cm}
	\hspace*{-.8cm}\includegraphics[ height=5cm, width=17cm, trim={.0cm -.6cm -.0cm .9cm}, clip]{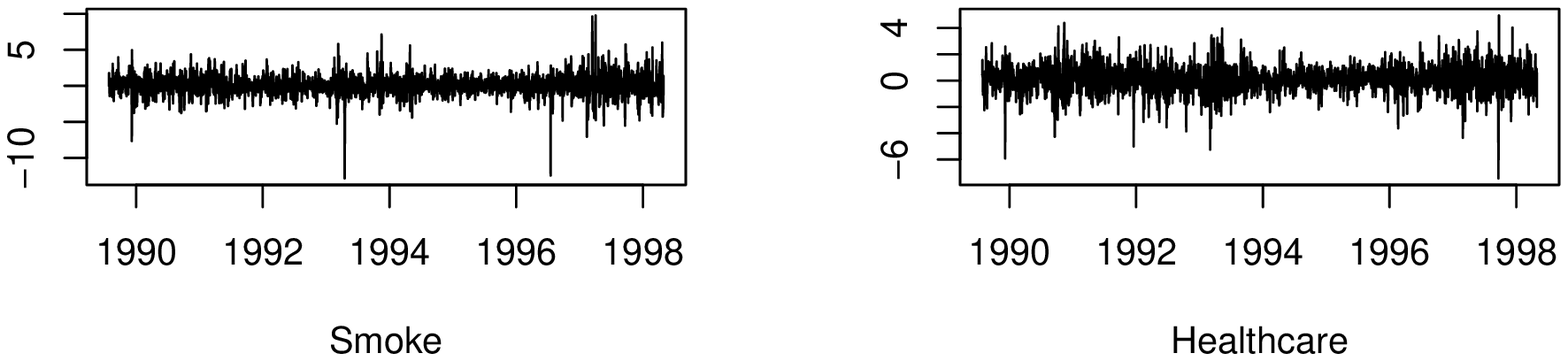}\centering\vspace{-1.0cm}
	\hspace*{-.8cm}\includegraphics[ height=5cm, width=17cm,   trim={.0cm -.6cm -.0cm .9cm}, clip]{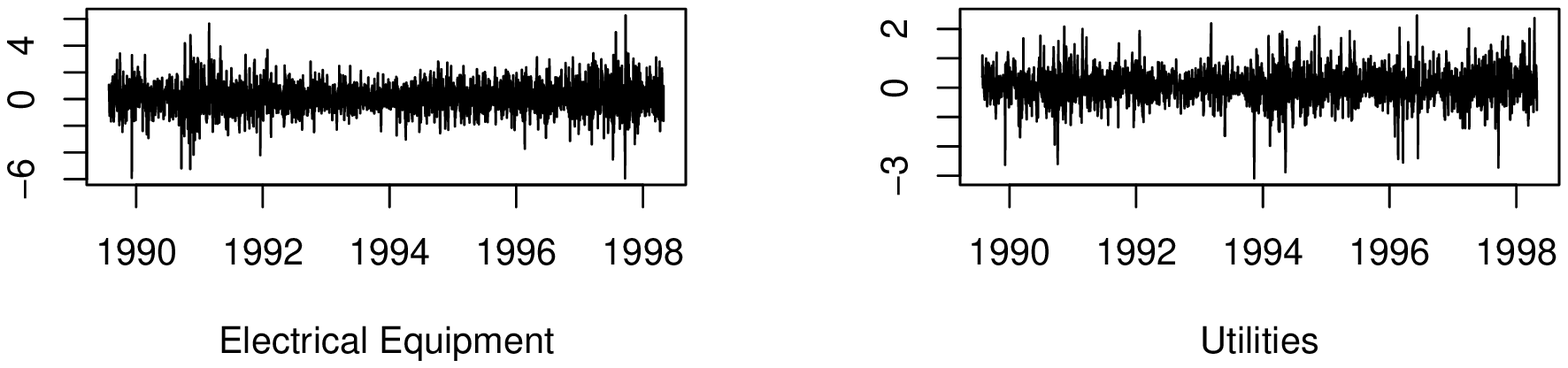}\vspace{-1cm}
	\hspace*{-.8cm}\includegraphics[ height=5cm, width=17cm,   trim={.0cm -.6cm -.0cm .9cm}, clip]{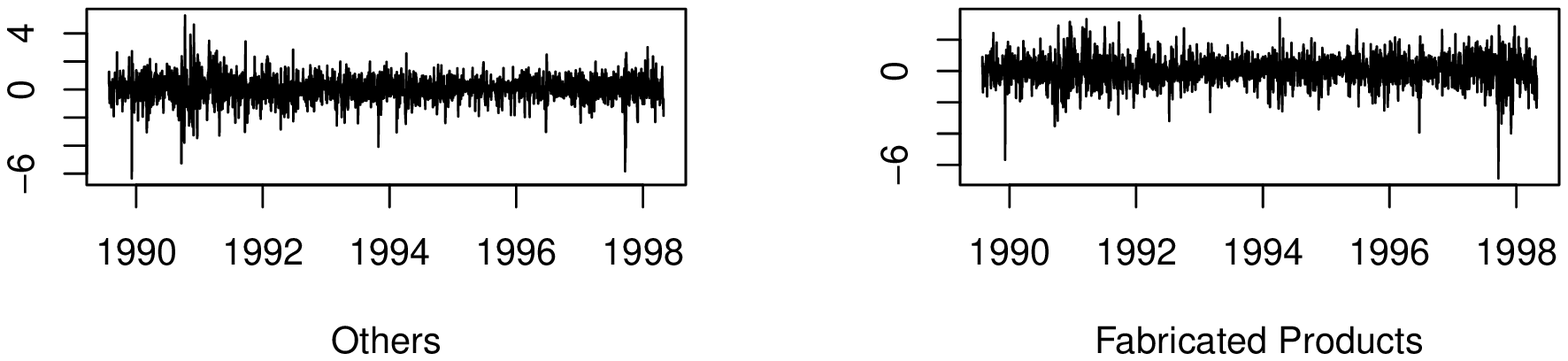}\vspace{-1cm}
	\hspace*{3.55cm}\includegraphics[ height=5cm, width=17.4cm, trim={.0cm -.6cm -.0cm .9cm}, clip]{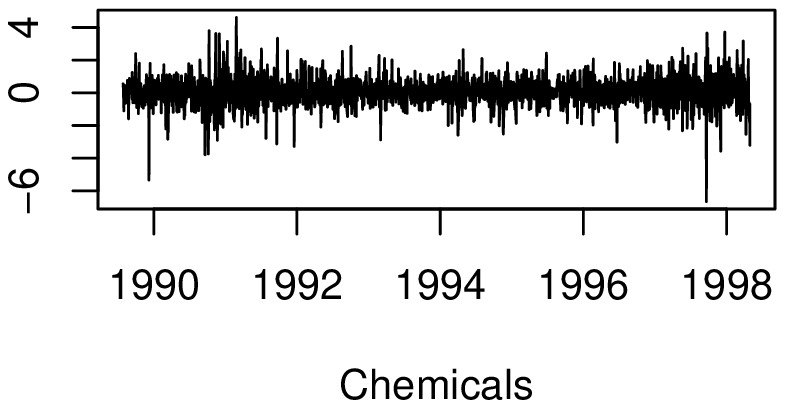}\vspace{-1cm}
	\caption{Time-series of \%-returns for the seven industries from 01.06.1989 until 15.06.1998.}
	\label{TS}
\end{figure}

As we assess dependence in extreme negative returns, we transform the vector of the 7 portfolio time series to $\boldsymbol{X}^{*}=\max(-\boldsymbol{X},0)$. 
Our aim is to fit a recursive ML model model to $\boldsymbol{X}^*$.
We transform the data by the empirical integral transform to standardize them to Fr\'echet$(2)$ margins (see for instance, p.~381 in \cite{beirlant}, or \cite{cooley}).
	We map $(S, H, EE, U, O, FP, Ch)\mapsto (1,2,3,4,5,6,7)$ and define for $i=1,\dots,7$
\begin{align}\label{empstand}
X_{\ell i}\coloneqq\Big(-\log\big\{\frac{1}{n+1}\sum_{j=1}^{n}\mathds{1}_{\{X^*_{ji}\leq X^*_{\ell i}\}}\big\}\Big)^{-1/2},\quad  \ell=1,\dots,n=2285.
\end{align}
	We also collect the 7-dimensional data into a vector time series $\bs{x}^*_\ell=(x^*_{\ell,1},x^*_{\ell,2},x^*_{\ell,3},x^*_{\ell,4},x^*_{\ell,5},x^*_{\ell,6},x^*_{\ell,7})$, $\ell=1,\dots,2285$.

{\em Running maxima.} In order to provide some insight in the data structure, we start our analysis with a time-line of the high risk events of the seven industry portfolios and their association with shocks entering the industry network from the relevant component of the innovation vector $\boldsymbol{Z}$. 
	The horizontal axes of Figure~\ref{occur} shows every time point, when the maximum over all seven standardized industry returns happens and indicates on the vertical axes the respective innovation component causing this maximum. 
		Since the innovations $Z_i$ for $i=1,\dots,d$ are atomfree and by Assumption~(1) independent, representation  \eqref{rmlmequat} ensures that each recursive ML component $X_i$ realises its maximum in exactly one innovation. 
		If two components of $\bsx$ are realised by the same innovation, still the realised values of the two components are different by Lemma~\ref{ineq}, which implies that $\max\{X_1,...,X_d\}$ is unique. These innovations are indicated in Figure~\ref{occur}. 
		\begin{figure}[H]
			\vspace{-0.4cm}\hspace*{-.8cm}\includegraphics[ height=5.5cm, width=17cm,   trim={.0cm -.6cm -.0cm .9cm}, clip]{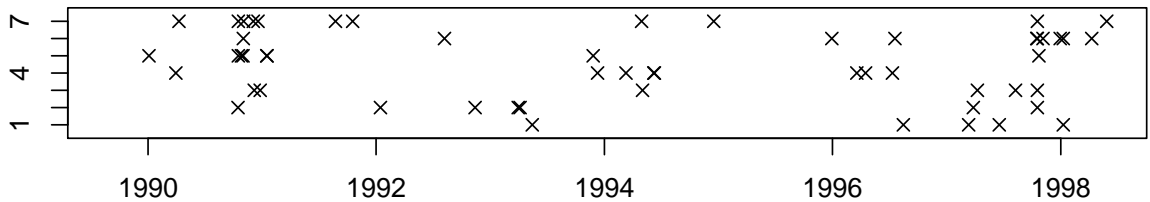}
			\vspace{-2.5cm}
			\caption{Dates of the maximal occurrences associated to the components $i\in\{1,...,7\}$.}\label{occur}
		\end{figure}
		We find that most of the shocks originate from the initial node Chemicals (7) during the time period 1990-1991 which coincides with the Gulf War, caused by the invasion of Kuwait by Iraq. During the war it was feared that Iraq made use of chemical warfare. 
		Regarding Fabricated Products (6), the larger losses occur close to the end of 1997, which is associated with the slowdown in the Asian economies and which had spillover effects on the U.S economy, eventually leading also to the October 27, 1997 Mini-Crash. 
		The Utilities (4) experience large losses in the years 1994 and 1996 associated with deregulation of the electric energy supply in the US, which was initiated in 1992. 
		The Healthcare sector (2) experiences losses in the period 1992-1993 associated with the Clinton Health Care reform.
		
		Our final goal is to approximate the causal dependence structure via a recursive ML model by means of the learning algorithm presented in the previous sections of this paper. This algorithm is based on all returns above a high threshold, not only on the maximum value.

{\em Bivariate extremes.} In a second exploratory analysis we plot the bivariate extremes (real data and simulated ones) in Figure~\ref{tb1plot} of Appendix~\ref{p}. We also simulate a 7-dimensional random vector $\boldsymbol{X}\in RV^{7}_+(2)$ from an innovation vector $\boldsymbol{Z}\in RV^{{7}}_+(2)$
            with independent standard Fr\'echet(2) components of dimension $n=2285$ via $\boldsymbol{X}=\hat{A}\times_{\max} \boldsymbol{Z}$, where the estimated matrix $\hat{A}$ is given in \eqref{esA}.
            Two columns always belong together, the left one gives the empirical bivariate extremes of two of the seven components, respectively, which have also been the basis for the estimation procedure of Section~\ref{s81}. The right one presents a simulation of the estimated model. We plot only those bivariate observations with the 50 largest radii.
            Left and right (real and simulated data) look very much alike, indicating that the estimated bivariate models are valid approximations to the biviariate empirical distribution in the tails. 
            
We give an interpretation of
		{\em Utilities versus Electrical Equipment} (line {6},  columns 3 and 4): Here we notice that large losses in Utilities do not necessarily occur together with high risk for Electrical Equipment. {This can be verified in the plot of realised bivariate extremes (column 3), where large losses for Utilities occurring close to the vertical axis correspond only to negligible losses for Electrical Equipment.} This suggests that the common large losses between Utilities and Electrical Equipment are rather caused by their common ancestors. This may be due to idiosyncratic risks associated with one but not the other, {which in this case might correspond to the innovation terms $Z_3$ $(\hat{a}_{43}=0)$ and $Z_4$ $(\hat{a}_{34}=0)$.} 
	
{\em Fitting a recursive ML model.} Finally, we approximate the extreme dependence structure of $\boldsymbol{X}^*$ by a recursive ML model.
To this end, according to Section~\ref{regcase}, we have to choose a threshold value $k=o(n)$ of the radial components.
We choose $k\approx\sqrt{n}$ and set $k=50$.
For identification of the initial nodes we employ Algorithm~\ref{datdalg} with $a=1.01$ and $\eps_1=0.0045,\eps_2=0.0045$, and then in order to reorder the remaining nodes Algorithm~\ref{last.alg1}.
As output we obtain the ordered vector $$\boldsymbol{X}=(X_S,X_{H},X_{EE},X_U,X_O,X_{FP},X_{Ch}).$$
	Finally, we estimate the ML coefficient matrix  by the estimation version of Algorithm~\ref{recalg2} and setting $\hat{A}=\sqrt{\max(\hat{A^2},0)}$.
	The estimated  standardised ML coefficient matrix is given by
\begin{align}
	\hspace{-15mm}
	\hat{A}=\begin{bmatrix}
0.708 & 0 & 0.462& 0 & 0.142 & 0.333 & 0.476\\
0 & 0.566 & 0.459 & 0 &0.076 & 0.147& 0.687\\
0 & 0 & 0.649& 0& 0.344&  0.136& 0.686\\
0 & 0 & 0 & 0.709& 0.333& 0.188 &  0.593\\
0 & 0 & 0 & 0 & 0.682 & 0.250 & 0.688\\
0 & 0 & 0 & 0 & 0 & 0.674& 0.739\\
0 & 0 & 0 & 0 & 0 & 0& 1\\
	\end{bmatrix}.
	\label{esA}
	\end{align}
The estimated squared scaling parameters of the components of $\boldsymbol{X}$ are obtained by summing the squared entries of the respective row. We find the estimated vector of scalings $(1.036, 1.01, 1.01, 1, 1, 1, 1)$ and recall that the theoretical ones are all equal to 1. Deviations from scalings of 1 stem from the fact that we set $\hat{A}=\sqrt{\max(\hat{A}^2,0)}$. In doing so, once $\hat{A}^2$ has been computed, we ignore its entries that are close to zero but negative, for instance $\hat{a}_{12}^2, \hat{a}_{14}^2, \hat{a}_{24}^2, \hat{a}_{34}^2$. Consequently this can make the sums of the square entries of the respective row be slightly greater than one.  
	
		The DAG corresponding to $\hat{A}$ is given in Figure \ref{founddag}. We recall that $a_{ij}=0$ {implies no edge} from $j$ to $i$. 
		
		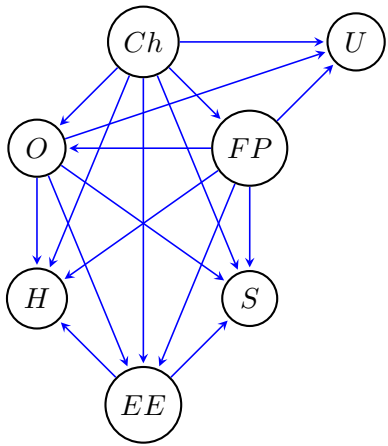
\begin{figure}[htbp]
			\centering
			\begin{tikzpicture}[
			> = stealth, 
			shorten > = 1pt, 
			auto,
			node distance = 2cm, 
			semithick 
			]
			\tikzstyle{every state}=[
			draw = black,
			thick,
			fill = white,
			minimum size = 4mm
			]
			
			\node[state,  ] (7) {$Ch$};
			
			\node[state] (6) [below right of=7]{$FP$};
			\node[state] (5) [below left of=7] {$O$};
			
			\node[state] (4) [above right of=6] {$U$};
			\node[state] (2) [below of=5] {$H$};
			\node[state] (1) [below of=6] {$S$};
			\node[state] (3) [below left of=1] {$EE$};
			
			\path[->][blue] (7) edge node {} (5);
			\path[->][blue] (7) edge node {} (6);
			\path[->][blue] (7) edge node {} (4);
			\path[->][blue] (7) edge node {} (3);
			\path[->][blue] (7) edge node {} (2);
			\path[->][blue] (7) edge node {} (1);
			
			\path[->][blue] (6) edge node {} (5);
			\path[->][blue] (6) edge node {} (4);
			\path[->][blue] (6) edge node {} (3);
			\path[->][blue] (6) edge node {} (2);
			\path[->][blue] (6) edge node {} (1);
			
			\path[->][blue] (5) edge node {} (4);
			\path[->][blue] (5) edge node {} (3);
			\path[->][blue] (5) edge node {} (2);
			\path[->][blue] (5) edge node {} (1);

			\path[->][blue] (3) edge node {} (2);
		\path[->][blue] (3) edge node {} (1);

			\end{tikzpicture}
			\small{\caption{{Learned} DAG structure for the seven considered industries. Note that Chemicals is the only initial node of the DAG.}\label{founddag}}
		\end{figure}

			The estimated DAG should provide insight into the causality structure of the seven industries. {We  associate the estimated scalings with the standardised risk of each component.} Then the estimated entries in $\hat A$ indicate the {proportions} of risk inferred from the causal dependence. 
			
		The only initial node is Chemicals, whose high losses impact risk on all other industries as it has out-degree 6. For instance, the line of productions in Fabricated Products, Healthcare (which includes pharmaceutical industry), Electrical Equipment, Utilities (which includes electricity services and supply), and Smoke are highly dependent on the supply of chemical products and chemical processing.
		
		Fabricated Products has out-degree 5 with its high risk affecting Utilities, Others, Smoke, Healthcare, and Electrical Equipment. A reason for this may lie in the industrial fabrication of many products in the affected industries.
		In particular the impact on Smoke, whose production line depends heavily on machinery, is stronger than that on Utilities, Electrical Equipment, and Others, since $a_{16}>\max(a_{26},\dots,a_{56})$.
		
		Others, whose components include also Cogeneration Power Producers, has out-degree 4 with its high risk affecting Utilities, Smoke and Electrical Equipment, and Healthcare to a lesser extent. On the other hand, Others has in-degree 2, so high risk in Chemicals or Fabricated Products affects Others, which can be seen from $a_{57}$ and $a_{56}$ with higher influence from Chemicals.
		
		Electrical Equipment has out-degree 2 and in-degree 3, so its high risk is caused by Others, Chemicals, and Fabricated Products, where the influence of Chemicals is about {twice} as large as Others, and the influence of high risk in Fabricated Products is much lower.
		On the other hand, high risk in Electrical Equipment impacts on Healthcare and Smoke in about equal proportions. 
		
		Utilities, Healthcare, and Smoke have out-degree 0, so these portfolios are affected by high risk of other portfolios, but their high risks do not {spread} elsewhere. This is seen from columns 1,2, and 4 of $\hat A$, where the quantities on the diagonal correspond to the idiosyncratic risk. The quantities to the right measure the high risk influencing these three portfolios.

\subsection{Dietary Supplement Data}\label{s83}

		The data is taken from a dietary interview from the NHANES report for the year 2015-2016, which is available at \url{https://wwwn.cdc.gov/Nchs/Nhanes/2015-2016/DR1TOT_I.XPT}; here also more details about the 168 data components can be found. The objective is that of estimating the total intake of calories, nutrients and non-nutrient food components from foods and beverages consumed a day prior to the interview. 
	From the above data	38 components have been investigated in \cite{JanWan} using the clustering approach mentioned already in the introduction. 
		
		We focus on four of the components, Vitamin A (DR1TVARA), Beta-Carotene (DR1TBCAR),  Lutein+Zeaxanthin (DR1TLZ) and Alpha-Carotene (DR1TACAR). We abreviate them as VA, BC, LZ, AC, respectively. 
		{For each component there are $n=9544$ observations, each corresponding to a different individual and generated from survey interviews, thus the data sample can be treated as an i.i.d. sample. }
		A Hill plot (see e.g. \cite{embc}, Section~6.4) suggests that all data components are regularly varying with some positive index.  We map (VA, BC, LZ, AC) $\mapsto (1,2,3,4)$ and apply the empirical integral transform  to standardize the data to Fr\'echet(2) margins as in \eqref{empstand}. {As in Section~\ref{s82}} we choose $k\approx\sqrt{n}$ as radial threshold (see Section \ref{regcase}), taking $k=100$ upper order statistics.
		The bivariate extremes (real data and simulated ones) are plotted in Figure~\ref{tb2plot} of Appendix~\ref{f1} with interpretations as in Section~\ref{s82}.
		
		We apply Algorithms~\ref{datdalg} and \ref{last.alg1} to reorder the nodes and estimate a recursive ML model.  In the two algorithms we set $a=1.01, \eps_1=0.002, \eps_2 =0.001$. This results in the causal order (VA, BC, LZ, AC). {Following the procedure in Section \ref{s82}, we obtain the ML coefficient matrix}
			\begin{align}
			\hat{A}=\begin{bmatrix}
			0.680 & 0.406 & 0.303 & 0.531\\
			 0& 0.651& 0.500 & 0.571\\
			0 & 0 & 0.960 & 0.281\\
			0& 0 & 0 & 1.000\\
			\end{bmatrix}.
					    \label{esAf}
			\end{align}
			The estimated scaling parameters of the components of $\bs X$ equals {up to three digits} (1,1,1,1). 
		The DAG corresponding to $\hat A$ is presented in Figure~\ref{founddag3}.
		We interpret dependence in high amounts of the four given food components.
	The estimated entries in $\hat A$ indicate the proportion of high intake from food consumption.
	
			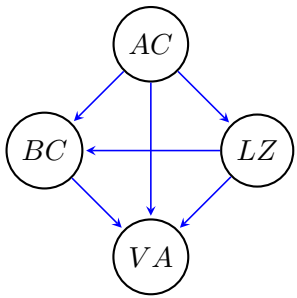
\begin{figure}[htbp]
				\centering
				\begin{tikzpicture}[
				> = stealth, 
				shorten > = 1pt, 
				auto,
				node distance = 2cm, 
				semithick 
				]
				\tikzstyle{every state}=[
				draw = black,
				thick,
				fill = white,
				minimum size = 4mm
				]
				\node[state,  ] (4) {$AC$};
				
				\node[state] (3) [below right of=4]{$LZ$};
				\node[state] (2) [below left of=4] {$BC$};
				\node[state] (1) [below left of=3] {$VA$};
				
				\path[->][blue] (4) edge node {} (2);
				\path[->][blue] (4) edge node {} (1);
				\path[->][blue] (3) edge node {} (2);
				
				\path[->][blue] (3) edge node {} (1);
				\path[->][blue] (4) edge node {} (3);

				\path[->][blue] (2) edge node {} (1);
				
				\end{tikzpicture}
				\small{\caption{DAG structure between the four food components.}\label{founddag3}}
			\end{figure}

	From the DAG in Figure~\ref{founddag3} we observe that the only initial node is Alpha-Carotene. 
	Having out-degree 3 its high {intake} affects the intake of the other three food components.
		Alpha-Carotene affects in particular Vitamin A, and Beta-Carotene, where in both cases it behaves as the main contributor of high values amongst the respective ancestral nodes. This can be seen from  the relative magnitude of the entries, namely $a_{14}>\max(a_{12},a_{13})$ and $a_{24}>a_{23}$. To a lesser degree  Alpha-Carotene also affects high intake of Lutein+Zeaxanthin as is seen from $a_{34}=0.281$.

		On the other hand, high proportions of Lutein+Zeaxanthin, which has in-degree 1 and out-degree 2,  lead to large intakes of Beta-Carotene, and affect those of Vitamin A in a similar fashion, but to a lesser proportion. From the estimated matrix $\hat A$ in \eqref{esAf}, we can infer that Lutein+Zeaxanthin, along with Alpha-Carotene is one of the main causes of high Beta-Carotene, with approximately equal contributions, {judging by the relative sizes of} $a_{23},a_{24}$. Finally Beta-Carotene, with in-degree 2 and out-degree 1 is the second largest contributor to high intake of Vitamin A, since $a_{12}>a_{13}$. Among all four components, Vitamin A has in-degree 3, and out-degree 0, showing that high intake of VA does not influence any of BC, LZ, or AC.

				\section{Conclusions}
			
			We have developed a new structure learning and estimation algorithm for the recursive ML model \eqref{rmlmequat}. The proposed methodology is designed for estimating recursive max-linear models for extreme events in a multivariate regular variation setting. The technique is non-parametric based on the empirically estimated spectral measure.
			The parametric estimation step focuses on the scalings and reflects the changes {caused by} simple scalar multiplications of the observed data variables {on} the scaling parameters. In addition, based on the very same scalings, we have shown how to estimate all extreme dependence parameters. The latter are shown to be asymptotically normal. Finally, the application of the new estimation method to financial and food dietary intake data shows that the recursive max-linear model can be fitted for capturing causal dependence structures in the extremes arising from real-life data.


\appendix

	\vspace{-3cm}\section{Figures: Portfolio Data}\label{p}
		\begin{figure}[H]
		\vspace{-.4cm}
		\hspace{-2cm}\includegraphics[height=6cm, width=10.4cm]{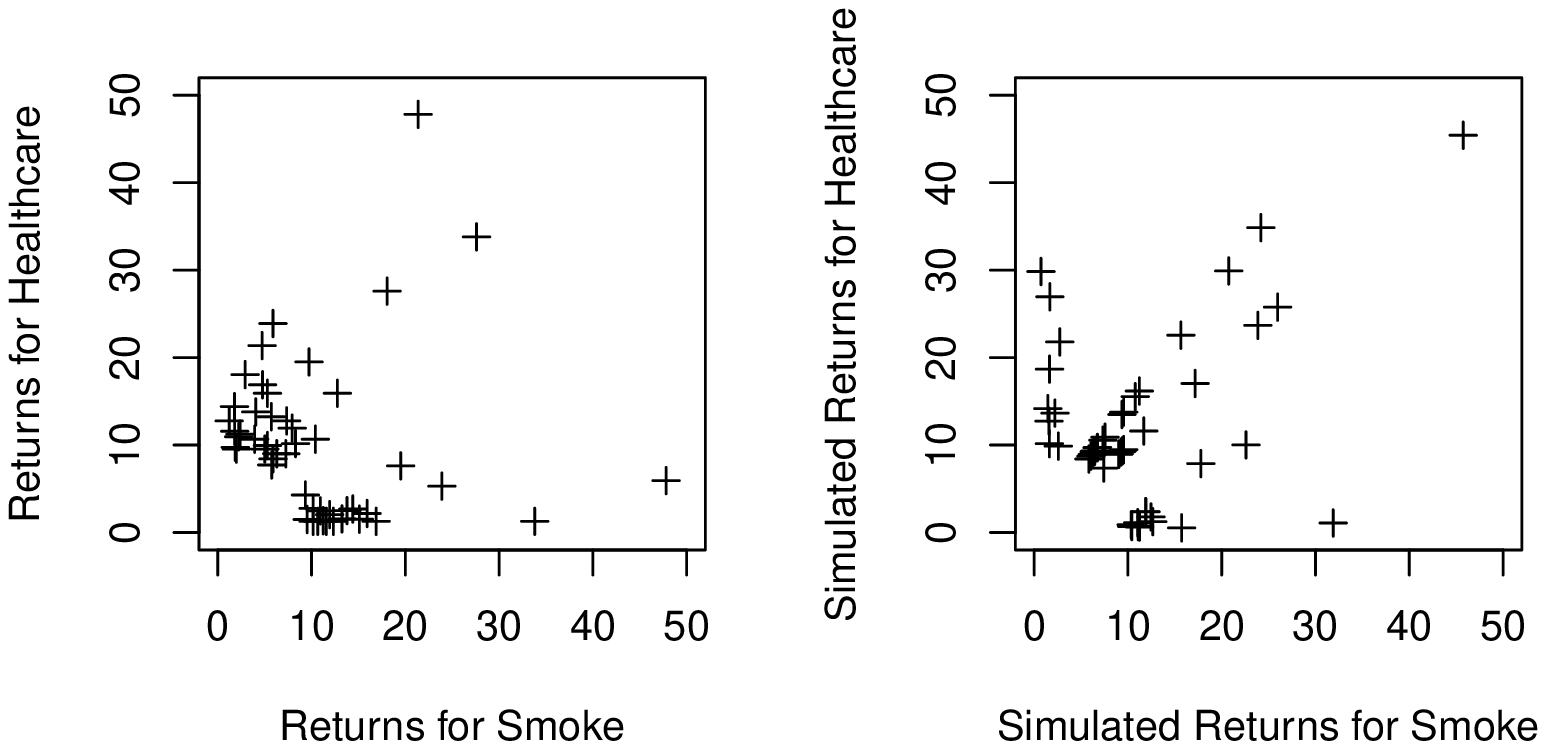}\vspace{-.0cm}
		\includegraphics[height=6cm, width=10.4cm]{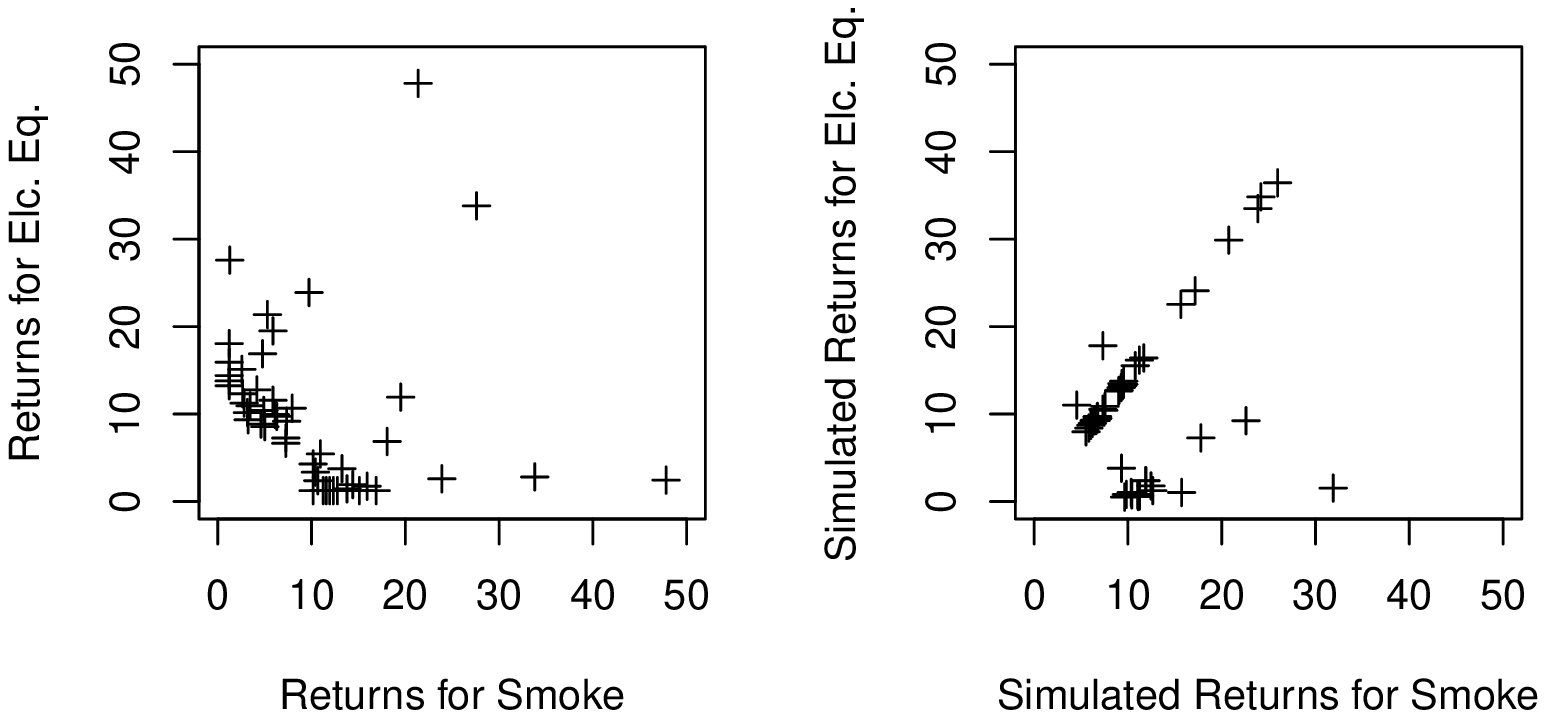}\vspace{-.0cm}
			
		\hspace{-2cm}\includegraphics[height=6cm, width=10.4cm]{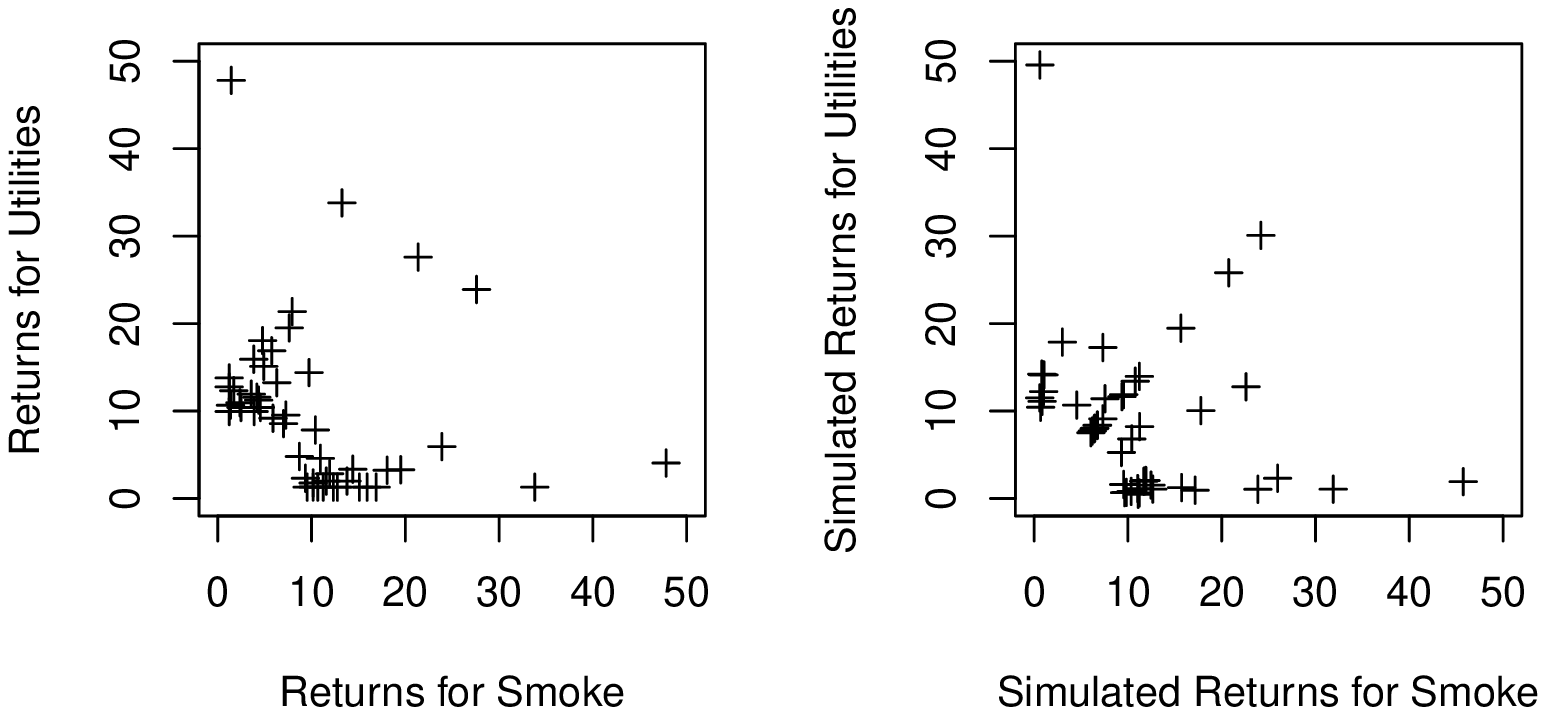}\vspace{-.0cm}
		\includegraphics[height=6cm, width=10.4cm]{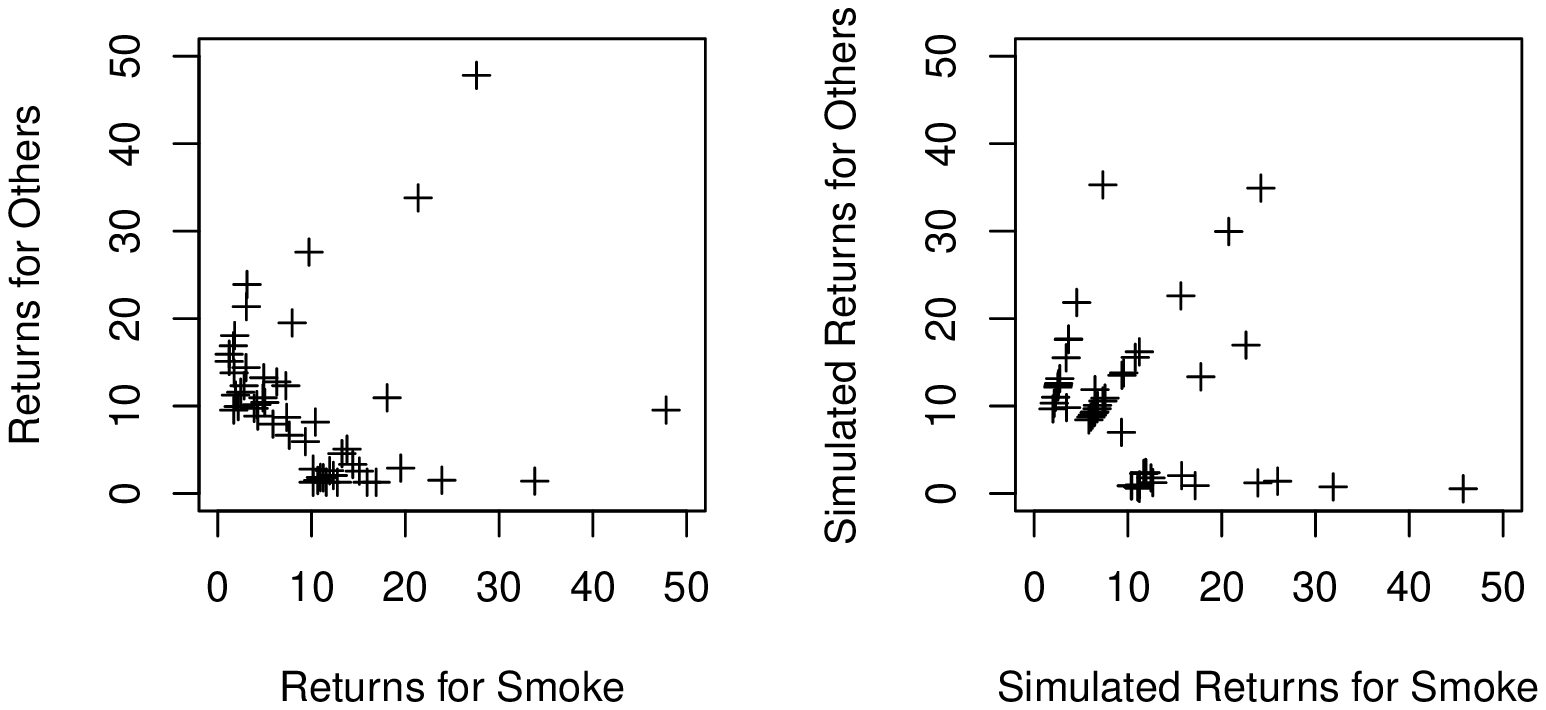}\vspace{-.0cm}
					
		\hspace{-2cm}\includegraphics[height=6cm, width=10.4cm]{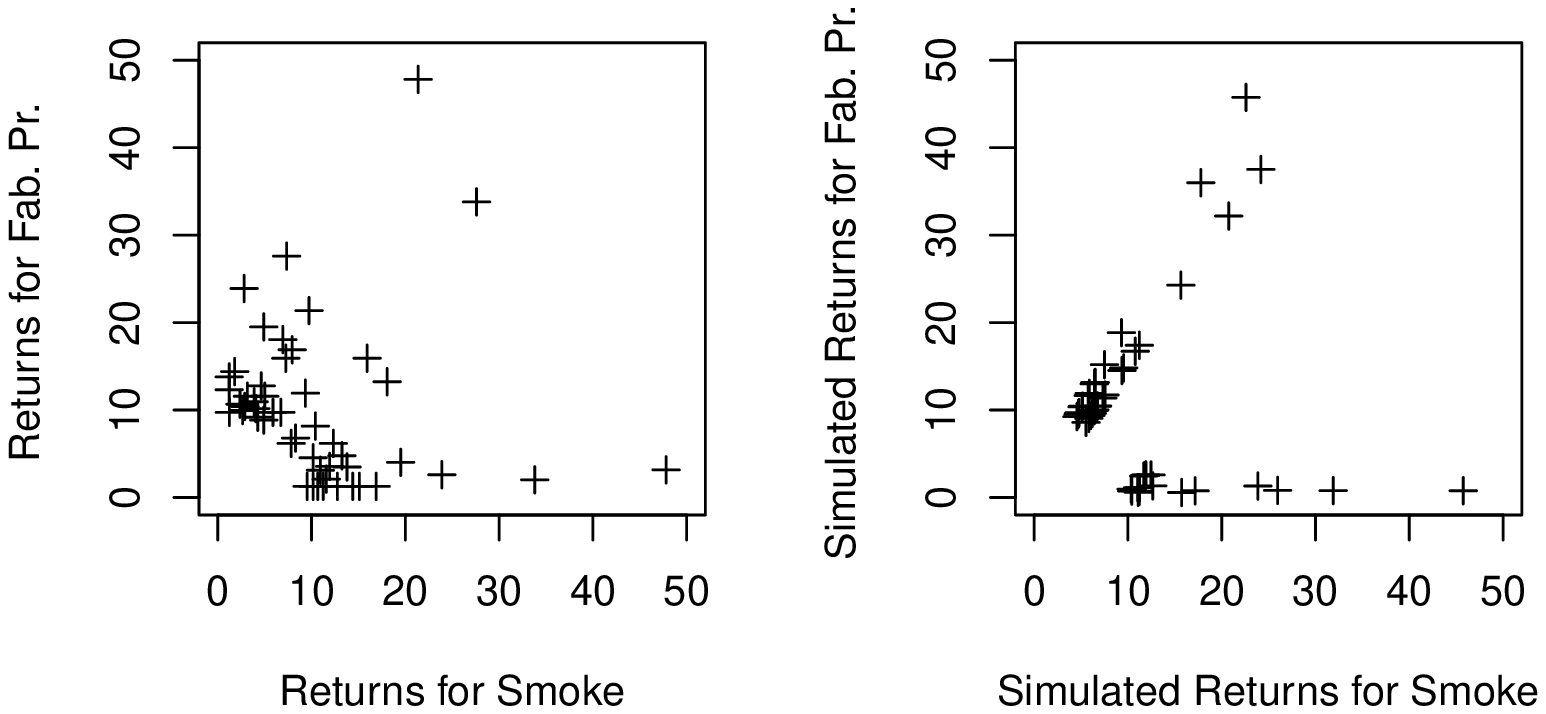}	
		\includegraphics[height=6cm, width=10.4cm]{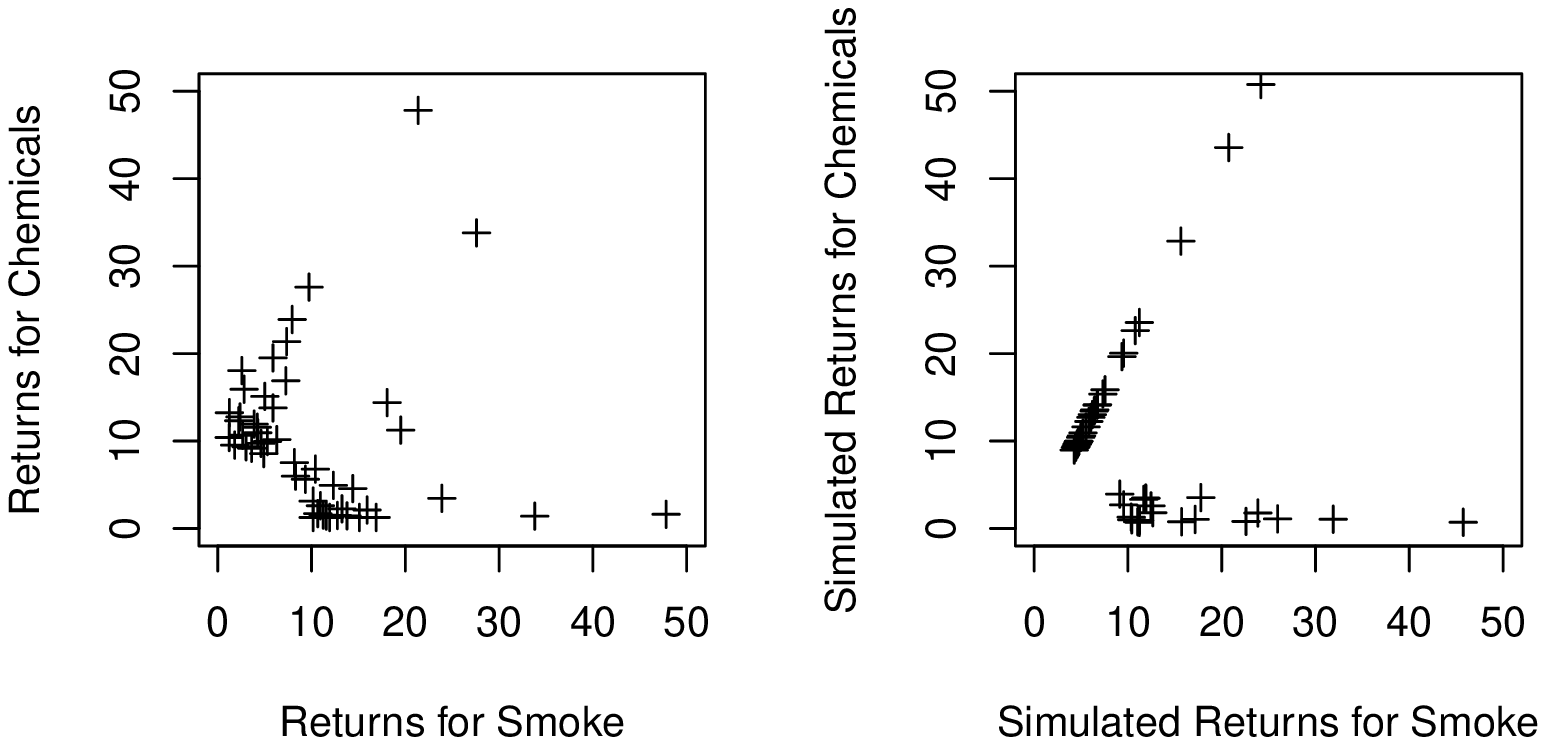}\vspace{-.0cm}
		
		\hspace{-2cm}\includegraphics[height=6cm, width=10.4cm]{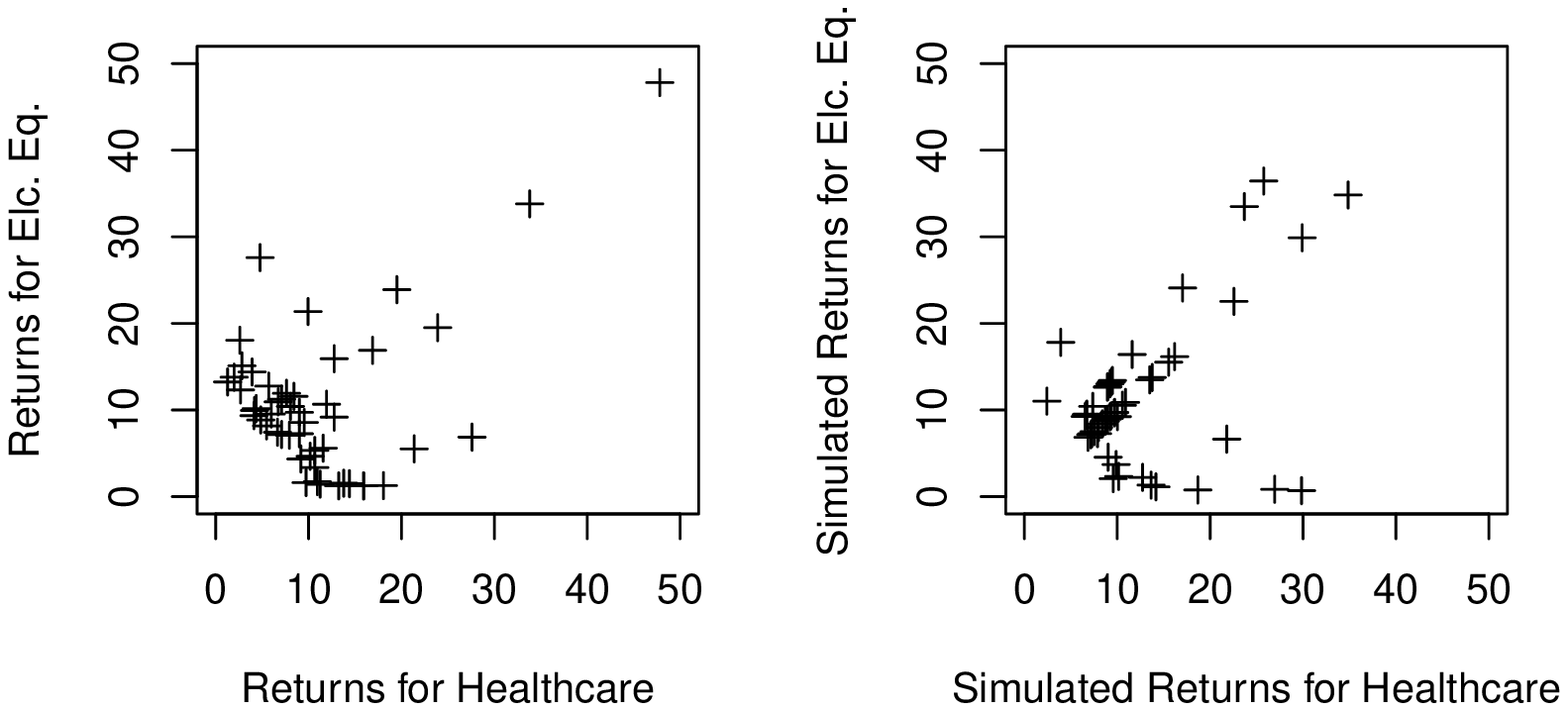}\vspace{-.0cm}
		\includegraphics[height=6cm, width=10.4cm]{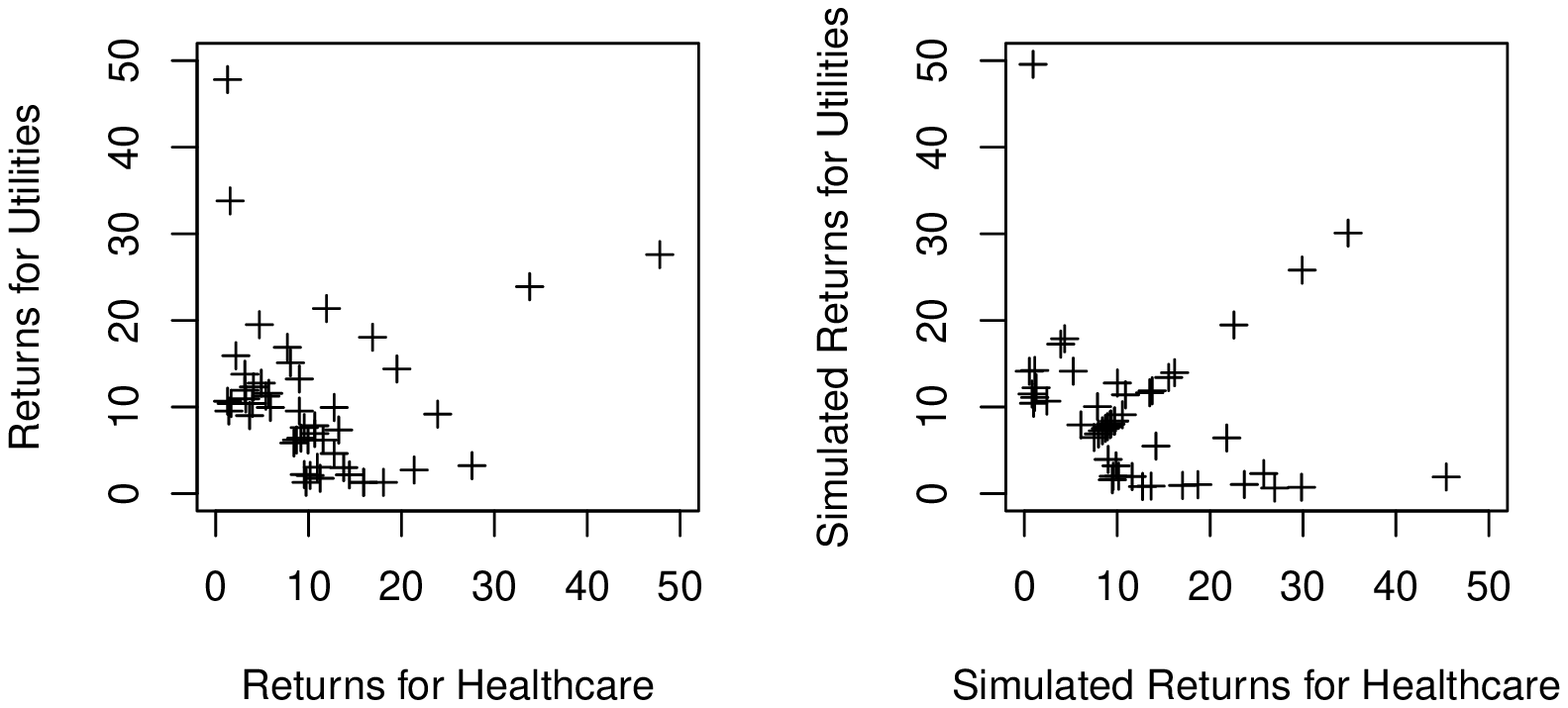}\vspace{-.0cm}
			\end{figure}
			
			\begin{figure}[H]
		
		\hspace{-2cm}\includegraphics[height=6cm, width=10.4cm]{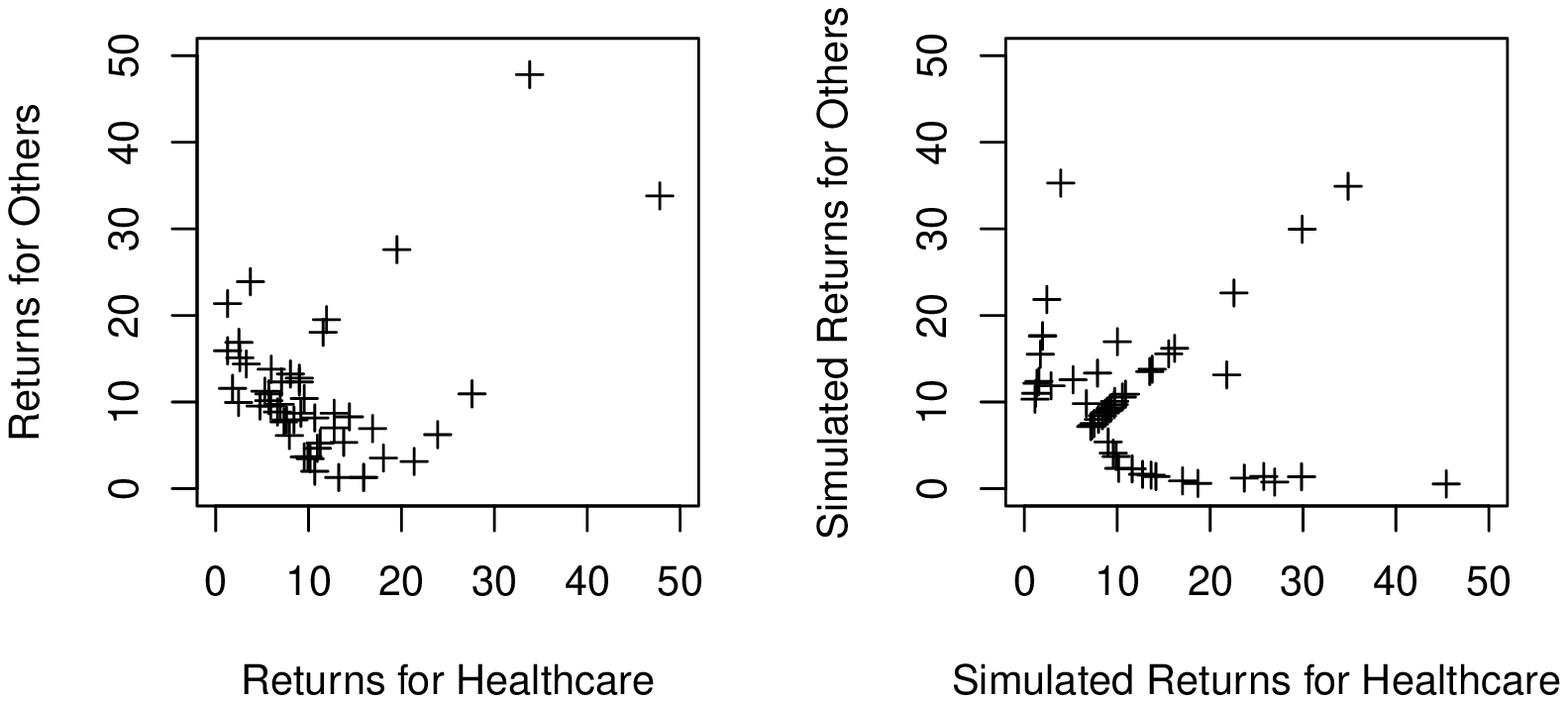}\vspace{-.0cm}
		\includegraphics[height=6cm, width=10.4cm]{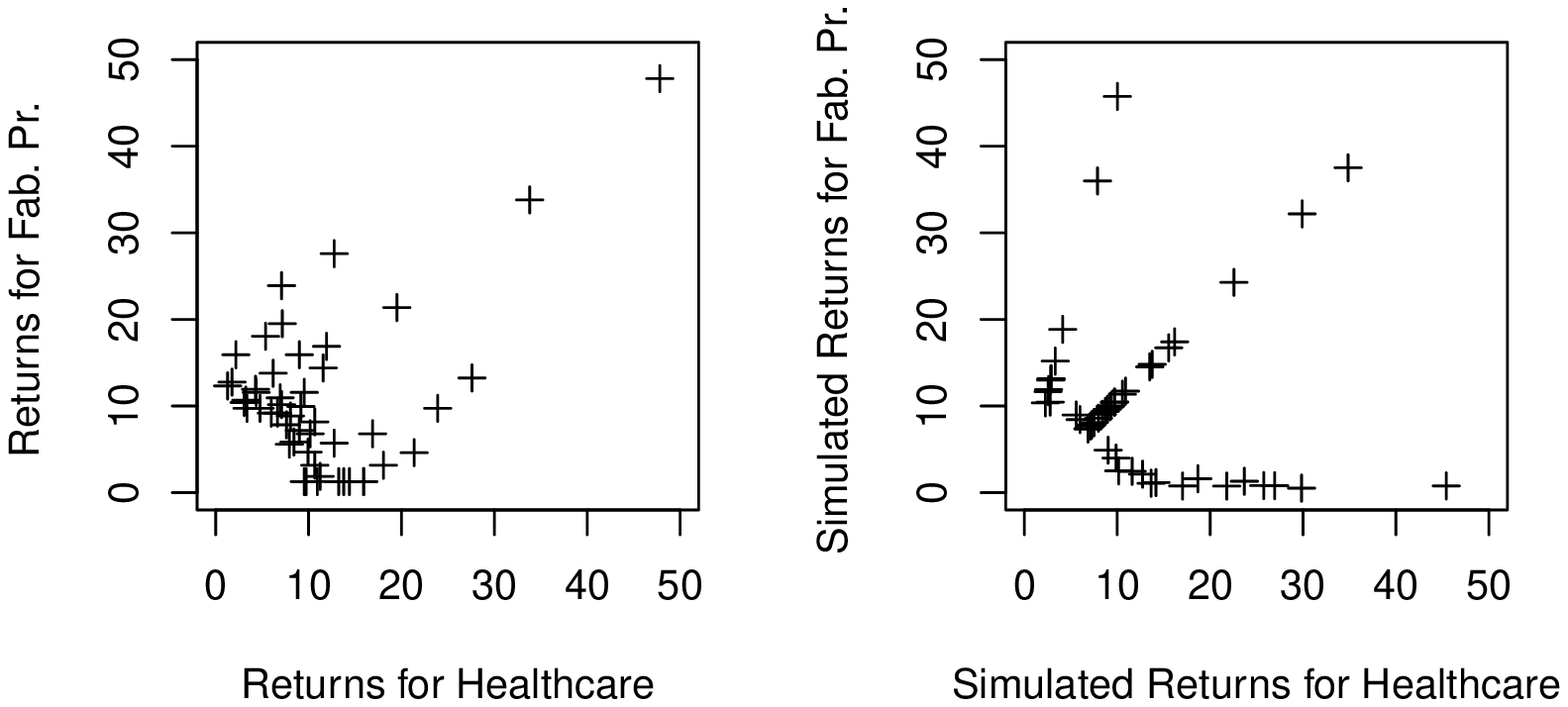}\vspace{-.0cm}
				
		\hspace{-2cm}\includegraphics[height=6cm, width=10.4cm]{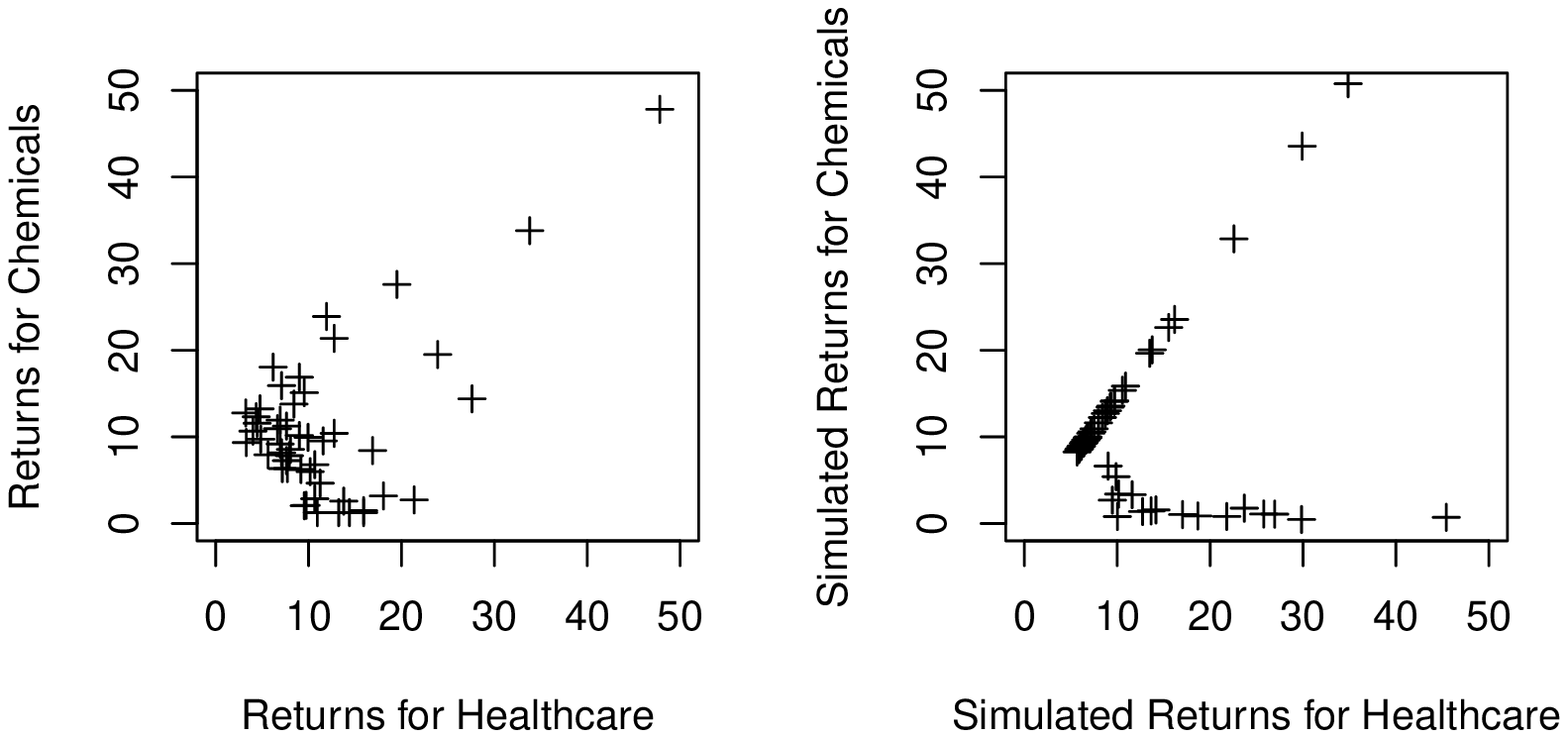}\vspace{-.0cm}
		\includegraphics[height=6cm, width=10.4cm]{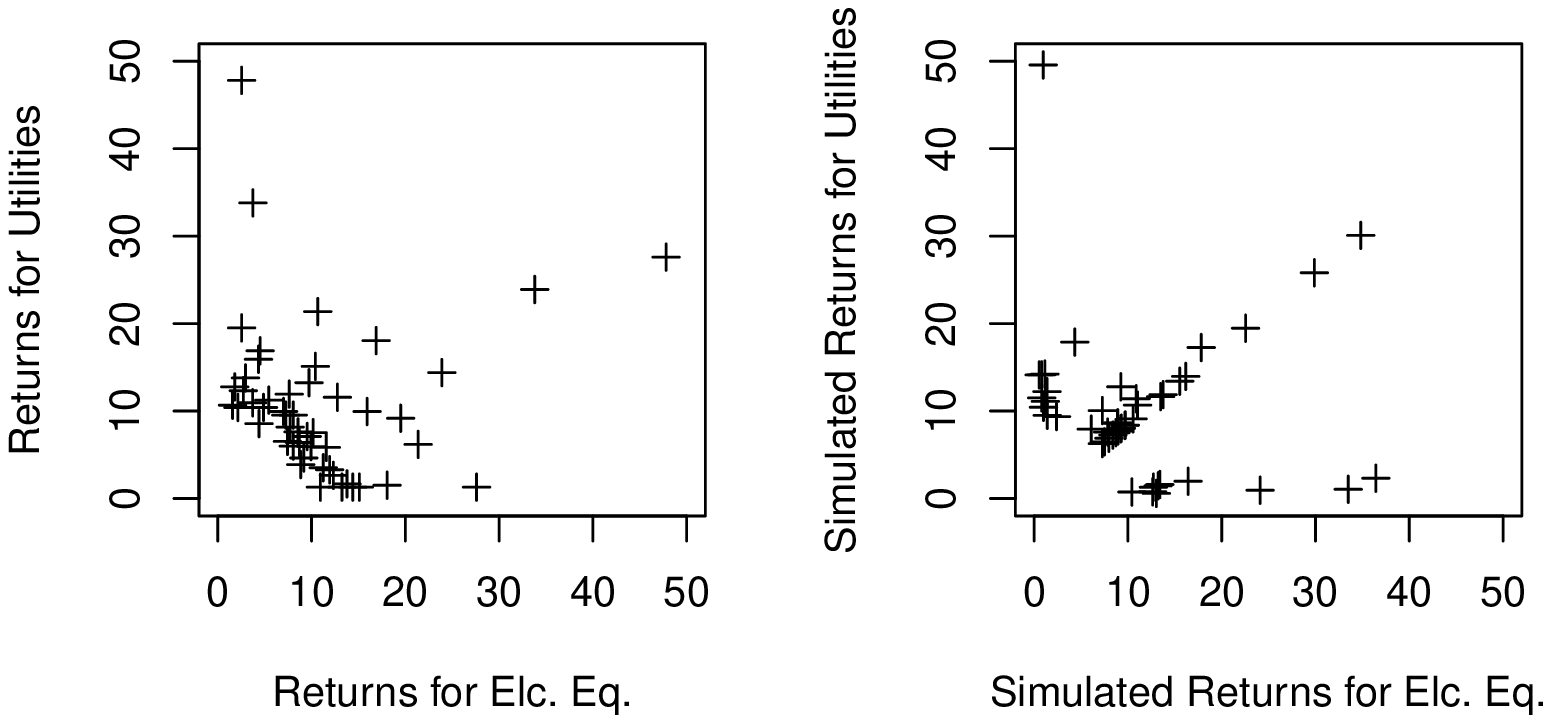}\vspace{-.0cm}
						
		\hspace{-2cm}	\includegraphics[height=6cm, width=10.4cm]{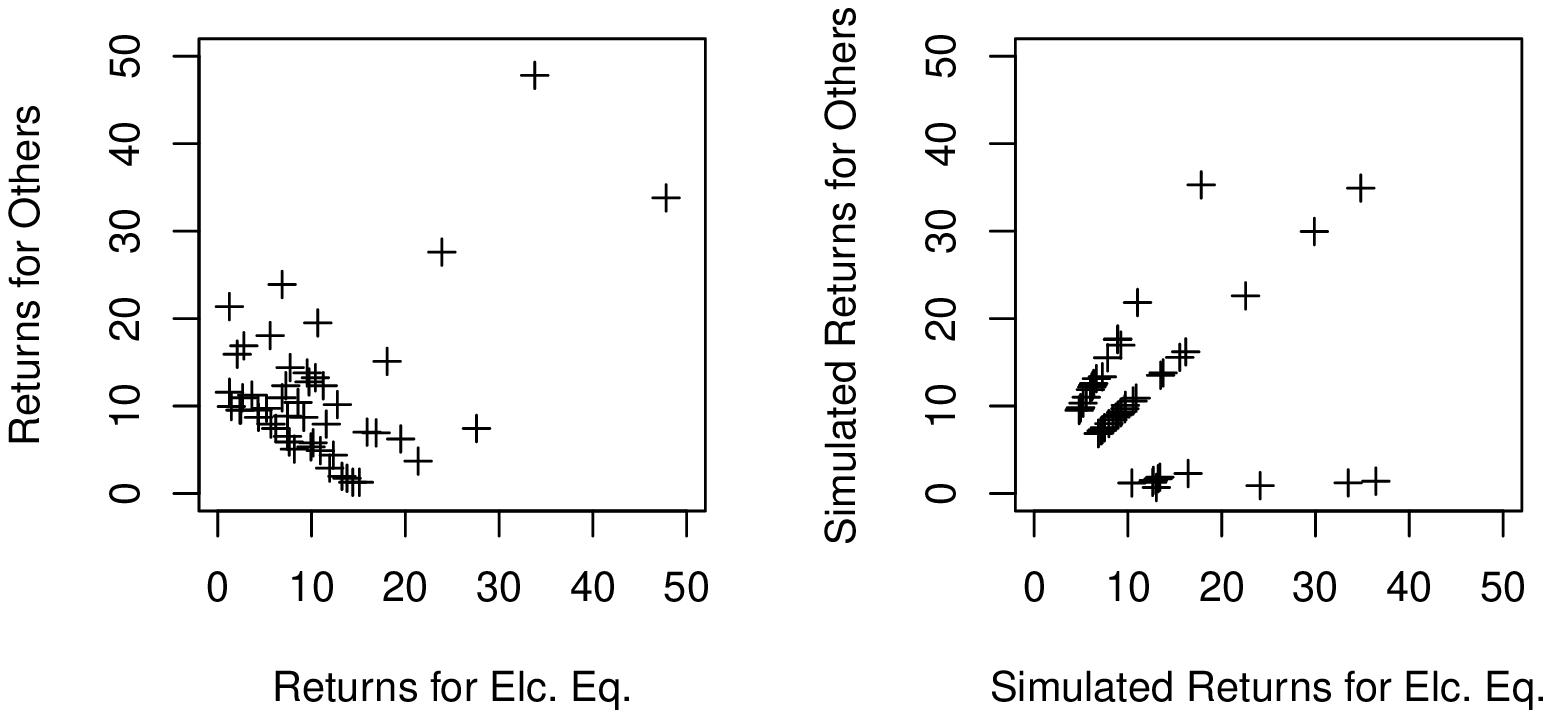}\vspace{-.0cm}
		\includegraphics[height=6cm, width=10.4cm]{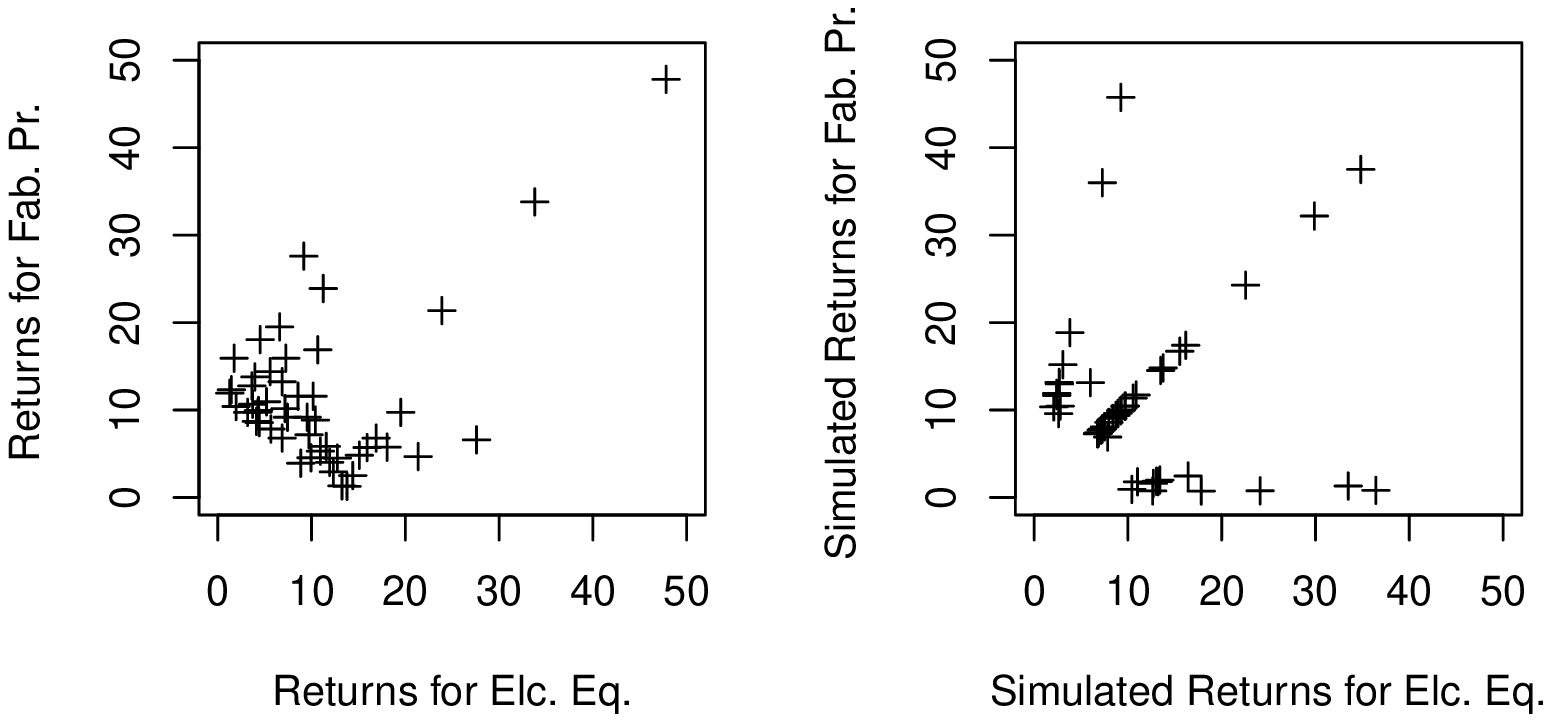}\vspace{-.0cm}
		
		\hspace{-2cm}\includegraphics[height=6cm, width=10.4cm]{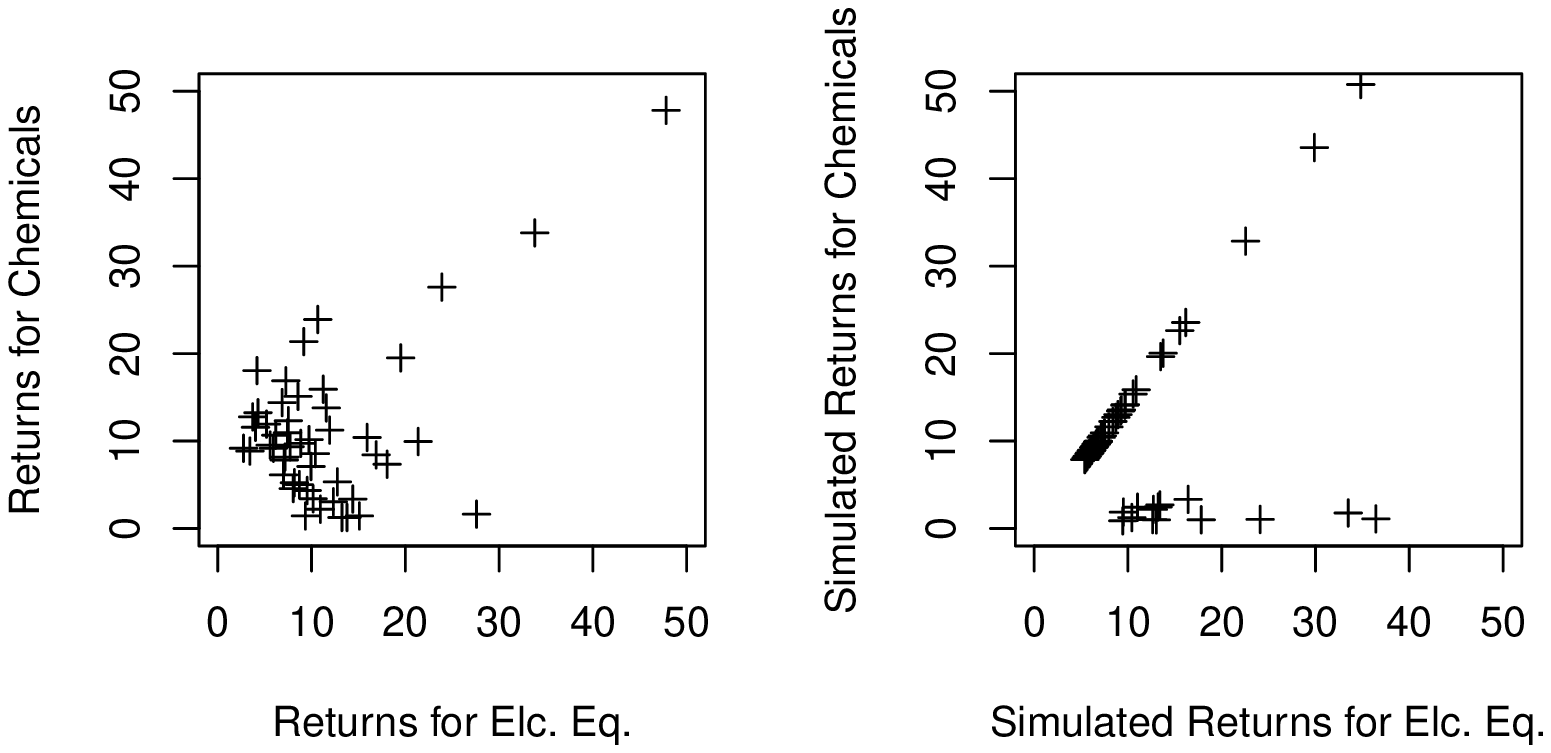}\vspace{-.0cm}
		\includegraphics[height=6cm, width=10.4cm]{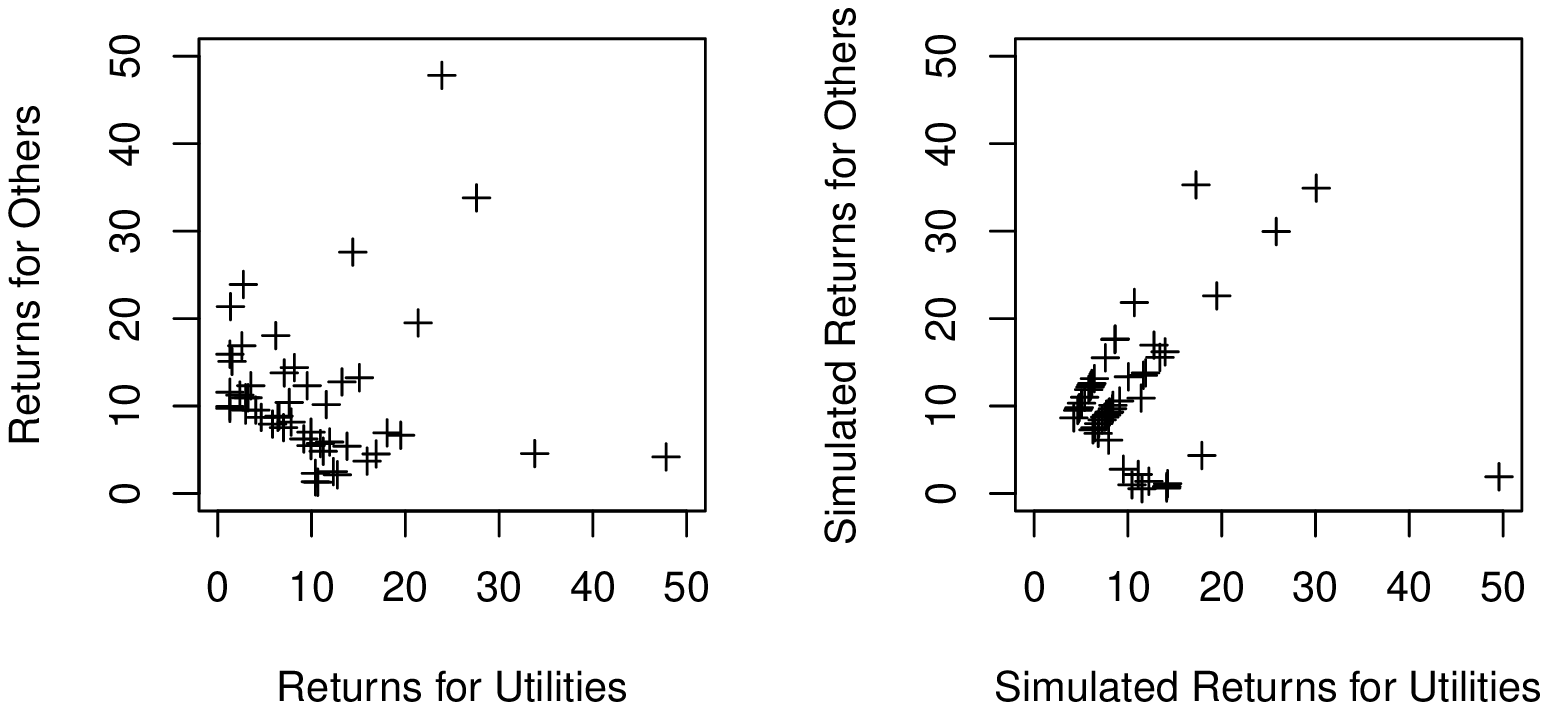}\vspace{-.0cm}
						
			
		\end{figure}
		\begin{figure}[H]
		\hspace{-2cm}\includegraphics[height=6cm, width=10.4cm]{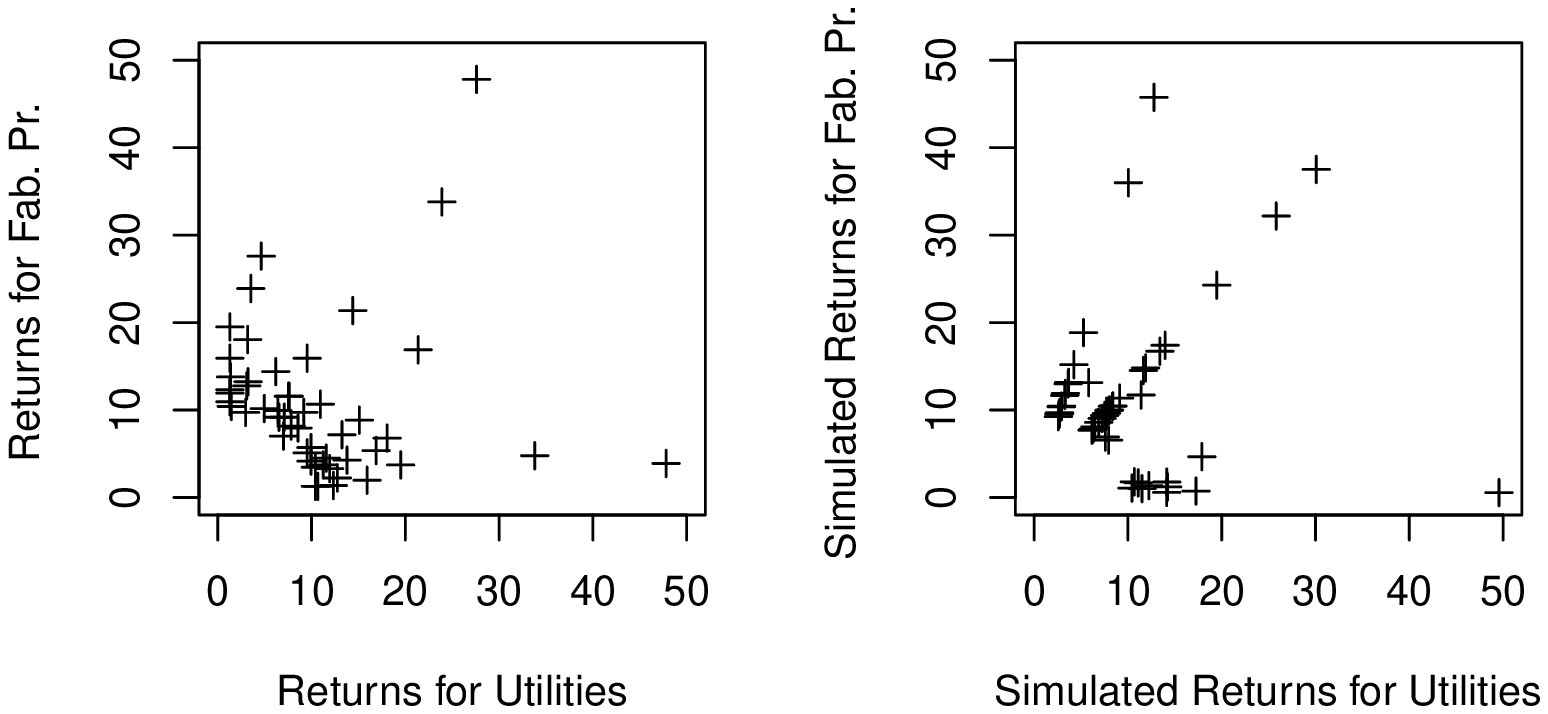}\vspace{-.0cm}
		\includegraphics[height=6cm, width=10.4cm]{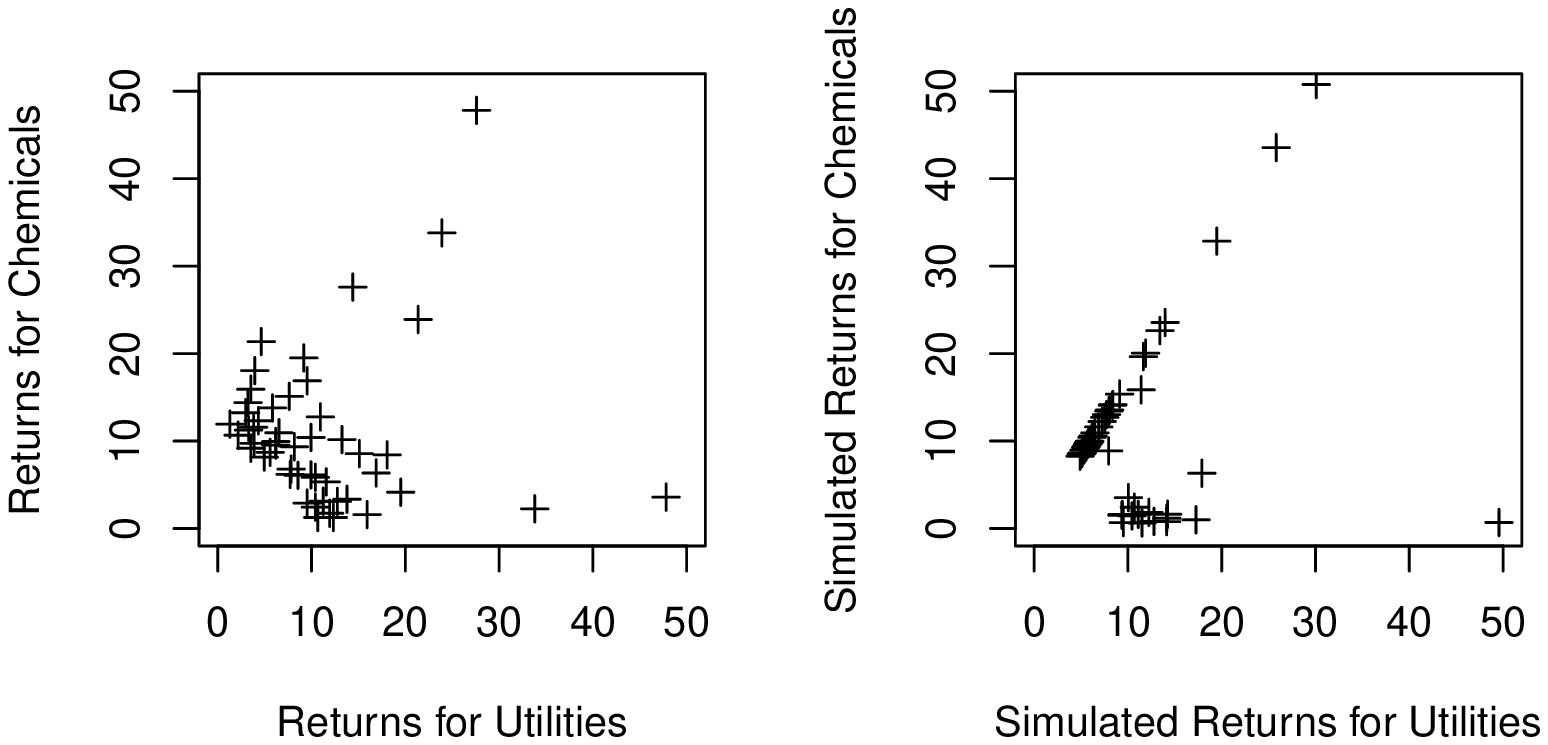}
        
        \hspace{-2cm}\includegraphics[height=6cm, width=10.4cm]{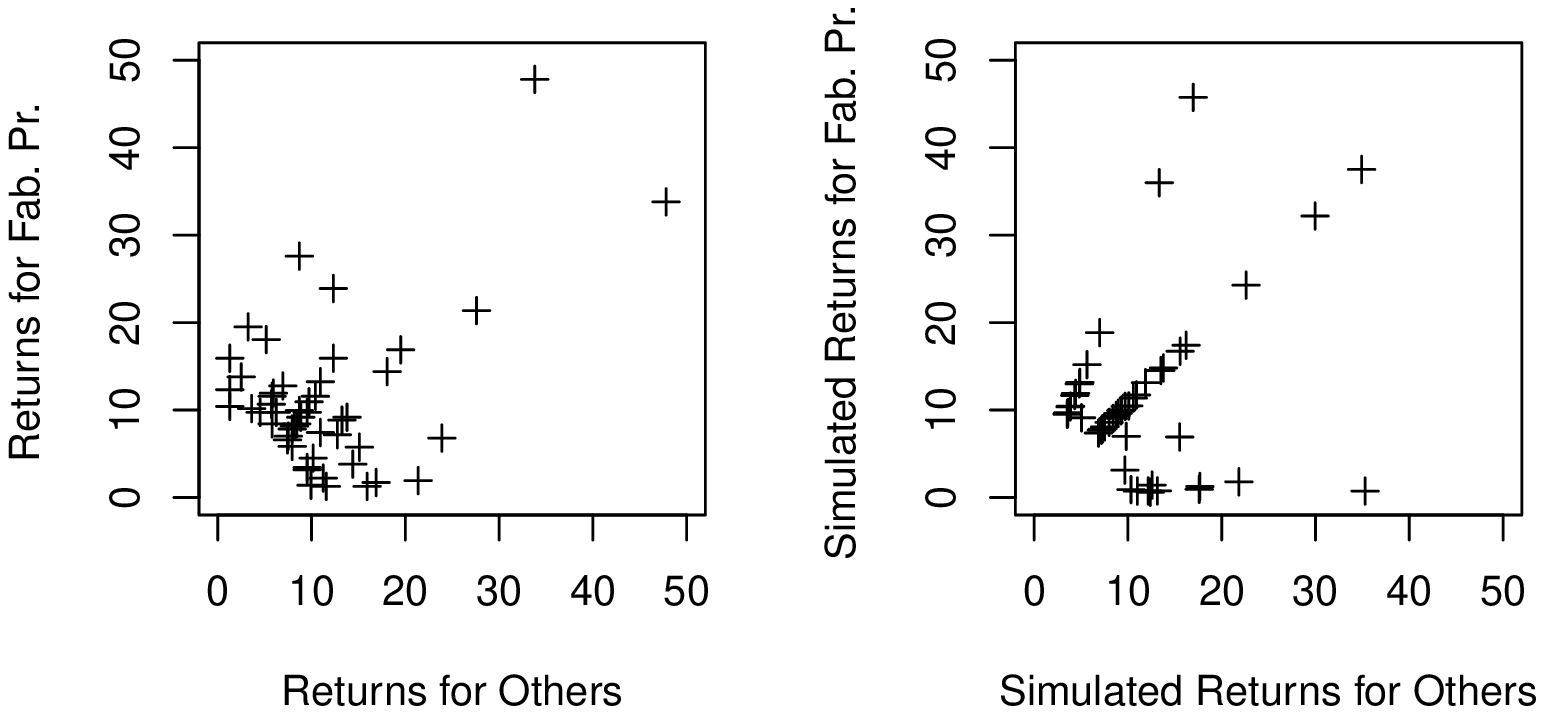}
        \includegraphics[height=6cm, width=10.4cm]{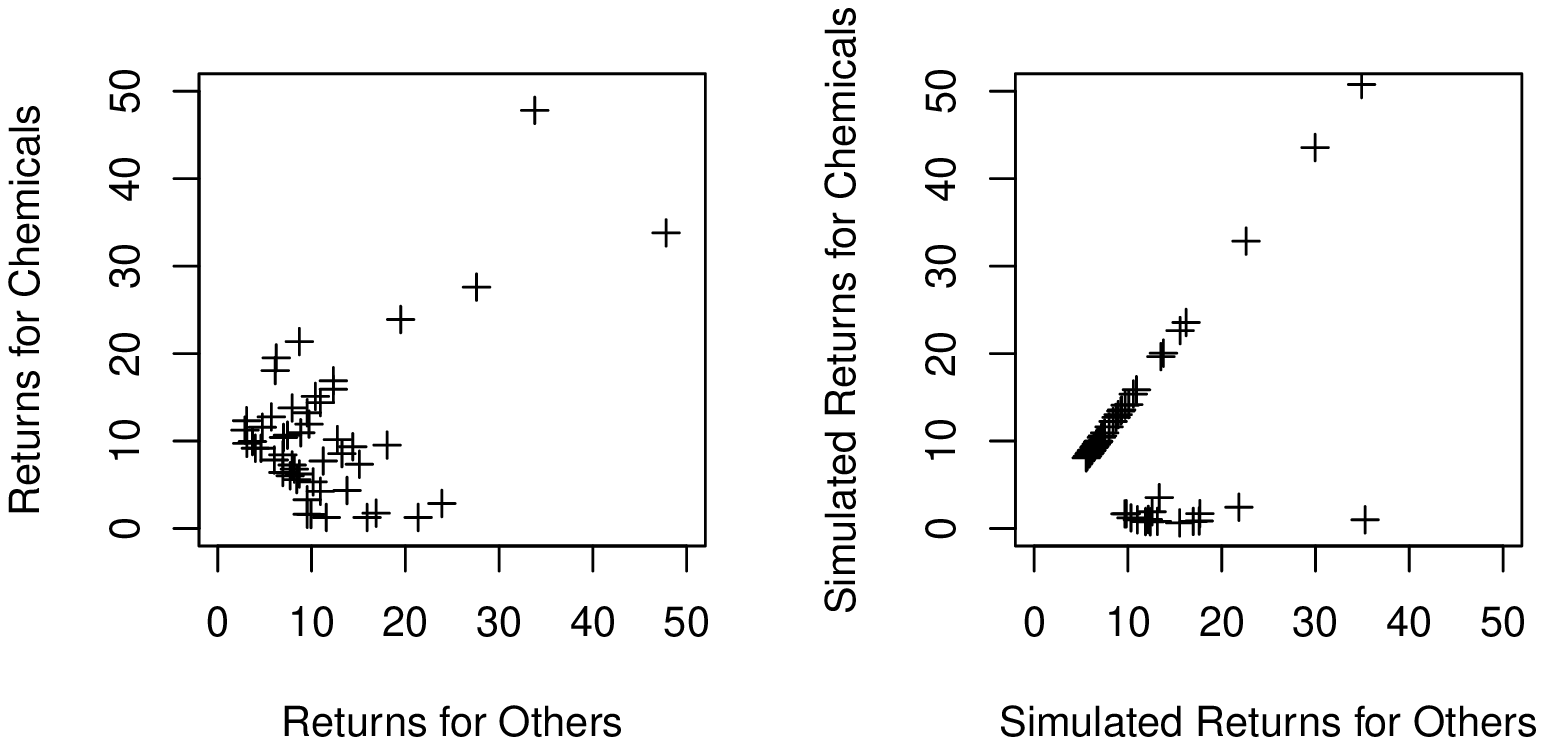}
        
        \hspace{3.25cm}\includegraphics[height=6cm, width=10.4cm]{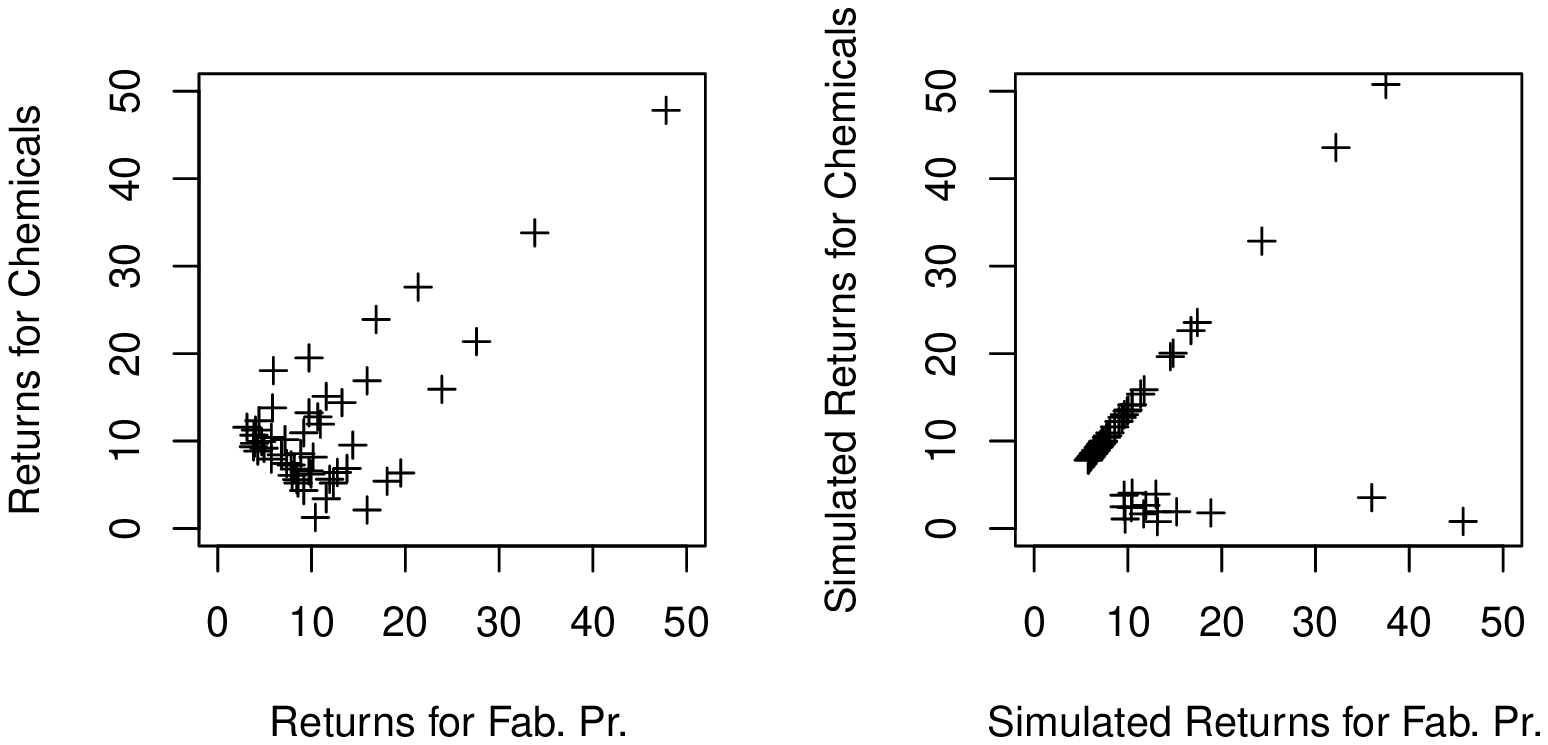}

		\tiny{\caption{Bivariate extremes above a high radial threshold for the DAG in Figure \ref{founddag}: The left plots contain the tails from the real data, the right ones contain simulated realizations.}\label{tb1plot}}
		\end{figure}
	
	\section{Figures: Food Components}\label{f1}
				\begin{figure}[H]
	                \hspace{-2cm}\includegraphics[height=6cm, width=10.4cm]{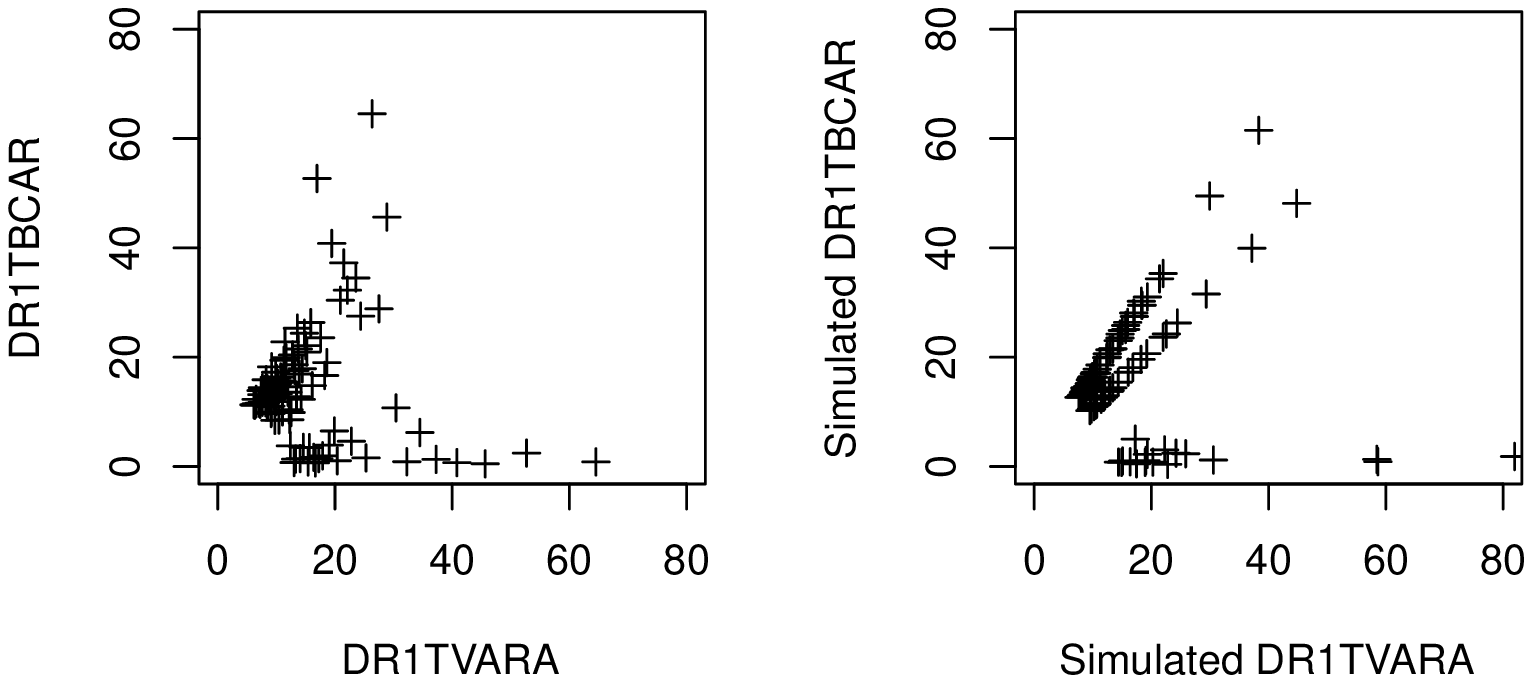}\vspace{-.0cm}
					\includegraphics[height=6cm, width=10.4cm]{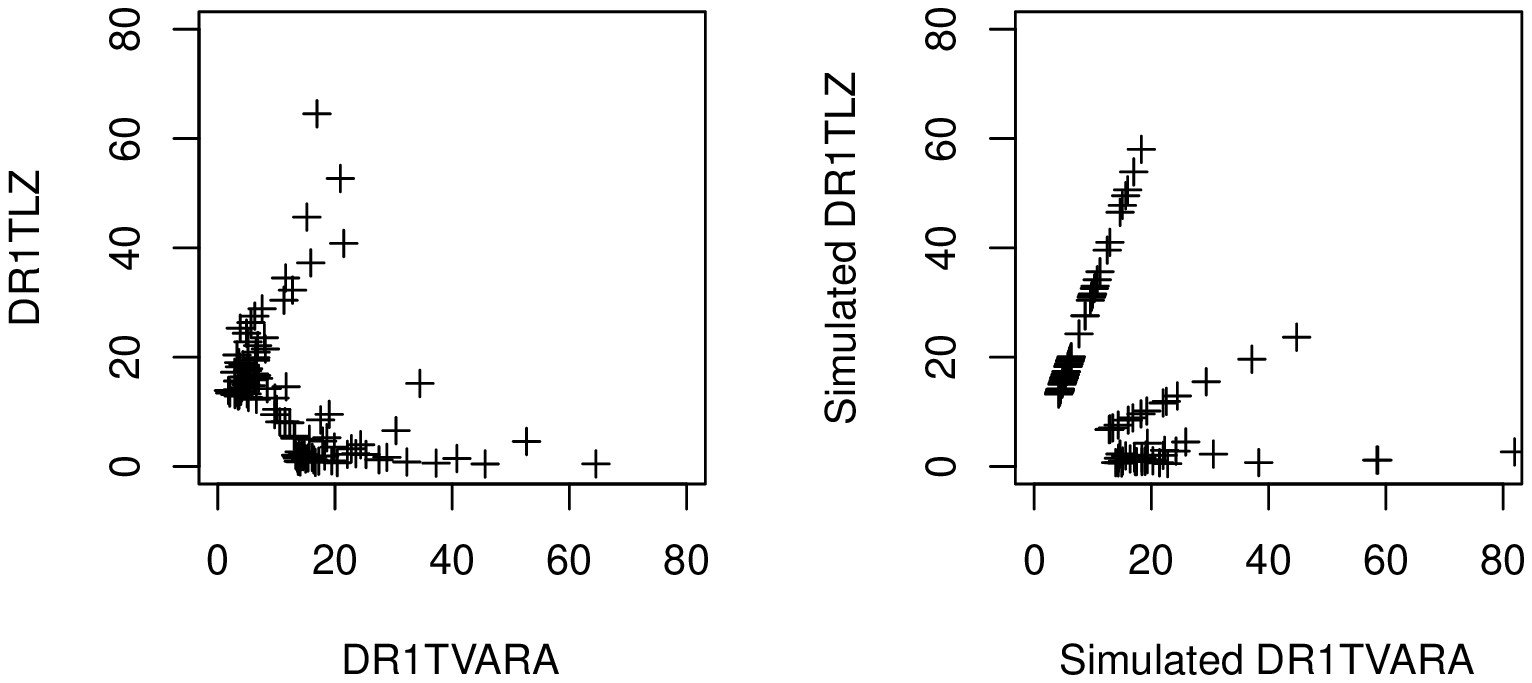}\vspace{-.0cm}
					
					\hspace{-2cm}\includegraphics[height=6cm, width=10.4cm]{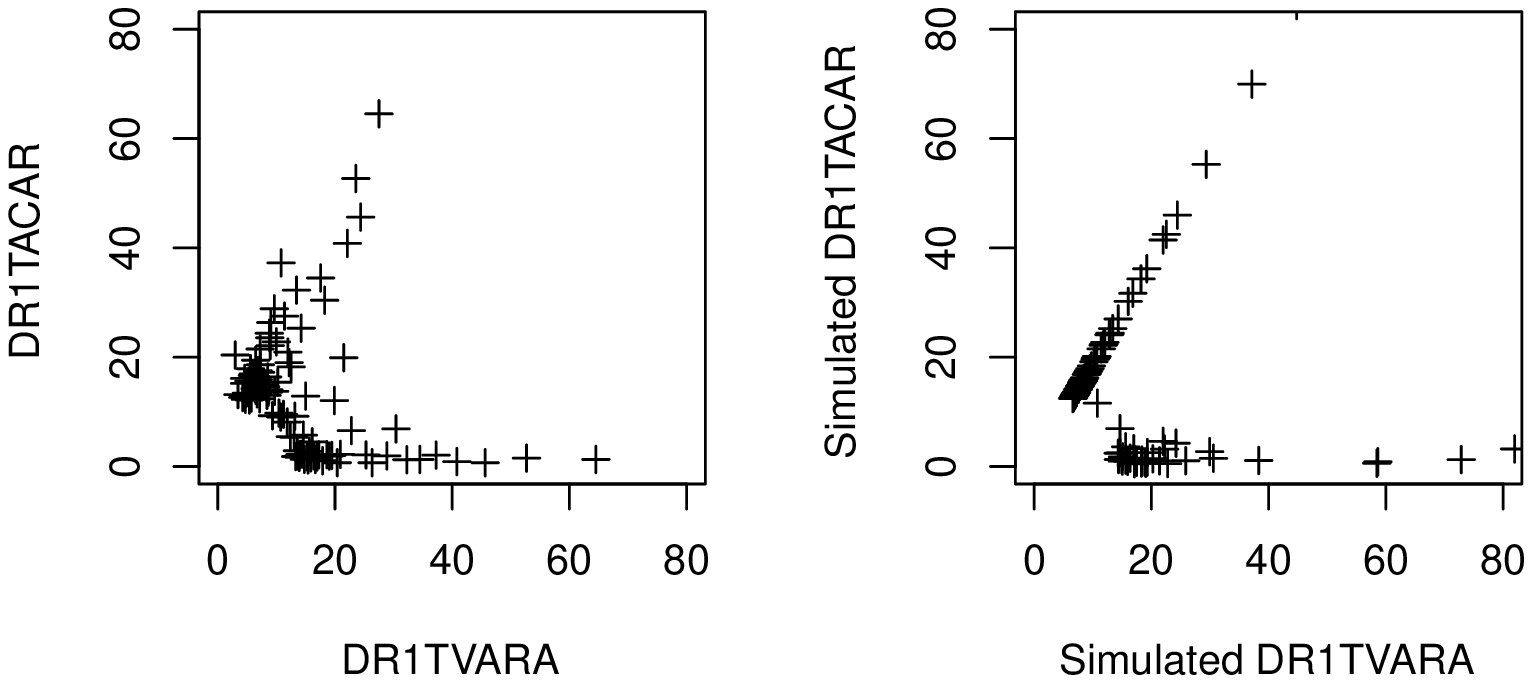}\vspace{-.0cm}
					\includegraphics[height=6cm, width=10.4cm]{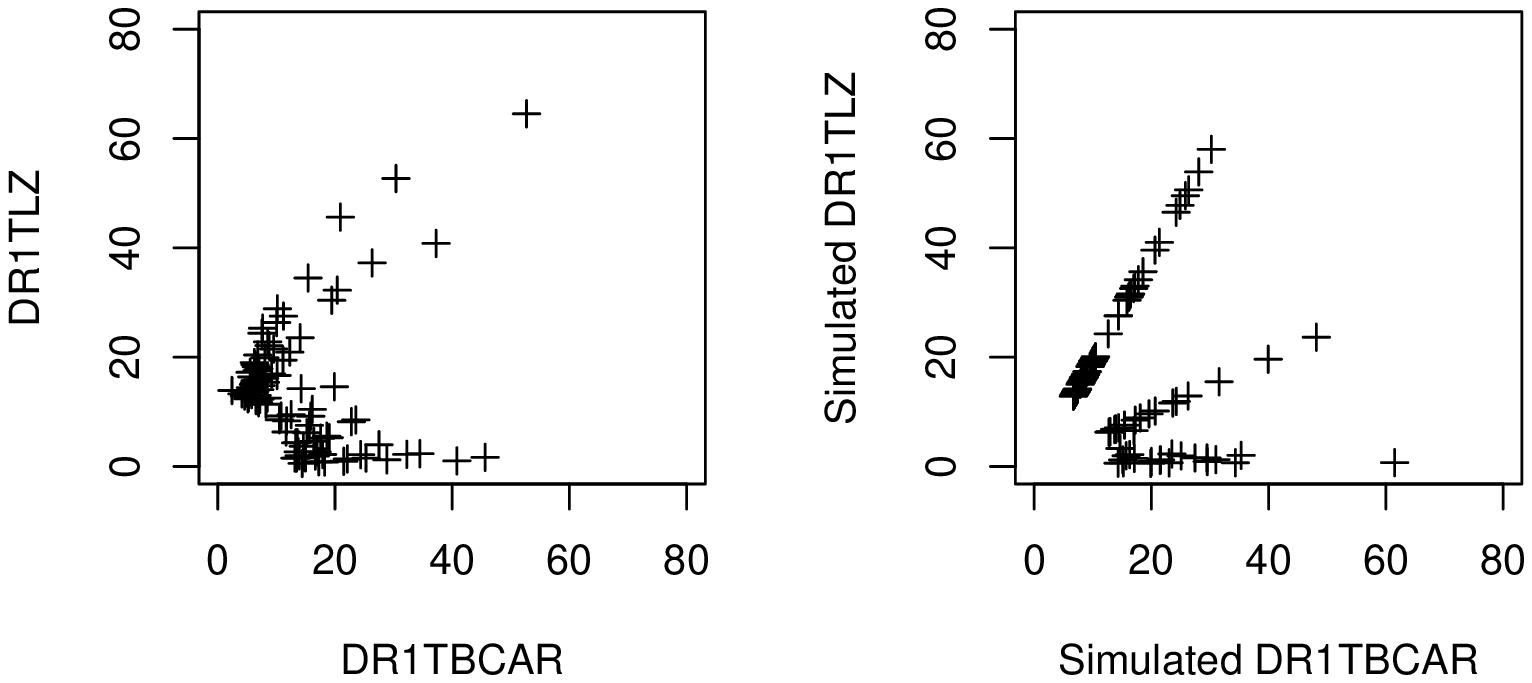}\vspace{-.0cm}
					
					\hspace{-2cm}\includegraphics[height=6cm, width=10.4cm]{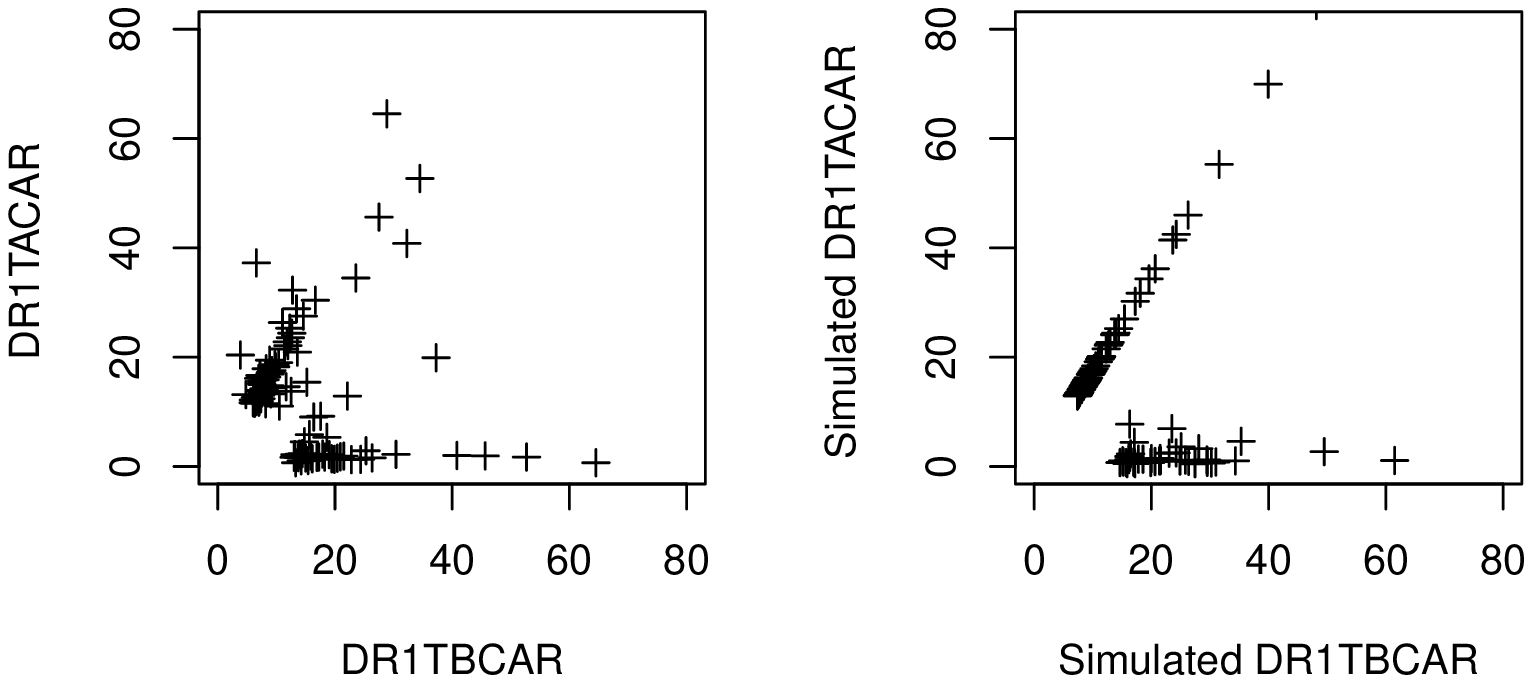}\vspace{-.0cm}
					\includegraphics[height=6cm, width=10.4cm]{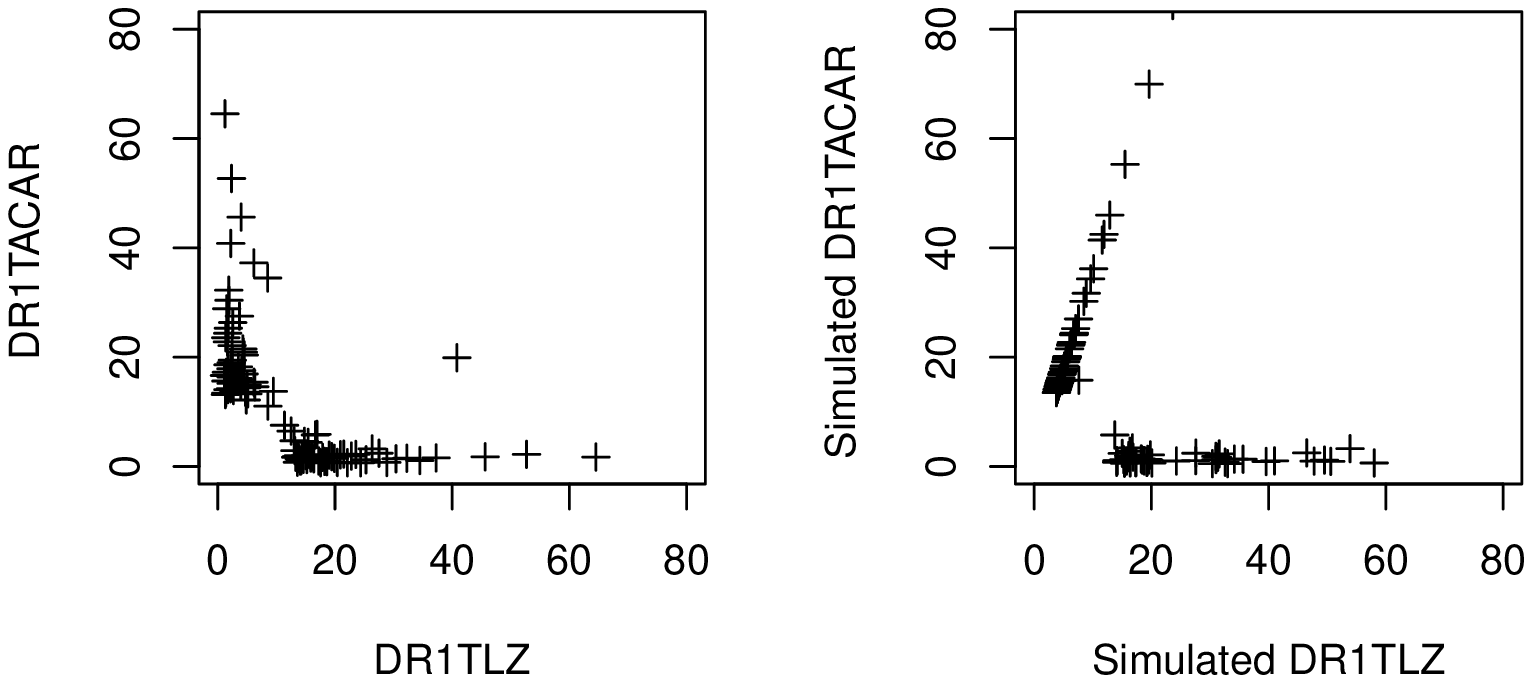}\vspace{-.0cm}

					\tiny{\caption{Bivariate extremes above a high radial threshold for the DAG in Figure \ref{founddag3}: The left plots contain the tails from the real data, the right ones contain simulated realizations.}\label{tb2plot}}
				\end{figure}

\end{document}